\documentclass[11pt,reqno]{amsart}  

\usepackage{amscd} 
\usepackage{mathrsfs}
\usepackage{amsmath,amsbsy,amsthm,amsfonts,amssymb,esint}
\usepackage[unicode,psdextra]{hyperref}
\usepackage{graphicx}
\usepackage{tikz}
\usepackage{bm}
\usepackage{mathabx}
\usepackage{mathtools}
\usepackage{soul}
\usepackage{color}

\usepackage{slashed}

\usepackage[sort&compress,capitalise,nameinlink]{cleveref}
\crefname{section}{\textsection}{\textsection}
\crefname{subsection}{\textsection}{\textsection}
\crefname{appendix}{\textsection}{\textsection}
\usepackage{dsfont}
\usepackage{enumerate}

\DeclareMathAlphabet\mathbfcal{OMS}{cmsy}{b}{n}

\newtheorem{theorem}{Theorem}[section]
\newtheorem{lemma}[theorem]{Lemma}
\newtheorem{proposition}[theorem]{Proposition}
\newtheorem{corollary}[theorem]{Corollary} 
\theoremstyle{definition}
\newtheorem{definition}{Definition}[section]
\newtheorem{example}[theorem]{Example}
\newtheorem{remark}[theorem]{Remark} 
\numberwithin{equation}{section}

\newcommand{\bb}{\beta}
\newcommand{\one}{\mathds{1}}

\newcommand{\cv}{\mathbf{c}}

\newcommand{\xv}{\mathbf{x}}
\newcommand{\xvp}{\mathbf{x}^{\prime}}
\newcommand{\kv}{\mathbf{k}}

\newcommand{\qv}{\mathbf{q}}

\newcommand{\psiv}{\bm{\psi}}
\newcommand{\phiv}{\bm{\phi}}
\newcommand{\varphiv}{\bm{\varphi}}
\newcommand{\fv}{\bm{f}}
\newcommand{\giv}{\bm{G}}
\newcommand{\gmv}{\mathcal{G}}

\newcommand{\dom}{\mathscr{D}}
\newcommand{\diff}{\mathrm{d}}

\newcommand{\xiv}{\bm{\xi}}

\newcommand{\aav}{\mathbf{A}}

\newcommand{\avg}[1]{\lf\langle #1 \ri\rangle}
\newcommand{\teta}{\vartheta}

\newcommand{\beq}{\begin{equation}}
\newcommand{\eeq}{\end{equation}}

\newcommand{\OO}{\mathcal{O}}

\newcommand{\UU}{\mathcal{U}}
\newcommand{\WW}{\mathcal{W}}
\newcommand{\VV}{\mathcal{V}}

\newcommand{\Z}{\mathbb{Z}}
\newcommand{\CC}{\mathbb{C}}
\newcommand{\R}{\mathbb{R}}

\newcommand{\bdm}{\begin{displaymath}}
\newcommand{\edm}{\end{displaymath}}
\newcommand{\bdn}{\begin{eqnarray}}
\newcommand{\edn}{\end{eqnarray}}
\newcommand{\bay}{\begin{array}{c}}
\newcommand{\eay}{\end{array}}
\newcommand{\ben}{\begin{enumerate}}
\newcommand{\een}{\end{enumerate}}
\newcommand{\beqn}{\begin{eqnarray}}
\newcommand{\eeqn}{\end{eqnarray}}
\newcommand{\bml}[1]{\begin{multline} #1 \end{multline}}
\newcommand{\bmln}[1]{\begin{multline*} #1 \end{multline*}}

\newcommand{\be}{\begin{equation}}%
\newcommand{\ee}{\end{equation}}

\newcommand{\lf}{\left}
\newcommand{\ri}{\right}
\newcommand{\disp}{\displaystyle}
\newcommand{\tx}{\textstyle}

\newcommand{\braket}[2]{\lf\langle #1|#2 \ri\rangle}

\newcommand{\braketl}[2]{\lf.\lf\langle #1\ri|#2 \ri\rangle}

\newcommand{\meanlrlr}[3]{\lf\langle #1\lf|#2\ri|#3\ri\rangle}

\renewcommand{\leq}{\leqslant}
\renewcommand{\geq}{\geqslant}

\newcommand{\hpauli}{H_{\mathrm{P}}}
\newcommand{\hpaulif}{H^{(\mathrm{F})}_{\mathrm{P}}}
\newcommand{\hsch}{H_{\mathrm{S}}}
\newcommand{\hschf}{H^{(\mathrm{F})}_{\mathrm{S}}}

\DeclareMathOperator*{\slim}{s\,--\,lim}

\title[Pauli Hamiltonians with an Aharonov-Bohm flux]{Pauli Hamiltonians with an Aharonov-Bohm flux}

\author{William Borrelli}
\address{Dipartimento di Matematica, Politecnico di Milano, P.zza Leonardo da Vinci, 32, 20133, Milano, Italy}
\email{william.borrelli@polimi.it}
\urladdr{}

\author{Michele Correggi}
\address{Dipartimento di Matematica, Politecnico di Milano, P.zza Leonardo da Vinci, 32, 20133, Milano, Italy}
\email{michele.correggi@gmail.com}
\urladdr{https://sites.google.com/view/michele-correggi}

\author{Davide Fermi}
\address{Dipartimento di Matematica, Politecnico di Milano, P.zza Leonardo da Vinci, 32, 20133, Milano, Italy
and Istituto Nazionale di Fisica Nucleare, Sezione di Milano, Italy}
\email{davide.fermi@polimi.it}
\urladdr{https://fermidavide.com}

\begin{document}

\begin{abstract} 
We study a two-dimensional Pauli operator describing a charged quantum particle with spin $1/2$ moving on a plane in presence of an orthogonal Aharonov-Bohm magnetic flux. We classify all the admissible self-adjont realizations and give a complete picture of their spectral and scattering properties. Symmetries of the resulting Hamiltonians are also discussed, as well as their connection with the Dirac operator perturbed by an Aharonov-Bohm singularity.
\end{abstract}



\maketitle
\section{Introduction}

The motion of a charged particle with spin $1/2$ on a two-dimensional plane in presence of an external magnetic field perpendicular to it is described in non-relativistic quantum mechanics by the Pauli operator
\be\label{eq:Pauli}
	H_{\mathrm{P}} = \big(\bm{\sigma}\cdot(-i\nabla +\mathbf{A})\big)^2\,, 
\ee
where, for any $\xv = (x,y) \in \mathbb{R}^2$, $\mathbf{A}$ is a vector potential such that $b = \mbox{curl}\mathbf{A} = \nabla^{\perp} \cdot \aav $, $ \nabla^{\perp} : = (- \partial_y, \partial_x) $, equals the magnetic field and $\bm{\sigma} = (\sigma_1,\sigma_2)$ is a matrix-valued vector whose components are the Pauli matrices
	\be
		\sigma_1=\begin{pmatrix}0 & 1 \\ 1 &0 \end{pmatrix}\,, \qquad \sigma_2=\begin{pmatrix}0 & -i \\ i &0 \end{pmatrix}.
	\ee
For regular magnetic fields, such an operator has already been studied in the literature and we refer to the recent works \cite{BTRS21,BTRS23} and references therein for further details. A similar analysis on non-simply connected domains is performed in \cite{LB22}, sharing some analogies with our purpose here (see also \cite{Fi23} for the case of the Dirac operator).

We are indeed interested in discussing the properties of the Pauli operator when the magnetic field $b$ is concentrated at a point, which we choose as the origin without loss of generality, \emph{i.e.}, 
\[
b(\xv) = 2\pi \alpha\, \delta(\xv) \,, \qquad \alpha \in \R\,.
\]
This ideally corresponds to an infinite solenoid of zero diameter passing through the origin. In the Coulomb gauge, the associated vector potential reads
\be\label{eq:AB}
	\mathbf{A}(\xv) := \alpha\, {\xv^{\perp} \over |\xv|^2}\,,
\ee
where $\xv^{\perp} := (-y,x)$.
This setting matches the one of the famous Aharonov-Bohm (AB) effect \cite{AB59}, where however we are here taking also into account the spin degrees of freedom. 

It is known that the singularity at the origin of the field affects the self-adjointness of the Hamiltonian operator and, accordingly, different self-adjoint extensions of the formal expression exist.
The analogous subject for the Schr\"odinger operator, \emph{i.e.} for a spinless charged particle, has been thoroughly analysed in the literature (we refer to \cite{AT98,CF20,CF23,CF24,DS98,F23,PR11} for further details; see also \cite{DF23,DG21} and \cite{GHT23,GPS22} for results about the related radial operators) and our starting point is a similar classification of all self-adjoint realizations of the Pauli operator \eqref{eq:Pauli} (\cref{sec:friedrichs} and \cref{sec: saext}). Such a question has only been partially addressed in \cite{GS04,Pe05,Pe06}, also for more than one solenoid, yet focusing mainly on the Aharonov-Casher phenomenon about the existence and number of zero-energy modes. In this connection we also mention \cite{EV02}, where a distinguished self-adjoint extension is characterized for Pauli operators with measure-valued magnetic fields.

Our main goal is however a careful description of the spectral and scattering properties of the self-adjoint realizations of the Pauli operator (\cref{sec:spesca}): with the exception of possible embedded eigenvalues, we provide a complete picture of the spectrum of each self-adjoint extension, including zero-energy resonances, and derive an explicit expression of the generalized eigenfunctions. Similar results for regular Pauli operators are contained in \cite{BK23,Ko22} and \cite{AHK23,CFK20} (in the latters some applications to the Friedrichs realization of the AB Pauli operator are discussed). After proving the existence and completeness of the wave operators, this in turn allows us to obtain an expression for the scattering amplitude.

Inspired by the analysis in \cite{FS92,FS93} of the physical symmetries of the regular Pauli operator inherited from quantum field theory, we also investigate such symmetries (\cref{app:2Dsym}) in presence of a singular magnetic field (\cref{sec:symdir}). Furthermore, since the Pauli operator can be formally viewed as the square of the Dirac operator $D = \bm{\sigma} \cdot (-i\nabla +\mathbf{A})$, we examine the interplay of this relation with the self-adjoint extensions mentioned above.

\bigskip

\begin{footnotesize}
\noindent
\textbf{Acknowledgments.}
The present research has been supported by MUR grant ``Dipartimento di Eccellenza'' 2023-2027 of Dipartimento di Matematica, Politecnico di Milano. MC acknowledges the support of PNRR Italia Domani and Next Generation Eu through the ICSC National Research Centre for High Performance Computing, Big Data and Quantum Computing. WB acknowledges the support of the MUR - PRIN 2022 project ``Nonlinear dispersive equations in presence of singularities'' (Prot. N. 20225ATSTP).
\end{footnotesize}

\section{Main Results}

As anticipated our first goal is to classify the self-adjoint realizations of the two-dimensional Pauli operator with an AB magnetic flux at the origin and characterize their spectral and scattering properties, starting from the formal expression given in \eqref{eq:Pauli}.
We use polar coordinates $ (r,\teta) \in (0,+\infty) \times [0, 2\pi) $ and exploit the natural Hilbert space isomorphism $L^2(\mathbb{R}^2; \mathbb{C}^2) \simeq L^2(\mathbb{R},r\,dr) \otimes L^2([0, 2\pi),\diff\teta) \otimes \mathbb{C}^2 $. Without loss of generality, we assume $\alpha \in (0,1)$. {The latter condition can indeed be realized by means of a unitary transformation: for any $\alpha \in \mathbb{R}$, let $ \lf\lfloor \alpha \ri\rfloor \in \mathbb{Z}$ denote the corresponding integer part, then $(U \psi)(r,\teta) := e^{- i \lf\lfloor \alpha \ri\rfloor \teta} \, \psi(r,\teta)$ is a unitary operator on $L^2(\mathbb{R}^2; \mathbb{C}^2)$ fulfilling
	$$ U\, H_{\mathrm{P}}\, U^{-1} = \big(\bm{\sigma} \cdot (-i\nabla + \widetilde{\mathbf{A}})\big)^2 \,, \qquad \widetilde{\mathbf{A}} := \big(\alpha - \lf\lfloor \alpha \ri\rfloor \big) {\xv^{\perp} \over |\xv|^2}\,. $$
	Let us stress that $U$ is indeed a unitary map, yet not a gauge transformation since the magnetic fluxes associated to $H_{\mathrm{P}}$ and $U H_{\mathrm{P}} U^{-1}$ are different. 
	
The operator $H_{\mathrm{P}}$ is understood as a closable symmetric operator on the dense domain $C^\infty_c(\mathbb{R}^2 \setminus \{\mathbf{0}\}; \mathbb{C}^2) = C^{\infty}_c(\mathbb{R}^2  \setminus \{\mathbf{0}\}) \otimes \mathbb{C}^2 $. However, it is also easy to verify that $ H_{\mathrm{P}} $ is not essentially self-adjoint on such a domain, due to the singularity of $ \mathbf{A} $ at the origin. At the algebraic level, the operator \eqref{eq:Pauli} can be written as
	\be\label{eq:Pauli2}
		H_{\mathrm{P}} = \begin{pmatrix}  \Pi_{-} \Pi_{+} &	0  \\	0 & \Pi_{+} \Pi_{-}\end{pmatrix} ,
	\ee
where $\Pi_\pm:=(p_1\pm ip_2)$ and $p_j $, $ j = 1,2$, are the components of the vector $ \mathbf{p} := - i \nabla + \aav $. It is readily seen that $\Pi_\pm$ are one the formal adjoint of the other. Passing to polar coordinates and noting that $\mathbf{A}(r,\teta) = \tfrac{\alpha}{r}\,(-\sin\teta,\cos\teta)$, $\partial_{x}=\cos\teta\,\partial_r - \tfrac{\sin\teta}{r}\,\partial_\teta$ and $\partial_{y} = \sin\teta\,\partial_r+\tfrac{\cos\teta}{r}\,\partial_\teta$, by a simple computation we get
	\bdm
		\Pi_+ = e^{i\teta}\left(-i\partial_r + \tfrac{1}{r}\,\partial_\teta + \tfrac{i \alpha}{r} \right) ,\qquad 
		\Pi_- = e^{-i\teta}\left(-i\partial_r-\tfrac{1}{r}\,\partial_\teta - \tfrac{i\alpha}{r} \right) .
	\edm
Then, using the basic relation $\big[\partial_\teta, e^{\pm i\teta}\big] = \pm i\,e^{i\teta}$, we find
	\begin{align*}
		\Pi_- \Pi_+ & = \left(-i\partial_r-\tfrac{1}{r}\,\partial_\teta-\tfrac{i(\alpha+1)}{r} \right)\left(-i\partial_r+\tfrac{1}{r}\,\partial_\teta+\tfrac{i\alpha}{r} \right)
			= - \,\tfrac{1}{r}\, \partial_r \big(r \,\partial_{r} \,\cdot \big) + \tfrac{1}{r^2}\,(- i\,\partial_{\teta} + \alpha)^2\,, 
			\\
		\Pi_+ \Pi_- & = \left(-i\partial_r+\tfrac{1}{r}\,\partial_\teta+\tfrac{i(\alpha-1)}{r} \right) \left(-i\partial_r-\tfrac{1}{r}\,\partial_\teta-\tfrac{i\alpha}{r} \right)
			= - \,\tfrac{1}{r}\, \partial_r \big(r \,\partial_{r} \,\cdot \big) + \tfrac{1}{r^2}\,(- i\,\partial_{\teta} + \alpha)^2\,.
	\end{align*}
The above identities show that formally
	\be
		\label{eq: AB Schroedinger}
		\Pi_{-} \Pi_{+} = \Pi_{+} \Pi_{-} = H_{\mathrm{S}} \,,
	\ee
where $H_{\mathrm{S}} = (-i\nabla + \mathbf{A})^2$ is the AB Schr\"odinger operator acting on scalar functions.

\subsection{Friedrichs extension}\label{sec:friedrichs}

We start by discussing the properties of the most natural self-adjoint extension of $ \hpauli $, \emph{i.e.}, the Friedrichs one. The simplest way to define it is to consider the quadratic form associated to $ \hpauli $:
	\be\label{eq: Qpdef}
		Q_{\mathrm{P}}[\psiv] := \sum_{s \in \{\uparrow,\downarrow\}}\int_{\mathbb{R}^2}\!\! \diff \xv\; \big|(-i \nabla + \mathbf{A}) \psi_{s}\big|^2\,, 
	\ee
	making sense at least for
	\be
		\psiv := \begin{pmatrix} \psi_{\uparrow} \\ \psi_{\downarrow} \end{pmatrix} \in C^{\infty}_{c}\big(\mathbb{R}^2 \setminus  \{\mathbf{0}\}; \mathbb{C}^2\big)\,.
	\ee
	We use the following convention: spinors, \emph{i.e.}, functions from $ \R^2 \to \mathbb{C}^2 $, are denoted by italic bold letters (\emph{e.g.}, $ \psiv, \giv $, etc.), while regular vectors are denoted by bold letters (\emph{e.g.}, $ \xv, \qv $, etc.).	
Making reference to the associated norm $\|\psiv\|_{Q_{\mathrm{P}}} := \|\psiv\|_{2} + Q_{\mathrm{P}}[\psiv]$, we introduce the Friedrichs realization
	\bdm
		\dom\big[Q_{\mathrm{P}}^{(\mathrm{F})}\big] := \overline{C^{\infty}_{c}\big(\mathbb{R}^2 \setminus \{\mathbf{0}\},\mathbb{C}^2\big)}^{\; \|\,\cdot\,\|_{Q_{\mathrm{P}}}}\,, \qquad
		Q_{\mathrm{P}}^{(\mathrm{F})}[\psiv] := Q_{\mathrm{P}}[\psiv]\,.
	\edm
By a straightforward adaptation of \cite[Prop. 1.1]{CF20} (see also \cite[Prop. 1.2]{CF23}), we get the forthcoming \cref{prop:QFHF}}.
Here and in the sequel we refer to the angular average $ \avg{f} : \R^{+} \to \mathbb{C} $ of any scalar function $ f: \mathbb{R}^2 \to \mathbb{C} $, given by
	\begin{equation*}
		\avg{f}(r) : = \frac{1}{2 \pi} \int_0^{2\pi}\!\!\! \diff \teta \: f(r,\teta)\,.
	\end{equation*}

	\begin{proposition}[Friedrichs realization]\label{prop:QFHF}
		\mbox{} \\
		Let $\alpha \in (0,1)$. Then,
		\begin{enumerate}[i)]
			\item the quadratic form $Q_{\mathrm{P}}^{(\mathrm{F})}$ is closed and non-negative on the domain
				\bdm
					\dom\big[Q_{\mathrm{P}}^{(\mathrm{F})}\big] = \lf\{ \psiv \in H^1(\mathbb{R}^2; \mathbb{C}^2) \;\big|\; A_j \psiv \in L^2(\mathbb{R}^2;\mathbb{C}^2)\;\mbox{for } j=1,2 \ri\}\,;
				\edm
			\item for any $\psiv  \in \dom\big[Q_{\mathrm{P}}^{(\mathrm{F})}\big]$ and for any $ s \in \{\uparrow,\downarrow\} $,
				\be\label{asyF}
					\lim_{r \to 0^{+}} \lf\langle\, |\psi_{s}|^2 \,\ri\rangle(r) = 0\,, \qquad 
					\lim_{r \to 0^{+}} r^2 \lf\langle\, |\partial_r\psi_{s}|^2 \,\ri\rangle(r) = 0\,;
				\ee
			
			\item the unique self-adjoint operator $ \hpaulif $ associated to $Q_{\mathrm{P}}^{(\mathrm{F})}$ is 
				\be
					\dom\big(H_{\mathrm{P}}^{(\mathrm{F})}\big) := \lf\{ \psiv \in \dom\big[Q^{(\mathrm{F})}\big] \;\Big|\; H_{\mathrm{P}}\, \psiv \in L^2(\mathbb{R}^2;\mathbb{C}^2) \ri\} , \qquad
					H_{\mathrm{P}}^{(\mathrm{F})} \psiv := H_{\mathrm{P}}\, \psiv\,.
				\ee
		\end{enumerate}
	\end{proposition}

	\begin{remark}[Decomposition of $ \hpaulif $]\label{rem:HPF}
		\mbox{}	\\
		Due to the diagonal structure \eqref{eq:Pauli2} \eqref{eq: AB Schroedinger} of the Pauli operator $H_{\mathrm{P}}$, from \cref{prop:QFHF} it readily follows that
		\be\label{HPFdiag}
			\hpaulif = \begin{pmatrix} H_{\mathrm{S}}^{(\mathrm{F})} &	0  \\0 & H_{\mathrm{S}}^{(\mathrm{F})} \end{pmatrix} ,
		\ee
		where $H_{\mathrm{S}}^{(\mathrm{F})}$ is the Friedrichs realization of the AB Schr\"odinger  operator \eqref{eq: AB Schroedinger} characterized in \cite[Prop. 1.1]{CF20}.
	\end{remark}
	
\subsection{Self-adjoint extensions}
\label{sec: saext}

As anticipated, our first goal is to classify all the self-adjoint realizations of the operator $ \hpauli $. To this purpose we will provide a family of quadratic forms inspired by those associated to the self-adjoint extensions of the AB Schr\"{o}dinger operator $ \hsch $ and show {\it a posteriori} that such forms are closed and bounded from below, as well as the fact that the associated operators exhausts all possible extensions of $ \hpauli $.

Let then $g_{\lambda}^{(0)},g_{\lambda}^{(-1)} \in L^2(\mathbb{R}^2;\mathbb{C})$ be the unique solutions to the AB Sch\"odinger defect equation 	
	\bdm
		\big(H_{\mathrm{S}} + \lambda^2\big) g_{\lambda}^{(\ell)} = 0, \qquad	\mbox{in } \mathbb{R}^2 \setminus \{\mathbf{0}\}, 
	\edm
	namely (see \cite[\S 10.31]{OLBC10}),
	\be\label{eq:g}
		g_{\lambda}^{(\ell)}(r,\teta) = \;\lambda^{|\ell+\alpha|}\,K_{|\ell+\alpha|}(\lambda r)\,\tfrac{e^{i \ell \teta}}{\sqrt{2\pi}}, \qquad \ell \in \{0,-1\}\,.
	\ee
	We stress (see next \eqref{eq:G01asy}) that such functions have a local singularity at the origin proportional to $ r^{-|\ell + \alpha|} $ and for this reason they do not belong to the domain of the Friedrichs realization $ \hschf $. We construct then out these defect functions four independent  solutions in $L^2(\mathbb{R}^2; \mathbb{C}^2)$ of the formal equation $(H_{\mathrm{P}} + \lambda^2) \giv_{\lambda} = 0$ in $\mathbb{R}^2 \setminus \{\mathbf{0}\}$, {\it i.e.},
	\be\label{G1234}
		\giv_{\lambda,s}^{(\ell)} := g_{\lambda}^{(\ell)} \begin{pmatrix} \delta_{s,\uparrow} \\ \delta_{s,\downarrow} \end{pmatrix} , \qquad s \in \{\uparrow,\downarrow\},\, \ell \in \{0,-1\}\,.
	\ee
By a heuristic evaluation of the expectation value $\meanlrlr{\psi}{H_{\mathrm{P}}}{\psi}$ for spinors  of the form 
\beq
	\psiv = \phiv_{\lambda} + \sum_{s,\ell} q_{s}^{(\ell)} \giv_{\lambda,s}^{(\ell)}
\eeq
with $\phiv_{\lambda} \!\in\! \dom\big[Q_{\mathrm{P}}^{(\mathrm{F})}\big]$ and $\qv := \big(q_{\uparrow}^{(0)}, q_{\uparrow}^{(-1)}, q_{\downarrow}^{(0)}, q_{\downarrow}^{(-1)} \big) \in \mathbb{C}^4$, we are lead to consider the quadratic form
	\be\label{QB2}
		Q_{\mathrm{P}}^{(\bb)}[\psiv] := Q_{\mathrm{P}}^{(\mathrm{F})}[\phiv_{\lambda}] - \lambda^2 \,\|\psiv\|_{2}^{2} + \lambda^2 \,\| \phiv_{\lambda}\|_{2}^{2} +\qv^*\cdot[L(\lambda)+\bb]\qv \,,
	\ee
where $ \beta = \big(\bb_{s\,s'}^{(\ell\,\ell')}\big)_{s,s' \in \{\uparrow,\downarrow\},\, \ell,\ell' \in \{0,-1\}} \in \mathrm{M}_{4,\mathrm{Herm}}(\mathbb{C})$ is any $4 \times 4$ Hermitian matrix labeling the extension and 
\begin{equation}\label{eq:L}
L(\lambda) = \left( \tfrac{\pi\,\lambda^{2|\ell + \alpha|}}{2 \sin(\pi \alpha)}\,\delta_{s\,s'}\,\delta_{\ell\,\ell'}\right)_{s,s' \in \{\uparrow,\downarrow\},\, \ell,\ell' \in \{0,-1\}} .
\end{equation}
Here and below we systematically write $\sum_{s,\ell}$ to indicate the double sum $\sum_{s \in \{\uparrow,\downarrow\}}\sum_{\ell \in \{0,-1\}}$. Moreover, we shall refer to the decomposition in angular harmonics, for fixed $s \in \{\uparrow,\downarrow\}$,
	\bdm
		\psi_{s}(r,\teta) = \sum_{\ell \in \mathbb{Z}}\,\psi_{s}^{(\ell)}(r) \,\tfrac{e^{i \ell \teta}}{\sqrt{2\pi}}\,.
	\edm

	\begin{theorem}[Self-adjoint extensions of $ \hpauli $]\label{thm:QBHB}
		\mbox{} \\
		Let $\alpha \in (0,1)$ and $ \lambda > 0 $. Then, for any $\bb \in \mathrm{M}_{4,\mathrm{Herm}}(\mathbb{C})$,
		\begin{enumerate}[i)]
			\item the quadratic form $Q_{\mathrm{P}}^{(\bb)}$ is well-defined on the domain
				\bdm
					\dom\big[Q_{\mathrm{P}}^{(\bb)}\big] = \Big\{ \psiv = \phiv_{\lambda} + \tx{\sum_{s,\ell}}\, q_s^{(\ell)} \giv_{\lambda,s}^{(\ell)} \in L^2(\mathbb{R}^2; \mathbb{C}^2) \,\Big|\, \phiv_{\lambda} \in \dom\big[Q_{\mathrm{P}}^{(\mathrm{F})}\big],\; \qv  \in \mathbb{C}^4 \Big\}\,,
				\edm
				and it is independent of $\lambda$, closed and bounded from below;
			
			\item the unique self-adjoint operator $H_{\mathrm{P}}^{(\bb)}$ associated to $Q_{\mathrm{P}}^{(\bb)}$ is
				\bml{\label{HBdom}
					\dom\big(H_{\mathrm{P}}^{(\bb)}\big) = \lf\{ \psiv = \phiv_{\lambda} + \tx{\sum_{s,\ell}}\, q_s^{(\ell)} \giv_{\lambda,s}^{(\ell)} \in \dom\big[Q_{\mathrm{P}}^{(\bb)}\big] \;\Big|\; \phiv_{\lambda} \in \dom\big(H_{\mathrm{P}}^{(\mathrm{F})}\big)\,, \ri. \\
					\big[\big(L(\lambda) + \bb\big) \qv \big]_{s}^{(\ell)} 
					= 2^{|\ell+\alpha|-1}\, \Gamma\big(|\ell+\alpha|\big) \lim_{r \to 0^{+}} \frac{1}{r^{|\ell+\alpha|}} \big(|\ell+\alpha| + r\, \partial_r \big)\phi_{\lambda,s}^{(\ell)} \,\bigg\}\, ,
				}
			
				\be\label{HBaction}
					\big( H_{\mathrm{P}}^{(\bb)} \!+ \lambda^2 \big) \psiv := \big( H_{\mathrm{P}}^{(\mathrm{F})} \!+ \lambda^2\big) \phiv_{\lambda}\,;
				\ee
				
			\item the family $ \big(H_{\mathrm{P}}^{(\bb)} \big)_{\bb \in \mathrm{M}_{4,\mathrm{Herm}}(\mathbb{C})} $ exhausts all possible self-adjoint extensions of the symmetric operator $ \hpauli $ given in \eqref{eq:Pauli}.
		\end{enumerate}
	\end{theorem}

\begin{remark}[Local singularities]
	\label{remjan}
	\mbox{}	\\
	For a generic $ \bb \in M_{4,\,\mathrm{Herm}}(\mathbb{C})$, considering that the Friedrichs Hamiltonian $H_{\mathrm{P}}^{(\mathrm{F})}$ is the direct sum of two copies of the Friedrichs AB Schr\"odinger operator (see \cref{rem:HPF}), working in polar coordinates, using results from \cite{DR17} and exploiting the asymptotic expansion of $\giv_{\lambda,s}^{(\ell)}$ (see \eqref{G1234} and \eqref{eq:G01asy}, below), for any $\psiv \in \dom\big(H_{\mathrm{P}}^{(\bb)}\big)$ one infers that, as $r \to 0^{+}$,
\begin{equation*}
	\psi_{s}(r,\theta) 
	= \tx{\sum_{\ell \,\in\, \{0,-1\}}} \lf[ 
		\tfrac{\Gamma( |\ell+\alpha| ) }{ 2^{1-|\ell+\alpha|}}\,\tfrac{q_{s}^{(\ell)}}{r^{|\ell+\alpha|}} 
		+ \left( \phi_s^{(\ell)} + q_{s}^{(\ell)}\,\tfrac{\Gamma(\!-|\ell+\alpha|) }{ 2^{1 + |\ell+\alpha|}}\, \lambda^{2|\ell+\alpha|} \right) r^{|\ell+\alpha|} 
		\ri] \tfrac{e^{i \ell \theta}}{ \sqrt{2\pi}} + o(r)\,,
\end{equation*}
for some $\phi_s^{(\ell)} \in \CC$. Noticing that the boundary conditions in \eqref{HBdom} can be equivalently expressed as
\begin{equation*}
\big[\big(L(\lambda) + \bb\big) \qv \big]_{s}^{(\ell)} = 2^{|\ell+\alpha|}\, \Gamma\big(|\ell+\alpha| + 1\big) \,\phi_s^{(\ell)}\,,
\end{equation*}
and recalling the definition \eqref{eq:L} of $L(\lambda)$, the above asymptotics can be equivalently expressed as
\begin{equation}\label{eq:psiasy}
	\psi_{s}(r,\theta) 
	= \tx{\sum_{\ell \,\in\, \{0,-1\}}} \bigg[ 
		\tfrac{\Gamma( |\ell+\alpha| ) }{ 2^{1-|\ell+\alpha|}}\,{q_{s}^{(\ell)} \over r^{|\ell+\alpha|}}
		+ \frac{1}{2^{|\ell+\alpha|}\, \Gamma(|\ell+\alpha| + 1)}\,\big[\bb \qv \big]_{s}^{(\ell)}\, r^{|\ell+\alpha|} 
		\bigg] \,\tfrac{e^{i \ell \theta}}{ \sqrt{2\pi}} + o(r)\,,
\end{equation}
\end{remark}

	\begin{remark}[Friedrichs and Krein extensions]
		\mbox{}	\\
		The Friedrichs Hamiltonian $H_{\mathrm{P}}^{(\mathrm{F})}$ is formally recovered taking $\bb = \mbox{``}\infty\mbox{''} \one $, \emph{i.e.}, setting all the charges $ \qv $ equal to zero.
		Another distinguished extension is the Krein's Hamiltonian, \emph{i.e.}, the smallest positive one in form sense.
		By directed inspection of \eqref{QB2}, noting that $L(\lambda)$ is positive (see \eqref{eq:L}), it can be inferred that the latter coincides with $ H_{\mathrm{P}}^{(0)} $, namely, the extension with $\bb = 0$. 
		It is also noteworthy that the Friedrichs' and the Krein's Hamiltonians are homogeneous of degree $-2$ under the $L^2$-unitary scaling		
	\[
	u(\cdot)\mapsto \eta\, u(\eta\,\cdot)\qquad (\eta>0)\,.
	\]
	The same property is shared by the other Hamiltonians $ H_{\mathrm{P}}^{(\beta)} $ which coincide with either the Friedrichs or the Krein one in each of the channels $(s,\ell)$, $s \in \{\uparrow,\downarrow\}$ and $\ell \in \{0,-1\}$, separately.
	\end{remark}
	
	\begin{remark}[von Neumann parametrization]
		\mbox{}	\\
		Given the solutions of the defect equation, it is possible to parametrize the self-adjoint extensions of $ \hpauli $ via the von Neumann theory. Denoting by $ \overline{H}_{\mathrm{P}} $ the closure of the symmetric operator \eqref{eq:Pauli}, there exists a 4-parameter family of self-adjoint extensions given, for any $4 \!\times\! 4$ unitary matrix $U \in \mathrm{M}_{4,\mathrm{Unit}}(\mathbb{C})$, by
		\begin{gather}
		\begin{aligned}
\dom\big(H_{\mathrm{P}}^{(U)}\big) = \Big\{ \psiv = \fv + \mathcal{G}_{+} \cv + \mathcal{G}_{-} U \cv \,\in L^2(\mathbb{R}^2; \mathbb{C}^2) \;\Big|\; 
			& \fv \in \dom\big(\overline{H}_{\mathrm{P}}\big)\,,\; \cv \!\in \mathbb{C}^4 \Big\}\,, 
		\end{aligned} \label{eq:vNdom} \\
		H_{\mathrm{P}}^{(U)} \psiv = H_{\mathrm{P}} \fv + i\, \mathcal{G}_{+} \cv - i\, \mathcal{G}_{-} U \cv \,,\label{eq:vNact}
		\end{gather}
	where we put
	\be\label{eq:defGpm}
		\gmv_{\pm} : \mathbb{C}^4 \to L^2(\mathbb{R}^2,\mathbb{C}^2)\,, \qquad \gmv_{\pm} \bm{c} := \tx{\sum_{s,\ell}}\, c_{s}^{(\ell)} \giv_{\pm,s}^{(\ell)} \,\big/\, \big\| \giv_{\pm,s}^{(\ell)} \big\|_2 \,,
	\ee
	with 
	\be\label{eq:gpm}
		\giv_{\pm,s}^{(\ell)} = g_{\pm}^{(\ell)}  \begin{pmatrix} \delta_{s,\uparrow} \\ \delta_{s,\downarrow} \end{pmatrix}, \qquad 
		s \in \{\uparrow,\downarrow\},\, \ell \in \{0,-1\}\,,
	\ee
	\be\label{eq:xipm}
		g_{\pm}^{(\ell)}(r,\teta) = e^{\mp i \tfrac{\pi}{4}\, |\ell+\alpha| }\, \sqrt{\tfrac{4}{\pi} \cos\big( \tfrac{\pi}{2}\,|\alpha + \ell|\big)}\,\tx{e^{i \ell \teta} \over \sqrt{2\pi}}\,K_{|\ell+\alpha|}\big(e^{\mp i \pi/4}\, r\big)\,,
		\qquad \ell \in \{0,-1\}\,.
	\ee
	Of course there is a one-to-one correspondence between such a family and the family of operators introduced in \cref{thm:QBHB}, which can be made explicit by deriving a relation $ U = U(\bb) $ (see next \cref{prop:quadKrein} and \cref{rem: vN and Krein}).

	Comparing the asymptotics at the origin \eqref{eq:psiasy} for the functions in the domains $\dom\big(H_{\mathrm{P}}^{(\bb)}\big)$ and $\dom\big(H_{\mathrm{P}}^{(U)}\big)$ (see \eqref{HBdom} and \eqref{eq:vNdom}, respectively) one can check that the Friedrichs and Krein realizations, respectively, correspond to
	\bdm
			U^{(\mathrm{F})} = -\one\,, \qquad
			U^{(\mathrm{K})} = \big(-e^{i\pi\vert \alpha+\ell\vert}\delta_{ss'}\delta_{\ell\ell'}\big) \,.
		\edm	
	\end{remark}

\subsection{Spectral and scattering properties}\label{sec:spesca}

In order to investigate the spectral and scattering properties of the self-adjoint operators $ H_{\mathrm{P}}^{(\bb)} $, we exploit general resolvent arguments using the Birman-Krein-Vishik theory of self-adjoint extensions to write the resolvent operator. We are going to refer to the general theory described in \cite{Po01,Po08}. 

Let us first notice that the Friedrichs Hamiltonian $H_{\mathrm{P}}^{(\mathrm{F})}$ is positive semi-definite and consider the associated resolvent operator
	\bdm
		R^{(\mathrm{F})}_{\mathrm{P}}(z) := \big(H_{\mathrm{P}}^{(\mathrm{F})} - z\big)^{-1} : L^2(\mathbb{R}^2;\mathbb{C}^2) \to \dom\big(H_{\mathrm{P}}^{(\mathrm{F})}\big) \,, \qquad
		\mbox{for\; $z \in \mathbb{C} \setminus \R^{+}$} .
	\edm
	
Taking into account the diagonal structure \eqref{HPFdiag} of the Friedrichs Hamiltonian $H_{\mathrm{P}}^{(\mathrm{F})}$ we readily  get
	\be\label{eq:Pauli2res}
		R^{(\mathrm{F})}_{\mathrm{P}}(z) = \begin{pmatrix} R^{(\mathrm{F})}_{\mathrm{S}}(z) &	0  \\	0 & R^{(\mathrm{F})}_{\mathrm{S}}(z) \end{pmatrix} ,
	\ee
where $R^{(\mathrm{F})}_{\mathrm{S}}(z) = \big(H_{\mathrm{S}}^{(\mathrm{F})} - z\big)^{-1}$ is the resolvent operator for the Friedrichs realization of the AB Schr\"odinger operator, acting as an integral operator with kernel \cite[Eq. (3.2)]{AT98}
	\be\label{eq:Fresolv}
		R_{\mathrm{S}}^{(\mathrm{F})}(z; \xv, \xvp) = \sum_{\ell\in\Z}\, I_{\vert \ell+\alpha\vert}\big(-i\sqrt{z}\,(r \land r')\big)\,K_{\vert \ell+\alpha\vert}\big(-i\sqrt{z}\,(r \lor r')\big)\, \tfrac{e^{i \ell (\teta-\teta')}}{2\pi} \,.
	\ee
	Notice that the kernel in \eqref{eq:Fresolv} is slightly different, compared to \cite{AT98}, as we write it using the modified Bessel functions of second kind $I_\nu, K_\nu$ in place of the Bessel and Hankel functions $J_{\nu}, H_{\nu}^{(1)}$ (this is obtained using the connection formulas \cite[Eqs. 10.27.6 and 10.27.8]{OLBC10}).
Here and in the sequel we refer to the determination of the square root with $\Im \sqrt{z} > 0$ for all $z \in \mathbb{C} \setminus \R^{+}$, ensuring in particular that $\Re\big(\!-i \sqrt{z}\big) > 0$.

In view of the boundary conditions appearing in \eqref{HBdom}, we further  introduce the trace operator
	\be\label{eq:traceop}
		\tau = \bigoplus_{s,\ell} \tau_{s}^{(\ell)} : \dom\big(H_{\mathrm{P}}^{(\mathrm{F})}\big) \to \mathbb{C}^4, \qquad
		\tau_{s}^{(\ell)} \phiv := 2^{|\ell+\alpha|-1}\, \Gamma\big(|\ell+\alpha|\big) \lim_{r \to 0^{+}} \frac{1}{r^{|\ell+\alpha|}}\big(|\ell+\alpha| + r\, \partial_r \big)\phi_{s}^{(\ell)} .
	\ee
For any $z \in \mathbb{C} \setminus \R^{+}$, we put
	\be\label{eq:Guzdef}
		\widecheck{\mathcal{G}}(z) := \tau R_{\mathrm{P}}^{(\mathrm{F})}(z) :\, L^2(\mathbb{R}^2; \mathbb{C}^2) \to \mathbb{C}^{4}\,,
	\ee
and define the corresponding single layer operator as
	\be\label{eq:Gzdef}
		\gmv(z) := \big( \widecheck{\mathcal{G}}(z^{*}) \big)^{*} :\, \mathbb{C}^{4} \to L^2(\mathbb{R}^2; \mathbb{C}^2)\,.
	\ee
	The Weyl operator then reads
	\be\label{eq: defLambda}
		\Lambda(z) := \tau \lf(\gmv(-1) - \gmv(z) \ri) \,:\, \mathbb{C}^{4} \to \mathbb{C}^{4},
	\ee
	where we have chosen as a reference spectral point $ z_0 = -1 $. The first result we state is precisely about the resolvent of the self-adjoint realizations of $ \hpauli $.

	\begin{theorem}[Resolvent of the self-adjoint extensions of $ \hpauli $]\label{thm:HTheta}
		\mbox{}		\\
		Let $\Theta \in \mathrm{M}_{4,\mathrm{Herm}}(\mathbb{C})$. There exists a non-empty open set $ \mathcal{Z} \subset \mathbb{C} \setminus \R $, such that, for any $z \in \mathcal{Z}$, the bounded operator 
			\begin{equation}\label{eq: RThetaKrein}
				R_{\mathrm{P}}^{(\Theta)}(z) := R_{\mathrm{P}}^{(\mathrm{F})}(z) + \gmv(z) \big[ \Lambda(z) + \Theta \big]^{-1} \widecheck{\mathcal{G}}(z)
			\end{equation}
		is the resolvent of a self-adjoint operator $H_{\mathrm{P}}^{(\Theta)}$ coinciding with $H_{\mathrm{P}}^{(\mathrm{F})}$ on $ \ker (\tau)$ and defined by
			\begin{gather}
				\dom\big(H_{\mathrm{P}}^{(\Theta)}\big) := \lf\{ \psiv \!\in\! L^2(\mathbb{R}^2;\mathbb{C}^2)\;\Big|\; \psiv = \varphiv_{z} + \mathcal{G}(z) \mathbf{q},\; \varphiv_{z} \!\in\! \dom\big(H_{\mathrm{P}}^{(\mathrm{F})}\big),\; \mathbf{q} \!\in\! \mathbb{C}^{4}, \; \tau \varphiv_z = \big[\Lambda(z) + \Theta \big] \mathbf{q} \ri\} , \nonumber \\
				\big(H_{\mathrm{P}}^{(\Theta)}\! - z\big) \psiv = \big(H_{\mathrm{P}}^{(\mathrm{F})}\! - z\big) \varphiv_z\,. \label{eq:HThetadef}
			\end{gather}
	\end{theorem}		
	
	\begin{remark}[Range of validity of \eqref{eq: RThetaKrein}]
		\mbox{}	\\
		The set $\mathcal{Z} \subset \mathbb{C}  \setminus  \R^{+}$ consists of points $ z $ in the complex plane for which the $4 \times 4$ matrix $\Theta + \Lambda(z)$ is invertible, and it is not difficult to see that such a set is certainly non-empty (see next \eqref{eq:Lzexp}). In fact, by \cite[Thm. 2.19]{CFP18}, the defining identity \eqref{eq: RThetaKrein} extends to any $z \in \rho\big(H_{\mathrm{P}}^{(\mathrm{F})}\big) \cap \rho\big(H_{\mathrm{P}}^{(\Theta)}\big)$, where $\rho\big(H_{\mathrm{P}}^{(\mathrm{F})}\big)$ and $\rho\big(H_{\mathrm{P}}^{(\Theta)}\big)$ are the resolvent sets of $H_{\mathrm{P}}^{(\mathrm{F})}$ and $H_{\mathrm{P}}^{(\Theta)}$, respectively.
	\end{remark}

	Of course, the above \cref{thm:HTheta} provides yet another parametrizion of the family of self-adjoint extensions of $ \hpauli $. In fact, this parametrization comprises all such realizations by general arguments and therefore there must be a one-to-one correspondence with the family in \cref{thm:QBHB}.
	
	\begin{proposition}[Equivalence of parametrizations] \label{prop:quadKrein}
		\mbox{}	\\
		There is a one-to-one correspondence between the families $ \big\{ H_{\mathrm{P}}^{(\Theta)} \big\}_{\Theta \in M_{4,\mathrm{Herm}}(\mathbb{C})} $ and $ \big\{ H_{\mathrm{P}}^{(\bb)} \big\}_{\bb \in M_{4,\mathrm{Herm}}(\mathbb{C})} $, given by
			\be \label{eq: ThetaB}
				\Theta = \Theta(\bb) = L(1) + \bb\,,
			\ee
		where $L(1) = \frac{\pi}{2\sin(\pi\alpha)}\,\one$\,, see \eqref{eq:L}, matches the choice $z_0=-1$ in \eqref{eq: defLambda}.
	\end{proposition}
	
	\begin{remark}[Friedrichs and Krein extensions] \label{rem: FriedInf}
		\mbox{}	\\
		Also in this case it appears that the Friedrichs Hamiltonian $H_{\mathrm{P}}^{(\mathrm{F})}$ is formally recovered for $\Theta = \mbox{``}\infty\mbox{''} \one$, while the Krein's one is identified by $ \Theta = \Theta(0) = L(1) $.
	\end{remark}

	\begin{remark}[von Neumann and Krein parametrizations]
		\label{rem: vN and Krein}
		\mbox{}	\\
		The one-to-one correspondence between the von Neumann and Krein families of self-adjoint realizations is realized explicitly by  (see, e.g., \cite[Thm. 4.1 and Thm. 4.3]{Po03} and \cite[Thm. 3.1]{Po08})
			\be \label{eq: ThetaU}
				\Theta = \Theta(U) = -\,i\, \widecheck{\mathcal{G}}(+i)\, \lf( \widehat{U} - \widehat{U}_{\ast} \ri) \lf( \widehat{U} + \widehat{U}_{\ast} \ri)^{\!-1} \gmv(-i)\,,
			\ee
		where $\widehat{U}$ is a unitary operator on the defect space $ \mathrm{span}\{ \gmv_{+} \mathbf{c} \} $, $ \mathbf{c} \in \mathbb{C}^4 $, acting as $\widehat{U}\,\gmv_{+} \mathbf{c} := \gmv_{-}\,U \mathbf{c}$, and $\widehat{U}_{\ast} := \big(H_{\mathrm{P}}^{(\mathrm{F})}\! - i\big) R_{\mathrm{P}}^{(\mathrm{F})}(- i)$ is the restriction of the Cayley transform of $H_{\mathrm{P}}^{(\mathrm{F})}$ to the same subspace. 
	\end{remark}

Let us continue the investigation of the spectral and scattering properties of the Pauli Hamiltonians characterized in \cref{sec: saext} as distinct self-adjoint realizations in $L^2(\mathbb{R}^2;\mathbb{C}^2)$ of the differential operator $H_{\mathrm{P}}$. For convenience, we start by dealing with the Friedrichs extension and study its scattering properties w.r.t. the (self-adjoint) free Pauli operator
	\bdm
		-\Delta_{\mathrm{P}} := \begin{pmatrix} - \Delta & 0 \\ 0 & -\Delta \end{pmatrix},
	\edm
with domain $  H^2(\mathbb{R}^2;\mathbb{C}^2) $. Notice that $\sigma(-\Delta_{\mathrm{P}}) = \sigma_{\mathrm{ac}}(-\Delta_{\mathrm{P}}) = \R^{+}$, so the projector onto the subspace of absolute continuity of $ -\Delta_{\mathrm{P}} $ satisfies $P_{\mathrm{ac}}(-\Delta_{\mathrm{P}}) = \one$. More in general, for any $\Theta \in \mathrm{M}_{4,\mathrm{Herm}}(\mathbb{C})$, we define the {\it wave operators}
	\bdm
		\Omega^{(\Theta)}_{\pm} \equiv \Omega_{\pm}\big(H_{\mathrm{P}}^{(\Theta)},-\Delta_{\mathrm{P}}\big) 
		:= \slim_{t \to \pm \infty}\, e^{i t H_{\mathrm{P}}^{(\Theta)}} e^{-i t (-\Delta_{\mathrm{P}})}\,.
	\edm
For later convenience, we refer here to the family $ H_{\mathrm{P}}^{(\Theta)} $, $ \Theta \in \mathrm{M}_{4,\mathrm{Herm}}(\mathbb{C}) $, described in \cref{thm:HTheta} (understanding $\Theta = \infty \one$ for the Friedrichs Hamiltonian $H_{\mathrm{P}}^{(\mathrm{F})}$, see \cref{rem: FriedInf}). We recall that, whenever they exist, the wave operators are said to be {\it complete} if \cite[Vol. III, p. 19]{RS81}
	\bdm
		\mbox{ran}\,\Omega^{(\Theta)}_{+} 
		= \mbox{ran} \,\Omega^{(\Theta)}_{-}
		= \mbox{ran}\,P_{\mathrm{ac}}\big(H_{\mathrm{P}}^{(\Theta)}\big) \,,
	\edm
where $ P_{\mathrm{ac}}\big(H_{\mathrm{P}}^{(\Theta)}\big)$ is the spectral projector onto the absolute continuity subspace of $L^2(\mathbb{R}^2; \mathbb{C}^2)$ associated to $H_{\mathrm{P}}^{(\Theta)}$. {\it Asymptotic completeness} further requires that $ \sigma_{\mathrm{sc}}\big(H_{\mathrm{P}}^{(\Theta)}\big) = \varnothing $.
Assuming that the wave operators exist, a fact we shall actually prove in the subsequent \cref{prop: waveopF} and \cref{prop: waveopTheta}, we proceed to introduce the {\it scattering operator}
	\bdm
		\mathrm{S}^{(\Theta)} := \big(\Omega^{(\Theta)}_{+} \big)^{*}\,\Omega^{(\Theta)}_{-}.
	\edm
Notice that $\mathrm{S}^{(\Theta)}$ is a unitary operator on $\mbox{ran}\,P_{\mathrm{ac}}\big(H_{\mathrm{P}}^{(\Theta)}\big)$ as soon as the wave operator $\Omega^{(\Theta)}_{\pm}$ are complete.

Finally, we introduce the following definition of zero-energy resonances of Pauli operators. We adopt an analogous definition for the zero-energy resonances of the Dirac operator (see next \cref{thm:Ddom} and \cref{prop: zeroenresDirac}).

\begin{definition}[Zero-energy resonance]
\label{def:res}
\mbox{}	\\
A  zero-energy resonance $ \psiv $ of $ H_{\mathrm{P}}^{(\Theta)} $ is a distributional solution of the equation $H_{\mathrm{P}} \psiv = 0$ in $L^2_{\mathrm{loc}}(\R^2) \setminus L^2(\R^2)$, which fulfils the boundary condition at $ \xv = \mathbf{0}$ encoded in $\dom\big(H_{\mathrm{P}}^{(\Theta)}\big)$ and remains bounded at infinity.
\end{definition}

As anticipated, we start by analyzing the Friedrichs realization.

\begin{theorem}[Scattering for $ H_{\mathrm{P}}^{(\mathrm{F})} $]
\label{prop: waveopF}
\mbox{}	\\
The wave operators $\Omega^{(\mathrm{F})}_{\pm} $ exist and are asymptotically complete. Moreover, the scattering operator $\mathrm{S}^{(\mathrm{F})} $ exists and is unitary on $L^2(\mathbb{R}^2;\mathbb{C}^2)$.
\end{theorem}

A straightforward consequence of the above result is the spectral characterization of $ H_{\mathrm{P}}^{(\mathrm{F})} $.

\begin{corollary}[Spectrum of $ H_{\mathrm{P}}^{(\mathrm{F})} $]
	\label{prop: spectrumF}
	\mbox{}	\\
		The spectrum of the Friedrichs Hamiltonian $H_{\mathrm{P}}^{(\mathrm{F})}$ satisfies
			\be
				\sigma\big(H_{\mathrm{P}}^{(\mathrm{F})}\big) = \sigma_{\mathrm{ac}}\big(H_{\mathrm{P}}^{(\mathrm{F})}\big) = [0,+\infty)\,,
			\ee
			and, in particular, $\sigma_{\mathrm{pp}}\big(H_{\mathrm{P}}^{(\mathrm{F})}\big) = \sigma_{\mathrm{sc}}\big(H_{\mathrm{P}}^{(\mathrm{F})}\big) = \varnothing$. 
	\end{corollary}
	
To investigate further the scattering and spectrum of $ H_{\mathrm{P}}^{(\mathrm{F})} $, we provide an explicit expression of the scattering matrix and amplitude, together with the scattering cross section.
To proceed, let us refer to the plane waves (for $s\in \{\uparrow,\downarrow\}$ and $\kv \in \mathbb{R}^2$)
	\be
		\varphiv_{(s,\kv)}(\xv) := \tfrac{1}{2\pi}\, e^{i \kv \cdot \xv} \begin{pmatrix} \delta_{s,\uparrow} \\ \delta_{s,\downarrow} \end{pmatrix}
		= \tfrac{1}{2\pi} \tx\sum_{\ell} e^{i \ell (\teta-\omega) + i \frac{\pi}{2}\,|\ell|}\, J_{|\ell|}(k\,r)\,\begin{pmatrix} \delta_{s,\uparrow} \\ \delta_{s,\downarrow} \end{pmatrix}  ,
	\ee
where  $\xv = (r,\teta) \in \R^{+} \!\times\! \mathbb{S}^{1}$, $\kv = (k,\omega) \in \R^{+} \!\times\! \mathbb{S}^{1}$ and we used \cite[Eq. 8.511.4]{GR07} and \cite[Eq. 10.4.1]{OLBC10}. Notice that, in the sense of distributions,
	\bdm
		-\Delta_{\mathrm{P}}\, \varphiv_{(s,\kv)} = k^2\, \varphiv_{(s,\kv)}\,.
	\edm
Analogously, the generalized eigenfunctions $\big(\varphiv^{(\mathrm{F},\pm)}_{(s,\kv)}\big)_{(s,\kv) \in \{\uparrow,\downarrow\} \times \mathbb{R}^2}$ corresponding to $\sigma_{\mathrm{ac}}\big(H_{\mathrm{P}}^{(\mathrm{F})}\big)$  are the distributional solutions of the eigenvalue problem
	\bdm
		H_{\mathrm{P}}\, \varphiv^{(\mathrm{F},\pm)}_{(s,\kv)} = k^2\, \varphiv^{(\mathrm{F},\pm)}_{(s,\kv)}\,, 
	\edm
fulfilling the local Friedrichs conditions $\nabla \varphiv^{(\mathrm{F},\pm)}_{(s,\kv)}, A_j \varphiv^{(\mathrm{F},\pm)}_{(s,\kv)} \in L^2_{\mathrm{loc}}(\mathbb{R}^2; \mathbb{C}^2) $ for $j = 1,2$ (see \cref{prop:QFHF}), and the incoming ($+$) or outgoing ($-$) Sommerfeld radiation conditions
	\be\label{eq: Sommerfeld}
		\lim_{r \to +\infty} r^{1/2}\,\big( \hat{\xv} \cdot \nabla \pm i k\big) \lf[ \varphiv^{(\mathrm{F},\pm)}_{(s,\kv)}(\xv) - \varphiv_{(s,\kv)}(\xv) \ri] = 0\,.
	\ee	
	
Correspondingly, we introduce the Fourier transform $\mathfrak{F} : L^2(\mathbb{R}^2;\mathbb{C}^2) \to L^2(\mathbb{R}^2;\mathbb{C}^2) $ defined as
	\beq
		\lf(\mathfrak{F} \psiv \ri)_{s}(\kv) := \!\sum_{s' \in \{\uparrow,\downarrow\}} \int_{\mathbb{R}^2} \!\!\diff\xv \; \big(\varphiv_{(s,\kv)}(\xv)\big)^*_{s'}\,\psi_{s'}(\xv)
			= \tfrac{1}{2\pi} \int_{\mathbb{R}^2} \!\!\diff\xv\; e^{-i \kv \cdot\xv}\,\psi_{s}(\xv)\,,
	\eeq
and the associated unitary map \cite[\S 4.5.1]{BW83}
	\bdm
		F : L^2(\mathbb{R}^2;\mathbb{C}^2) \to \int_{\sigma(-\Delta_{\mathrm{P}})}\!\!\!  \diff\lambda \,\lf( L^2(\mathbb{R}^2;\mathbb{C}^2) \ri)_{\lambda}, \qquad
		\lf(F \psiv \ri)_{\lambda,s}(\omega) := \lf(\mathfrak{F} \psiv \ri)_{s}\big(\sqrt{\lambda}\,,\omega\big) \in \lf( L^2(\mathbb{R}^2;\mathbb{C}^2) \ri)_{\lambda} ,
	\edm
providing a direct integral decomposition of $L^2(\mathbb{R}^2;\mathbb{C}^2)$ with respect to the spectral measure of $-\Delta_{\mathrm{P}}$.
Taking into account that $\mathrm{S}^{(\Theta)}$ commutes with the free Pauli operator $-\Delta_{\mathrm{P}}$ \cite[Vol. III, p. 74]{RS81}, we proceed to define the {\it scattering matrix} as the fiber-wise restriction to $\lf( L^2(\mathbb{R}^2;\mathbb{C}^2) \ri)_{\lambda} \equiv L^2(\mathbb{S}^{1};\mathbb{C}^2)$ of the scattering operator $\mathrm{S}^{(\Theta)}$, namely,
	\be\label{eq: scattmat}
		\mathrm{S}^{(\Theta)}(\lambda)\, \bm{u}_{\lambda} = F\,\mathrm{S}^{(\Theta)} F^{*}\, \bm{u}_{\lambda}\,.
	\ee
We shall typically refer to the associated integral kernel $\mathrm{S}^{(\mathrm{F})}_{s\,s'}(\lambda;\omega,\omega'\big)$, fulfilling 
	\bdm
		\lf(\mathrm{S}^{(\Theta)}(\lambda) \bm{u}_{\lambda} \ri)_{s}(\omega) = \tx\sum_{s'} \disp\int_{0}^{2\pi}\!\!\! \diff\omega'\, \mathrm{S}^{(\mathrm{F})}_{s\,s'}(\lambda;\omega,\omega'\big) \, u_{\lambda,s'}(\omega')
	\edm
	 for $ s \in \{\uparrow,\downarrow\} $ and $ \omega \in \mathbb{S}^{1} $. Following \cite{IT01,Ru83}, we define the {\it scattering amplitude}
	\be\label{eq: scattampl}
		f^{(\Theta)}_{s\,s'}(\lambda;\omega,\omega') := \lf(\tfrac{2\pi}{i \sqrt{\lambda}}\ri)^{\!\!1/2}\! \lf(\mathrm{S}^{(\Theta)}_{s\,s'}(\lambda;\omega,\omega'\big) - \delta_{s\,s'}\,\delta(\omega - \omega') \ri),
	\ee
and the {\it differential cross section}
	\be\label{eq: dcs}
		\frac{\diff \sigma_{s\,s'}^{(\Theta)}}{\diff \omega}(\lambda,\omega) := \big|f^{(\Theta)}_{s\,s'}(\lambda;\omega,0)\big|^2\,.
	\ee

	\begin{theorem}[Generalized eigenfunctions and scattering matrix of $H_{\mathrm{P}}^{(\mathrm{F})}$]\label{thm: scattmatF}
		\mbox{} \\
		The generalized eigenfunctions of $ H_{\mathrm{P}}^{(\mathrm{F})}$ are
	\be\label{eq: psiF}
		\varphiv^{(\mathrm{F},\pm)}_{(s,\kv)}(\xv) = \tfrac{1}{2\pi}  \sum_{\ell \in \mathbb{Z}} e^{i \ell (\teta-\omega_{\pm}) \pm i \frac{\pi}{2}\,|\ell + \alpha|}\, J_{|\ell + \alpha|}(k\,r) \, \begin{pmatrix} \delta_{s,\uparrow} \\ \delta_{s,\downarrow} \end{pmatrix} ,
	\ee
	where $\omega_{+} = \omega$ and $\omega_{-} = \omega + \pi$. 
		The integral kernel associated to the scattering matrix $\mathrm{S}^{(\mathrm{F})}(\lambda)$ is given by
			\be \begin{aligned}\label{eq: SFexp}
				\mathrm{S}^{(\mathrm{F})}_{s\,s'}(\lambda;\omega,\omega'\big) 
				& =\tfrac{1}{2\pi}\sum_{\ell \in \mathbb{Z}} e^{i \pi (\ell - |\ell + \alpha|) + i \ell(\omega- \omega')}  \, \delta_{s\,s'} \\
				& =\lf[ \cos(\pi \alpha)\, \delta(\omega-\omega') + \tfrac{i}{\pi}\,\sin(\pi \alpha)\;\mathrm{p.v.}\lf( \tfrac{1}{e^{i (\omega-\omega')} - 1}\ri) \ri] \delta_{s\,s'} \, ,
			\end{aligned} \ee
		where $s,s' \in \{\uparrow,\downarrow\}$, $\omega,\omega' \in [0,2\pi)$ and $\mathrm{p.v.}$ indicates the Cauchy principal value. Furthermore, the scattering amplitude is given by
			\beq\label{eq:1}
				f^{(\mathrm{F})}_{s\,s'}(\lambda;\omega,\omega') =  \lf(\tfrac{2\pi}{i \sqrt{\lambda}}\ri)^{\!\!1/2}\! \lf[\! \big(\cos(\pi \alpha)- 1 \big)\, \delta(\omega-\omega') + \tfrac{i}{\pi}\,\sin(\pi \alpha)\;\mathrm{p.v.}\!\lf( \tfrac{1}{e^{i (\omega-\omega')} - 1}\ri)  \ri] \, \delta_{s\,s'} \, ,
			\eeq
			and the differential cross section for $\omega \neq 0$ is
			\beq\label{eq:2}
				\frac{\diff \sigma^{(\mathrm{F})}_{s\,s'}}{\diff \omega}(\lambda,\omega) = 		
		\tfrac{1}{2\pi\sqrt{\lambda}}\, \tfrac{\sin^2(\pi \alpha)}{\sin^2(\omega/2)}\,\delta_{s\,s'}  \,.
			\eeq
	\end{theorem}

\begin{remark}[Pauli and AB Schr\"{o}dinger operators]
\mbox{}	\\
In accordance with the fact that $H_{\mathrm{P}}^{(\mathrm{F})}$ is just the direct sum of two copies of the Friedrichs AB Schr\"odinger Hamiltonian $H_{\mathrm{S}}^{(\mathrm{F})}$, see \eqref{HPFdiag}, the scattering operator and the scattering matrix also coincide with the direct sums of two copies of the analogous quantities related to the scalar case.
\end{remark}		

The analogue of \cref{prop: waveopF} about the scattering for the self-adjoint extensions is the following.

\begin{theorem}[Scattering for $ H_{\mathrm{P}}^{(\Theta)} $]
\label{prop: waveopTheta}
\mbox{}	\\
For any $\Theta \in \mathrm{M}_{4,\mathrm{Herm}}(\mathbb{C})$, the wave operators $\Omega^{(\Theta)}_{\pm} $ exist and are asymptotically complete. Moreover, the scattering operator $\mathrm{S}^{(\Theta)} $ exists and is unitary on $L^2(\mathbb{R}^2;\mathbb{C}^2)$.
\end{theorem}

The spectrum of the self-adjoint extensions $H_{\mathrm{P}}^{(\Theta)}$ is on the other hand much richer than the one of $ H_{\mathrm{P}}^{(\mathrm{F})} $.

\begin{theorem}[Spectrum of $ H_{\mathrm{P}}^{(\Theta)} $]
\label{prop: spectrumTheta}
		\mbox{}	\\
		Let $\Theta \in \mathrm{M}_{4,\mathrm{Herm}}(\mathbb{C})$. Then, the spectrum of the Hamiltonian $H_{\mathrm{P}}^{(\Theta)}$ is 
			\be
				\sigma\big(H_{\mathrm{P}}^{(\Theta)}\big) = \sigma_{\mathrm{ac}}\big(H_{\mathrm{P}}^{(\Theta)}\big) \cup \sigma_{\mathrm{pp}}\big(H_{\mathrm{P}}^{(\Theta)}\big)\,,
			\ee
		where
			\be	
				\sigma_{\mathrm{ac}}\big(H_{\mathrm{P}}^{(\Theta)}\big) = \R^{+} , \qquad
				 \big\{\! -\mu \in \mathbb{R}^{-}\;\big|\, \det\big[\Lambda(-\mu) + \Theta \big] = 0 \big\} \subseteq \sigma_{\mathrm{pp}}\big(H_{\mathrm{P}}^{(\Theta)}\big)\,,
			\ee
		and, in particular, $\sigma_{sc}\big(H_{\mathrm{P}}^{(\Theta)}\big) = \varnothing$. 
		Furthermore, the eigenfunction associated to any negative eigenvalue $- \mu \in \sigma_{\mathrm{pp}}\big(H_{\mathrm{P}}^{(\Theta)}\big)$ is given by $\gmv(-\mu) \qv$ with $ \qv \in \ker \big[\Lambda(-\mu) + \Theta \big]$.	
	\end{theorem}

	\begin{remark}[An extension with negative point spectrum] 
	\mbox{}	\\
	As an example, let us discuss the occurrence of negative eigenvalues for the self-adjoin extension with $\Theta = 0$. It can be checked by direct inspection that the condition $\det\big[\Lambda(-\mu)  \big] = 0$, $\mu \in \mathbb{R}^{+}$, is fulfilled if and only if $\mu = 1$ (see the explicit expression \eqref{eq:Lzexp} reported in the forthcoming \cref{lemma:GLaexp} for $\Lambda(z)$). Moreover, $\Lambda(-1) = 0$ so that $\ker[\Lambda(-1)] =  \CC^4$. Accordingly, $-1 \in \sigma_{\mathrm{pp}}\big(H_{\mathrm{P}}^{(0)}\big)$ is a fourthly degenerate eigenvalue of the Hamiltonian and the associated eigenspace is spanned by
	\be
	\giv_{1,s}^{(\ell)} = K_{|\ell+\alpha|}(r)\,\tfrac{e^{i \ell \teta}}{\sqrt{2\pi}} \begin{pmatrix} \delta_{s,\uparrow} \\ \delta_{s,\downarrow} \end{pmatrix} , \qquad s \in \{\uparrow,\downarrow\},\, \ell \in \{0,-1\}\,,
	\ee
	see \eqref{eq:GGzexp} below, together with \eqref{G1234} and \eqref{eq:g}.
	\end{remark}

	\begin{remark}[Embedded eigenvalues]
	\mbox{}	\\
	We do not discuss here the presence of eigenvalues embedded in the continuous spectrum. These would be exactly the exceptional points forming the set $\mathrm{e}_{+}^{(\Theta)}$ identified next in \cref{lemma:LAPTheta}. Such a question has been investigated for regular magnetic fields in \cite{AHK23,HH23}.
	\end{remark}
	
	We now proceed to characterize the (incoming and outgoing) generalized eigenfunctions $\varphiv^{(\Theta,\pm)}_{(s,\kv)}$  related to $\sigma_{\mathrm{ac}}\big(H_{\mathrm{P}}^{(\Theta)}\big)$. The result below also allows to compute explicitly the scattering matrix and the differential cross-section but we omit the details for the sake of brevity.
	
	\begin{proposition}[Generalized eigenfunctions of $H_{\mathrm{P}}^{(\Theta)}$]\label{thm: scattmatT}
		\mbox{} \\
		The generalized eigenfunctions of $ H_{\mathrm{P}}^{(\Theta)}$ have the form
		\be\label{eq: eigenfunctTheta}
				\varphiv^{(\Theta,\pm)}_{(s,\kv)} = \varphiv^{(\mathrm{F},\pm)}_{(s,\kv)} + \gmv_{\pm}(k^2)\big[\Lambda_{\pm}(k^2) + \Theta\big]^{-1} \tau \varphiv^{(\mathrm{F},\pm)}_{(s,\kv)}\,,
			\ee
where $\varphiv^{(\mathrm{F},\pm)}_{(s,\kv)}$ are the Friedrichs eigenfunctions \eqref{eq: psiF}, while $\Lambda_{\pm}(\lambda)$ and $\gmv_{\pm}(\lambda)$ are defined, respectively, as  
	\be
			\Lambda_{\pm}(\lambda):= \lim_{\varepsilon \to 0^{+}}\Lambda(\lambda \pm i \varepsilon) \,,  \qquad
			\gmv_{\pm}(\lambda) := \lim_{\varepsilon \to 0^{+}} \gmv(\lambda \pm i \varepsilon)\,.
	\ee	
Furthermore, the following asymptotics holds as $ r \to + \infty $:
	\beq\label{eq:asy}
		\big(\varphiv^{(\Theta,\pm)}_{(s,\kv)}\big)_{\!s'} (\xv)
			= \tfrac{1}{2\pi}\,e^{i \kv \cdot \xv} \,  \delta_{s\,s'}+ \tfrac{1}{r^{1/2}} f^{(\Theta,\pm)}_{(s,s'),\kv}\,e^{\mp i k\,r} + \mathcal{O}\lf(\tfrac{1}{r^{3/2}}\ri) ,
	\eeq
	where
		\bml{
		f^{(\Theta,\pm)}_{(s,s'),\kv}
					:=  \tfrac{e^{\pm i \frac{\pi}{4}}}{(2\pi)^{3/2} \sqrt{k}} \sum_{\ell \in \mathbb{Z}} \lf(e^{\pm i \pi |\ell + \alpha|} - e^{\pm i \pi |\ell|}\ri) e^{i \ell (\teta - \omega_{\pm})} \,\delta_{s\,s'}\\
				+ \tfrac{i \pi \,e^{\pm i \frac{\pi}{4}}}{(2\pi)^{3/2} \sqrt{k}} \sum_{\ell,\ell' \in \{0,-1\}}\! \Big\{\big[\Lambda_{\pm}(k^2) + \Theta\big]^{-1}\Big\}_{s'\!,s}^{\ell'\!,\ell}\, e^{i (\ell' \teta - \ell \omega_{\pm})}\;(\pm i k)^{|\ell + \alpha| + |\ell'+\alpha|}\,.
	}
	\end{proposition}

	Zero-energy resonances are discussed in the next result.
		
	\begin{proposition}[Zero-energy resonances for $ H_{\mathrm{P}}^{(\mathrm{F})} $ and $ H_{\mathrm{P}}^{(\Theta)} $]
	\label{prop: zeroenresTheta}
	\mbox{}	\\
		The Friedrichs Hamiltonian $H_{\mathrm{P}}^{(\mathrm{F})}$ has no zero-energy resonances. On the other hand, for any $\Theta \in \mathrm{M}_{4,\mathrm{Herm}}(\mathbb{C})$, the Hamiltonian $H_{\mathrm{P}}^{(\Theta)}$ has zero-energy resonances if and only if $\Lambda(0) + \Theta$ is singular. More precisely, any zero-energy resonance $ \psiv_0 $ has the form
		\be
			\psiv_0 = \begin{pmatrix}	\psi_{\uparrow}	\\ \psi_{\downarrow}	\end{pmatrix}, \qquad
			\psi_{s}(r,\teta) = \sum_{\ell \in \{0,-1\}} q_{s}^{(\ell)}\,\tfrac{2^{|\ell+\alpha| - 1} \Gamma(|\ell + \alpha|)}{r^{|\ell + \alpha|}}\, \tfrac{e^{i \ell \teta}}{\sqrt{2\pi}},
		\ee
		with $ \qv \in \ker \big[\Lambda(0) + \Theta\big]$.
	\end{proposition}

\begin{remark}[Alternative parametrization]\label{rmk:Paulires}
	\mbox{}	\\
	Making reference to the quadratic form parametrization $H_{\mathrm{P}}^{(\bb)}$ of the Pauli Hamiltonian, using the bijection in \cref{prop:quadKrein} and noting that $L(0) = 0$ (see \eqref{eq:L}), we get
	\bdm
		\Lambda(0) + \Theta = \bb\,.
	\edm
	Then, according to \cref{prop: zeroenresTheta}, $H_{\mathrm{P}}^{(\bb)}$ possesses zero-energy resonances whenever $\det \bb = 0$.
\end{remark}

\subsection{Symmetries and connection with the Dirac operator}\label{sec:symdir}

We investigate here the symmetries of the Pauli operator and of its self-adjoint realizations, as well as the connection with its ``square root'', \emph{i.e.}, the {\it Dirac operator}.

Let then $\WW : L^2(\mathbb{R}^2; \mathbb{C}^2) \to L^2(\mathbb{R}^2; \mathbb{C}^2)$ be any unitary or anti-unitary operator. We say that $\WW$ is a symmetry of the differential operator $H_{\mathrm{P}} \equiv H_{\mathrm{P}}(\mathbf{A})$ if
\begin{equation}
\WW\, H_{\mathrm{P}}(\mathbf{A})\, \WW^{-1} = H_{\mathrm{P}}(\widetilde{\mathbf{A}})\,,
\end{equation}
for some $\widetilde{\mathbf{A}}$ such that, in distributional sense
\begin{equation}\label{eq:Binv}
\widetilde{b} := \mathrm{curl} \: \widetilde{\aav} = 2\pi \alpha\,\delta_{\mathbf0}\,.
\end{equation}
Let us denote by $\sigma_3$ the third Pauli matrix, namely, 
	$$
		\sigma_3=\begin{pmatrix}1 & 0 \\ 0 &-1 \end{pmatrix}. 
	$$
In view of the algebraic properties reported in \cref{app:2Dsym} for Pauli operators with smooth vector potentials and since the differential operator of interest here is invariant under rotations around $\mathbf{x} = \mathbf{0}$, the only admissible symmetries of $H_{\mathrm{P}} $ are described by the following operators:

\begin{enumerate}[i)]
\item {\sl Linear transformations:} 
\begin{equation}\label{eq:UST}
	\UU(S,T) : L^2(\mathbb{R}^2;\mathbb{C}^2) \to L^2(\mathbb{R}^2;\mathbb{C}^2)\,, \qquad (\UU \psiv)(\mathbf{x}) = S(\mathbf{x})\, \psiv(T^{-1} \mathbf{x})\,,
\end{equation}
where (see \eqref{eq:TSO2})
\begin{equation}\label{eq:tso2}
T \in SO(2,\mathbb{R}), \quad S = e^{-i \eta_0 \one - i \eta_3 \sigma_3}, \qquad \mbox{for some } \eta_0 \in C^{1}(\mathbb{R}^2), \eta_3 \in \mathbb{R} \,.
\end{equation}

\item {\sl Anti-linear transformations:} 
\begin{equation}\label{eq:VST}
	\VV(S,T) : L^2(\mathbb{R}^2;\mathbb{C}^2) \to L^2(\mathbb{R}^2;\mathbb{C}^2)\,, \qquad (\VV \psiv)(\mathbf{x}) = S(\mathbf{x})\, \psiv^*(T^{-1} \mathbf{x})\,,
\end{equation}
where (see \eqref{eq:TO2anti})
\begin{equation}
T \in O(2,\mathbb{R})  \setminus SO(2,\mathbb{R}), \quad S = e^{-i \eta_0 \one - i \eta_3 \sigma_3}, \qquad \mbox{for some } \eta_0 \in C^{1}(\mathbb{R}^2), \eta_3 \in \mathbb{R}\,.
\end{equation}
\end{enumerate} 
By the above arguments, we deduce that the full symmetry group is
$$ 
U^{\mathrm{loc}}_{\mathrm{em}}(1) \times U_{\mathrm{ax}}(1) \times U_{\mathrm{rot}}(1) \times CP\,.
$$
Among the symmetries described above, there is the very well-known $U^{\mathrm {loc}}_{\mathrm {em}}(1)$-\emph{gauge symmetry} $\aav\mapsto\aav+\nabla\eta_0$ typical of the electromagnetic field, which is recovered when $T=\one$ and $\eta_3=0$ in \eqref{eq:tso2}. The two remaining continuous symmetries are given by the axial gauge and the rotational symmetries of the model. Lastly, the discrete CP-symmetry is the charge conjugation combined with parity.

We aim at investigating the extensions that are left invariant by the above symmetries \eqref{eq:UST} and \eqref{eq:VST}. To that purpose, we keep track of the dependence on the magnetic potential and denote $H_{\mathrm{P}}^{(\bb)} \equiv H_{\mathrm{P}}^{(\bb)}(\mathbf A)$. However, since we are interested in Pauli Hamiltonians with {\it fixed magnetic potential}, we only consider \emph{global} gauge transformations, that is, we take $\eta_0\in\R$ in  \eqref{eq:UST} and \eqref{eq:VST}. Notice that this means that the magnetic potential is unchanged, while we may act on the spinor in a non-trivial way.
%
This corresponds to determining the matrices $\bb$ in \eqref{HBdom} such that
\begin{equation}\label{eq:invariance}
\WW H_{\mathrm{P}}^{(\bb)}(\mathbf A)\,\WW^{-1}= H_{\mathrm{P}}^{(\bb)}(\mathbf A)\,,
\end{equation}
or, equivalently,
\begin{eqnarray}\label{eq:checkinv 1}
&& \WW\,\dom\big(H_{\mathrm{P}}^{(\bb)}(\mathbf A)\big)= \dom\big(H_{\mathrm{P}}^{(\bb)}(\mathbf A)\big)\,,	\\ && \WW\big(H_{\mathrm{P}}^{(\bb)}(\mathbf A)+\lambda^2\big)\, \WW^{-1}\widetilde{\psiv}=\big(H_{\mathrm{P}}^{(\bb)}(\mathbf A)+\lambda^2\big)\widetilde{\psiv}\,,\qquad\forall\, \widetilde{\psiv} \in \WW\,\dom\big(H_{\mathrm{P}}^{(\bb)}(\mathbf A)\big)\,, \label{eq:checkinv 2}
\end{eqnarray}
first for all the transformations $\WW = \UU$ of the form \eqref{eq:UST} and then for those $\WW = \VV$ of the form \eqref{eq:VST}. Here and in what follows, we always assume $\lambda>0$.

\begin{proposition}[Symmetries]
	\label{pro: symmetries}
	\mbox{}	\\
Let $\eta_0\in\R$. Then, transformations of the form \eqref{eq:UST} or \eqref{eq:VST} are symmetries of the self-adjoint extension $H_{\mathrm{P}}^{(\bb)}$ for any $ \eta_3 \in \R $, if and only if $\bb$ is a diagonal matrix. 
\end{proposition}

Concerning the connection with the Dirac operator, we first have to define it in presence of an AB magnetic potential, mostly referring to \cite{Pe06}. We start from the symmetric Dirac operator with domain $C^{\infty}_c\big(\mathbb{R}^2 \setminus \{0\}; \mathbb{C}^2\big)$, acting as
	\begin{equation}\label{eq:Dirac}
		H_{\mathrm{D}}\psiv= \bm{\sigma} \cdot(- i \nabla+\mathbf{A})\psiv \,,
	\end{equation}
where $\mathbf{A} = \mathbf{A}(x)$ is the AB potential \eqref{eq:AB}. Notice that, on $C^{\infty}_c(\mathbb{R}^2 \setminus \{0\};\mathbb{C}^2)$, there holds
	\begin{equation}\label{eq:D^2=H}
		H_{\mathrm{P}}=H_{\mathrm{D}}^2\,.
	\end{equation}
We are going to determine the self-adjoint extensions of $H_{\mathrm{D}}$ using the von Neumann approach. To this avail, we first have to take the closure of $C^{\infty}_{{c}}\big(\mathbb{R}^2 \setminus \{0\}; \mathbb{C}^2\big)$ w.r.t. the graph norm $ \psiv\mapsto \sqrt{\Vert \psiv\Vert^2+\Vert H_{\mathrm{D}}\psiv\Vert^2}\, $. Notice that, if $\psiv\in C^{\infty}_c(\mathbb{R}^2 \setminus \{0\};\mathbb{C}^2)$, it is immediate to see that
\[
\Vert H_{\mathrm{D}}\psiv\Vert^2= \meanlrlr{\psiv}{H_{\mathrm{P}}}{\psiv} = Q^{(\mathrm{F})}_{\mathrm{P}}(\psiv)\,,
\]
so that
\beq
\dom\big(\overline{H}_{\mathrm{D}}\big) = \dom\big[Q^{(\mathrm{F})}_{\mathrm{P}}\big]\,.
\eeq
Moreover, the operator $\overline{H}_{\mathrm{D}}$ has defect indices $(1,1)$, and the defect spaces $\chi_{\pm}:=\ker (\overline{H}_{\mathrm{D}}^*\mp i)$ are spanned by the functions
\begin{equation}\label{eq:Ddefect}
\xiv_{\pm}(r, \teta)=\begin{pmatrix}K_{1-\alpha}(r)\,e^{-i\teta} \\ \pm\, K_\alpha(r) \end{pmatrix} .
\end{equation}
We thus have a one-parameter family of self-adjoint extensions of $H_{\mathrm{D}}$, parametrized by the unitary operators mapping $\chi_{+}$ into $\chi_{-}$.

\begin{proposition}[Self-adjoint extensions of $ H_{\mathrm{D}} $]
\label{thm:Ddom}
\mbox{}	\\
For any $\alpha\in(0,1)$, the symmetric operator $H_{\mathrm{D}}$ admits a one-parameter family of self-adjoint extensions $H_{\mathrm{D}}^{(\gamma)}$, parametrized by $\gamma\in[0,2\pi)$ and given by
\begin{equation}\label{eq:Dextdom}
\dom\big(H_{\mathrm{D}}^{(\gamma)}\big) = \lf\{ \psiv =\phiv +\mu \lf(\xiv_+ +e^{i\gamma}\xiv_- \ri)\, \Big| \, \phiv\in \dom\big[Q^{(\mathrm{F})}_{\mathrm{P}}\big]\,,\; \mu \in \mathbb{C} \ri\} ,
\end{equation}
\begin{equation}\label{eq:Daction}
H_{\mathrm{D}}^{(\gamma)} \psiv=H_{\mathrm{D}}\phiv+i\mu(\xiv_+ -i\xiv_-)\,.
\end{equation}
\end{proposition}

Alternatively, as illustrated in \cite{Pe06}, the domain of the extensions $H_{\mathrm{D}}^{(\gamma)}$ can be described in terms of boundary conditions at the origin, as follows: given $\gamma\in[0,2\pi)$, one can define the linear functionals $c^{\uparrow,\downarrow}_{-\alpha}$, $c^{\uparrow,\downarrow}_{\alpha-1}$ on $\dom(H_{\mathrm{D}}^{(\gamma)})$ as 
\begin{gather}\label{eq:Diractraces}
 c^{s}_{-\alpha}(\psiv)=\lim_{r\to0^+}r^\alpha\,\langle \psi_{s}\rangle = \tfrac{1}{2\pi}\lim_{r\to0^+} r^\alpha\int^{2\pi}_0\diff\teta\, \psi_{s}(r,\teta)\,,  \\
 c^{s}_{\alpha-1}(\psiv)=\lim_{r\to0^+}r^{1-\alpha}\,\langle e^{i\teta}\psi_{s}\rangle= \tfrac{1}{2\pi}\lim_{r\to0^+} r^{1-\alpha}\int^{2\pi}_0\diff\teta \, e^{i\teta}\,\psi_{s}(r,\teta)\,,
\end{gather}
for $ s \in \lf\{\uparrow,\downarrow \ri\} $. Observe that, for a given $\psiv\in\dom(H_{\mathrm{D}}^{(\gamma)})$, we have a decomposition $\psiv=\phiv+\mu(\xiv_+ +e^{i\gamma}\xiv_-)$ as in \cref{thm:Ddom}, so that $\phiv$ vanishes at the origin, thus having zero boundary value. The asymptotics of Bessel functions $K_\nu$ at the origin \cite{GR07} then easily give
\begin{gather}
c^\uparrow_{-\alpha}(\psiv)=0\,,\qquad c^\uparrow_{\alpha-1}(\psiv)=\mu\,(1+e^{i\gamma})\, 2^{-\alpha}\,\Gamma(1-\alpha)\,, \label{bcDir1}\\
c^\downarrow_{\alpha-1}(\psiv)=0\,,\qquad c^\downarrow_{-\alpha}(\psiv)=\mu\,(1-e^{i\gamma})\, 2^{-(1-\alpha)}\,\Gamma(\alpha)\,, \label{bcDir2}
\end{gather}
where $\Gamma(z)$ is the Euler gamma function. 
We can then equivalently describe the domain of $H_{\mathrm{D}}^{(\gamma)}$ as
\begin{multline}\label{eq:Ddom2}
\dom(H_{\mathrm{D}}^{(\gamma)})= \lf\{\psiv\in L^2(\mathbb{R}^2; \mathbb{C}^2)\, \Big|\, H_{\mathrm{D}}\psiv \in L^2(\mathbb{R}^2; \mathbb{C}^2)\,,\; c^\uparrow_{\alpha-1}(\psiv) =i\cot(\gamma/2)\, \tfrac{2^{1-2\alpha}\Gamma(1-\alpha)}{\Gamma(\alpha)}\, c^\downarrow_{-\alpha}(\psiv)\,, \ri. \\
 \lf. c^\uparrow_{-\alpha}(\psiv)=0\,, c^\downarrow_{\alpha-1}(\psiv)=0 \ri\} .
\end{multline}

	\begin{remark}[Local singularities of functions in the Dirac domain]
	\label{remfia}
	\mbox{}	\\
	For a generic $ \gamma \in [0,2\pi)$, considering the characterization \eqref{eq:Dextdom} of $\dom\big(H_{\mathrm{D}}^{(\gamma)}\big)$, the asymptotic properties in \eqref{asyF}, the expressions \eqref{eq:Ddefect} of $\xiv_{\pm}$ and the expansions of the Bessel functions for small argument, for any $\psiv \in \dom\big(H_{\mathrm{D}}^{(\gamma)}\big)$ one infers that, as $r \to 0^{+}$,
\begin{equation*}
	\psiv(r,\theta) 
	= \mu \begin{pmatrix} (1 + e^{i \gamma}) \,\frac{2^{-\alpha} \Gamma(1 - \alpha)}{r^{1-\alpha}} \,e^{-i\teta} \vspace{0.15cm} \\ (1 - e^{i \gamma})\, \frac{2^{-(1 - \alpha)} \Gamma(\alpha)}{r^{\alpha}} \end{pmatrix} + o(1)
	= \tilde{\mu} \begin{pmatrix} i \cot(\gamma/2) \,\frac{2^{1-2\alpha} \Gamma(1 - \alpha)}{\Gamma(\alpha)}\, \frac{1}{r^{1-\alpha}} \,e^{-i\teta} \vspace{0.15cm}\\ \frac{1}{r^{\alpha}} \end{pmatrix} + o(1)\,,
\end{equation*}
for some $\tilde{\mu} \in \CC$.
\end{remark}

%
We have already observed that the symmetric operators $H_{\mathrm{P}}$ and $H_{\mathrm{D}}$ defined on $C^\infty_c(\mathbb{R}^2 \setminus \{0\}; \mathbb{C}^2)$, are  such that $H_{\mathrm{D}}^2=H_{\mathrm{P}}$. Now we investigate whether the same property holds for the the respective self-adjoint extensions. More precisely, we aim at determining which extension $H_{\mathrm{P}}^{(\bb)}$ is the square of an extension $H_{\mathrm{D}}^{(\gamma)}$, if any. It turns out that only for specific values of the parameter $\gamma$ and for specific choices of $\bb = \big(\bb_{s s'}^{(\ell \ell')}\big) \in \mathrm{M}_{4,\mathrm{Herm}}(\mathbb{C})$ there actually holds $(H_{\mathrm{D}}^{\gamma})^2 = H^{(\bb)}_{\mathrm{P}}$.

\begin{proposition}[Dirac and Pauli operators]
\label{pro: dirac vs pauli}
\mbox{}	\\
The identity
\begin{equation}\label{eq:D^2=P}
(H_{\mathrm{D}}^{\gamma})^2=H^{(\bb)}_{\mathrm{P}}
\end{equation}
holds if and only if one of the following two alternatives is realized:
\begin{enumerate}[(a)]
\item[(a)] $\gamma=0$, $\bb_{\uparrow \uparrow}^{(0 0)} = \bb_{\downarrow \downarrow}^{(0 0)} = \bb_{\downarrow \downarrow}^{(-1 -1)} = \infty$ and $\bb_{\uparrow \uparrow}^{(-1 -1)} = 0$ ($\bb_{s s'}^{(\ell \ell')}$ arbitrarily chosen otherwise), \emph{i.e.}, $H^{(\bb)}_{\mathrm{P}}$ coincides with the Krein extension in the $(\uparrow,-1)$ channel and with the Friedrichs extension in all the other channels;

\item[(b)] $\gamma=\pi$, $\bb_{\uparrow \uparrow}^{(0 0)} = \bb_{\uparrow \uparrow}^{(-1 -1)} = \bb_{\downarrow \downarrow}^{(-1 -1)} = \infty$ and $\bb_{\downarrow \downarrow}^{(0 0)} = 0$ ($\bb_{s s'}^{(\ell \ell')}$ arbitrarily chosen otherwise), \emph{i.e.}, $H^{(\bb)}_{\mathrm{P}}$ coincides with the Krein extension in the $(\downarrow,0)$ channel and with the Friedrichs extension in all the other channels.
\end{enumerate}

\end{proposition}
\begin{remark}
The Dirac extensions corresponding to $\gamma=0$ and $\gamma=\pi$ are two distinguished ones. They are indeed the only scale covariant realizations, \emph{i.e.}, homogeneous of degree $-1$ under scaling, for generic values of $\alpha \in (0,1)$.\footnote{Namely, for all $\alpha\neq 1/2$. The case $\alpha = 1/2$ is indeed exceptional, since in this instance any Dirac Hamiltonian $H_{\mathrm{D}}^{(\gamma)}$, $\gamma \in [0,2\pi)$, exhibits scale covariance.} This facts has implications at the dynamical level. Indeed, in \cite{CDYZ} the authors remark that taking
\begin{equation}\label{eq:distpar}
\begin{cases}
\gamma=\pi\,,\qquad \mbox{for }\alpha\in(0,1/2]\,, \\
\gamma=0\,,\qquad \mbox{for } \alpha\in(1/2,1)\,,
\end{cases}
\end{equation}
one gets the only self-adjoint realizations $H_{\mathrm{D}}^{(\gamma)}$ whose domain is contained in $H^{1/2}(\R^2\,;\CC^2)$, a fact which allows to identify them as distinguished ones. Notice that here we are referring to our parametrization, which is slightly different from the one adopted in \cite{CDYZ}. Some dispersive estimates are then proven in \cite{CDYZ}, for the choice of parameters \eqref{eq:distpar}. It is also mentioned therein that for other extensions such estimates seem to fail.
\end{remark}

In passing, we point out the following result which has its own interest: unlike Pauli operators with an AB flux, any self-adjoint realization of the Dirac operator with an AB flux admits a zero-energy resonance.
	
	\begin{proposition}[Zero-energy resonances for $ H_{\mathrm{D}}^{\gamma} $]
	\label{prop: zeroenresDirac}
	\mbox{}	\\
		For any $\gamma \in[0,2\pi)$, the Dirac Hamiltonian $H_{\mathrm{D}}^{\gamma}$ has a zero-energy resonance $ \psiv_0 $ given by
		\be \label{eq:resDirac}
			\psiv_0(r,\teta) =
			\begin{cases}
				 \begin{pmatrix}
					i\cot(\gamma/2)\, \tfrac{2^{1-2\alpha}\Gamma(1-\alpha)}{\Gamma(\alpha)}\, e^{-i \teta}\,r^{-(1-\alpha)} \vspace{0.1cm} \\
					r^{-\alpha}
				\end{pmatrix}\,	&	\mbox{for } \gamma \neq 0,	\\	
				\begin{pmatrix}
					e^{-i \teta}\, r^{-(1-\alpha)} \\
					0
				\end{pmatrix},	&	\mbox{for } \gamma = 0.
			\end{cases}
	\ee
	\end{proposition}
	\begin{remark}
	The above result is coherent with \cref{prop: zeroenresTheta}, concerning the resonances for Pauli AB Hamiltonians  $H^{(\beta)}_\mathrm P$(see also \cref{rmk:Paulires}). By direct inspection, it is readily seen that the latter coincide with those of $H^\gamma_{\mathrm D}$ in \eqref{eq:resDirac}, for the distinguished choices of the parameters $\gamma\in[0,2\pi)$ and $\beta \in \mathrm{M}_{4,\mathrm{Herm}}(\mathbb{C})$ described in the preceding \cref{pro: dirac vs pauli}, so that $(H^\gamma_{\mathrm{D}})^2=H^{(\beta)}_{\mathrm P}$.
	\end{remark}

\section{Proofs}
This section is devoted to the proofs of the results previously stated.
The first result we prove is the classification of self-adjoint extensions: we start by addressing the quadratic forms \eqref{QB2} in \cref{sec:quadforms} and then apply Krein's theory \cref{sec:krein} to get an alternative parametrization of the family, through the expression of the resolvents. The latter are also going to be used to investigate spectral and scattering properties in \cref{sec:specscatt}. Finally, in \cref{sec:sym} and \cref{sec:compdirac} we study the symmetry properties of the extensions and the relation with the self-adjoint realization of the Dirac operator with an AB flux, respectively.

\subsection{Quadratic forms}\label{sec:quadforms}
By studying the quadratic forms introduced in \eqref{QB2}, we prove points \emph{i)} and \emph{ii)} in \cref{thm:QBHB}. Preliminarily, we observe that (see \cite[Eq. 6.521.3]{GR07})
	\be\label{eq:GjL2}
		\big\|g^{(\ell)}_{\lambda}\big\|_{2}^{2} = \tfrac{\pi\, |\ell+\alpha|}{ 2\sin(\pi \alpha)}\,\lambda^{2|\ell+\alpha| - 2}\,,
	\ee
and 
	\be
		g^{(\ell)}_{\lambda}\!(r,\teta) =  \left[\tfrac{\Gamma( |\ell+\alpha| ) }{ 2^{1-|\ell+\alpha|}}\,{1 \over r^{|\ell+\alpha|}} + \tfrac{\Gamma(\!-|\ell+\alpha|) }{ 2^{1 + |\ell+\alpha|}}\, \lambda^{2|\ell+\alpha|}\,r^{|\ell+\alpha|} + \mathcal{O}\big(r^{2-|\ell+\alpha|}\big)\right] \tfrac{e^{i \ell \teta}}{ \sqrt{2\pi}}\,, \quad \mbox{for\, $r \to 0^{+}$}. \label{eq:G01asy}
	\ee
For later reference, let us also remark that \eqref{eq:GjL2} and \eqref{G1234} imply
	\be\label{GmGn}
		\braketl{ \giv_{\lambda,s}^{(\ell)}\,}{ \giv_{\lambda,s'}^{(\ell')}} = \tfrac{\pi \,|\ell+\alpha|}{ 2 \sin(\pi \alpha)}\,\lambda^{2|\ell+\alpha|-2}\, \delta_{s\,s'}\,\delta_{\ell\,\ell'}\,.
	\ee

\begin{proof}[Proof of Theorem \ref{thm:QBHB}]
\mbox{} \\
\textsl{i)} Let us first prove that the form $Q_{\mathrm{P}}^{(\bb)}$ is independent of the spectral parameter $\lambda > 0$. To this avail, let $\lambda_1 \neq \lambda_2$ and consider the two alternative representations 
\[
\psiv = \phiv_{\lambda_1} + \tx{\sum_{s,\ell} } q_s^{(\ell)}\, \giv_{\lambda_1,s}^{(\ell)}\,, \quad  \psiv = \phiv_{\lambda_2} + \tx{\sum_{s,\ell}}\, q_s^{(\ell)} \giv_{\lambda_2,s}^{(\ell)}\,.
\] 
In particular, noting that $\giv_{\lambda_1,s}^{(\ell)} - \giv_{\lambda_2,s}^{(\ell)} \in \dom\big[Q_{\mathrm{P}}^{(\mathrm{F})}\big]$ for any $s \in \{\uparrow,\downarrow\}$ and $\ell \in \{0,-1\}$, we have 
\[
\phiv_{\lambda_1} = \phiv_{\lambda_2} + \tx{\sum_{s,\ell}}\, q_s^{(\ell)} \big( \giv_{\lambda_2,s}^{(\ell)} - \giv_{\lambda_1,s}^{(\ell)}\big)\,.
\]
 Recalling the explicit expression \eqref{QB2} of the quadratic form, keeping in mind that $\giv_{\lambda_1,s}^{(\ell)},\giv_{\lambda_2,s}^{(\ell)}$ are defect functions for the Pauli operator and using the identity \eqref{GmGn}, via some integrations by parts we infer
	\bmln{
		Q_{\mathrm{P}}^{(\bb)}\!\lf[\phiv_{\lambda_1} + \tx{\sum_{s,\ell}}\, q_s^{(\ell)} \giv_{\lambda_1,s}^{(\ell)} \ri] - Q_{\mathrm{P}}^{(\bb)}\!\lf[\phiv_{\lambda_2} + \tx{\sum_{s,\ell}}\, q_s^{(\ell)} \giv_{\lambda_2,s}^{(\ell)} \ri] \\
		= - \,2 \Re \sum_{s,\ell} q_s^{(\ell)}\! \lim_{r \to 0^{+}}\! \int_{\partial B_r(\mathbf{0})}\hspace{-0.55cm} \diff\Sigma_r\; \phiv_{\lambda_2}^{*} \cdot \partial_r \big( \giv_{\lambda_2,s}^{(\ell)} 
		- \giv_{\lambda_1,s}^{(\ell)}\big)\\
			- \sum_{s,\ell} \,\big|q_s^{(\ell)}\big|^2 \lim_{r \to 0^{+}}\! \int_{\partial B_r(\mathbf{0})}\hspace{-0.55cm} \diff\Sigma_r\, \big( \giv_{\lambda_2,s}^{(\ell)}\! - \giv_{\lambda_1,s}^{(\ell)}\big)^{\!*}\! \cdot \partial_r \big( \giv_{\lambda_2,s}^{(\ell)}\! - \giv_{\lambda_1,s}^{(\ell)}\big) \\
			+ \sum_{s,\ell}\, \big|q_s^{(\ell)}\big|^2 \bigg[ \tfrac{\pi }{ 2 \sin(\pi \alpha)} \lf(\lambda_1^{2|\ell+\alpha|} - \lambda_2^{2|\ell+\alpha|}\ri)
				+ \lim_{r \to 0^{+}} \int_{\partial B_r(\mathbf{0})}\hspace{-0.4cm} \diff\Sigma_r\, \Big( \big(\giv_{\lambda_2,s}^{(\ell)}\big)^{\!*}\! \cdot \partial_r \giv_{\lambda_1,s}^{(\ell)} - \big(\partial_r \giv_{\lambda_2,s}^{(\ell)}\big)^{\!*} \!\cdot \giv_{\lambda_1,s}^{(\ell)} \Big)\bigg] \,.
	}
Using \eqref{asyF} for $\phiv_{\lambda_2} \!\in\! \dom\big[Q_{\mathrm{P}}^{(\mathrm{F})}\big]$ and the asymptotic expansion \eqref{eq:G01asy}, by Cauchy-Schwarz inequality, we get
	\begin{gather*}
		\lf| \int_{\partial B_r(\mathbf{0})}\hspace{-0.4cm} \diff\Sigma_r\; \phiv_{\lambda_2}^{*} \cdot \partial_r \big( \giv_{\lambda_1,s}^{(\ell)}\! - \giv_{\lambda_2,s}^{(\ell)}\big) \ri| 
		\leqslant C\, r^{|\ell+\alpha|}\, \sqrt{\avg{|\phiv_{\lambda_2|^2}}}
		\;\xrightarrow[r \to 0^{+}]{} \;0\,; \\
		\lf|\int_{\partial B_r(\mathbf{0})}\hspace{-0.4cm} \diff\Sigma_r\, \big( \giv_{\lambda_1,s}^{(\ell)}\! - \giv_{\lambda_2,s}^{(\ell)}\big)^{\!*}\! \cdot \partial_r \big( \giv_{\lambda_1,s}^{(\ell)}\! - \giv_{\lambda_2,s}^{(\ell)}\big)\ri|
		\leqslant C\, r^{2|\ell+\alpha|} \;\xrightarrow[r \to 0^{+}]{}\; 0\,.
	\end{gather*}
On the other hand, using again \eqref{eq:G01asy}, we infer
	\bdm
		\big(\giv_{\lambda_2,s}^{(\ell)}\big)^{\!*}\! \cdot \partial_r \giv_{\lambda_1,s}^{(\ell)} - \big(\partial_r \giv_{\lambda_2,s}^{(\ell)}\big)^{\!*}\! \cdot \giv_{\lambda_1,s}^{(\ell)}
		= \tfrac{\lambda_{2}^{2 |\ell+\alpha|} - \lambda_{1}^{2 |\ell+\alpha|} }{ 4 \sin(\pi \alpha)\, r} + \mathcal{O}\big(r^{1-2|\ell+\alpha|}\big)\,,
	\edm
which entails
	\bdm
		\lim_{r \to 0^{+}} \int_{\partial B_r(\mathbf{0})}\hspace{-0.4cm} \diff\Sigma_r\, \Big(\big(\giv_{\lambda_2,s}^{(\ell)}\big)^{\!*} \cdot \partial_r \giv_{\lambda_1,s}^{(\ell)} - \big(\partial_r \giv_{\lambda_2,s}^{(\ell)}\big)^{\!*} \cdot \giv_{\lambda_1,s}^{(\ell)} \Big)
		= \tfrac{\pi}{2 \sin(\pi \alpha)} \lf( \lambda_{2}^{2 |\ell+\alpha|} - \lambda_{1}^{2 |\ell+\alpha|}\ri) .
	\edm
Summing up, we obtain 
	\bdm
		Q_{\mathrm{P}}^{(\bb)}\!\lf[\phiv_{\lambda_1} + \tx{\sum_{s,\ell}}\, q_s^{(\ell)} \giv_{\lambda_1,s}^{(\ell)} \ri] - Q_{\mathrm{P}}^{(\bb)}\!\lf[\phiv_{\lambda_2} + \tx{\sum_{s,\ell}}\, q_s^{(\ell)} \giv_{\lambda_2,s}^{(\ell)} \ri] = 0\,,
	\edm
which proves that the form is independent of the spectral parameter.

Next, let us point out that,  for $\lambda > 0$ large enough, $L(\lambda) + \bb\geq c\,\lambda^{\min\{\alpha,1-\alpha\}}\, \one$, with $c>0$ and $\one$ being the identity operator, so that 
	\be\label{QBlower}
		Q_{\mathrm{P}}^{(\bb)}[\psiv] + \lambda^2 \|\psiv\|_{2}^{2} \geqslant Q_{\mathrm{P}}^{(\mathrm{F})}[\phiv_{\lambda}] + \lambda^2\,\|\phiv_{\lambda}\|_{2}^2 + c\,\lambda^{\min\{\alpha,1-\alpha\}}\, |\qv|^2 \,.
	\ee
Since $Q^{(\mathrm{F})}_{\mathrm{P}}[\phiv_{\lambda}]$ is non-negative, the above relation ensures that $Q_{\mathrm{P}}^{(\bb)}$ is bounded from below. The closedness of $Q_{\mathrm{P}}^{(\bb)}$ can be then deduced by standard arguments starting from \eqref{QBlower}, as in \cite[Thm 2.4]{CO18} (see also \cite{Te90}).
\vspace{0.1cm}	

\textsl{ii)} Let us consider the sesquilinear form defined by polarization of \eqref{QB2}, for $\psiv_1 = \phiv_{1,\lambda} + \tx{\sum_{s,\ell}}\,q_{1,s}^{(\ell)}\,\giv_{\lambda,s}^{(\ell)}$ and $\psiv_2 = \phiv_{2,\lambda} + \tx{\sum_{s,\ell}}\,q_{2,s}^{(\ell)}\,\giv_{\lambda,s}^{(\ell)}$,
	\[
		Q_{\mathrm{P}}^{(\bb)}[\psiv_1,\psiv_2] 
		= Q_{\mathrm{P}}^{(\mathrm{F})}[\phiv_{1,\lambda},\phiv_{2,\lambda}] - \lambda^2\! \braketl{\psiv_{1}}{\psiv_2} + \lambda^2\! \braketl{\phiv_{1,\lambda}}{\phiv_{2,\lambda}} + \qv^*_1\cdot  (L(\lambda)+\bb) \qv_2 \,.
	\]
Fixing $\qv_{1} = \mathbf{0}$, we get $ Q_{\mathrm{P}}^{(\bb)}[\phiv_{1,\lambda},\psiv_2] = Q_{\mathrm{P}}^{(\mathrm{F})}[\phiv_{1,\lambda},\phiv_{2,\lambda}] - \lambda^2 \braket{\phiv_{1,\lambda}}{\psiv_2 - \phiv_{2,\lambda}}$. Therefore, in order to have that $ Q^{(\bb)}[\phiv_{1,\lambda},\psiv_2] = \braket{\phiv_{1,\lambda}}{\bm{\chi}}$ for some $\bm{\chi} =: H_{\mathrm{P}}^{(\bb)} \psiv_2 \in L^2(\mathbb{R}^2;\mathbb{C}^2)$, we must assume that $\phiv_{2,\lambda} \in \dom\big(H_{\mathrm{P}}^{(\mathrm{F})}\big)$ and set
	\bdm
		\bm{\chi} = H_{\mathrm{P}}^{(\mathrm{F})} \phiv_{2,\lambda} - \lambda^2\, \tx{\sum_{s,\ell}}\, q_{2,s}^{(\ell)} \giv_{\lambda,s}^{(\ell)} \,\in\, L^2(\mathbb{R}^2;\mathbb{C}^2)\,.
	\edm
Taking this into account and imposing $Q_{\mathrm{P}}^{(\bb)}[\psiv_1,\psiv_2] = \big\langle \psi_1 \,\big|\, H_{\mathrm{P}}^{(\bb)} \psi_2 \big\rangle$ for $\qv_{1} \neq 0$, via integration by parts, we deduce
	\bmln{
		\qv^*_1\cdot  (L(\lambda)+\bb) \qv_2 = {\tx \sum_{s,\ell}} \big(q_{1,s}^{(\ell)}\big)^{*}\! \braketl{ \giv_{\lambda,s}^{(\ell)}}{\big(H_{\mathrm{P}}^{(\mathrm{F})}\! + \lambda^2\big) \phiv_{2,\lambda}} \\
		= {\tx \sum_{s,\ell}} \big(q_{1,s}^{(\ell)}\big)^{*} \lim_{r \to 0^{+}} \int_{\partial B_r(\mathbf{0})} \hspace{-0.3cm}\diff\Sigma_r\, \Big(\big(\giv_{\lambda,s}^{(\ell)}\big)^{\!*}\! \cdot \partial_r \phiv_{2,\lambda} - \big(\partial_r \giv_{\lambda,s}^{(\ell)}\big)^{\!*}\! \cdot \phiv_{2,\lambda} \Big)\,,
	}
which, in view of \eqref{G1234}, \eqref{eq:G01asy} and of the arbitrariness of $\qv_1$, ultimately accounts for the boundary condition in \eqref{HBdom}.
\end{proof}

\subsection{Krein's theory}\label{sec:krein}
For later convenience, we first point out the following inequalities, valid for any $r,r'>0$ and $\Im \sqrt{z}>0$ (see \cite[Eqs. 10.32.2 and 10.32.8]{OLBC10}):
	\begin{gather}
		\big| I_{\vert \ell+\alpha\vert}\big(\!-i\sqrt{z}\,(r \land r')\big) \big| \leqslant \left(\!\tfrac{\sqrt{|z|}}{\Im \sqrt{z}}\right)^{\!\vert \ell+\alpha\vert} I_{\vert \ell+\alpha\vert}\big(\Im \sqrt{z}\,(r \land r')\big)\,; \label{eq: estI}\\
		\big| K_{\vert \ell+\alpha\vert}\big(\!-i\sqrt{z}\,(r \lor r')\big) \big| \leqslant K_{\vert \ell+\alpha\vert}\big(\Im \sqrt{z}\,(r \lor r')\big)\,. \label{eq: estK}
	\end{gather}
We also recall the definition of $\Lambda(z)$ in \eqref{eq: defLambda}, and notice that general arguments (see \cite[Lemma 2.2]{Po01} and following discussion) ensure that it is well-defined and also entail the following identities, for all $z,w \in \mathbb{C} \setminus \R^{+}$:
	\begin{equation}
		\Lambda(z) - \Lambda(w) = (w-z) \,\widecheck{\mathcal{G}}(z)\, \gmv(w)\,, \qquad 
		\big(\Lambda(z)\big)^{*} \!= \Lambda(z^*)\,. \label{eq: Lambdaprop}
	\end{equation}	
\begin{proof}[Proof of \cref{thm:HTheta}]
Recalling once more that $H_{N}^{(\mathrm{F})}$ is positive semi-definite, the result follows from \cite[Thm. 2.1]{Po01} and \cite[Thm. 3.1 and Cor. 3.2]{Po08}.
\end{proof}
Once the form of the resolvent operator for each of the self-adjoint extensions is known, as given by Krein's theory, it is natural to investigate the relations between those extensions and those obtained in \eqref{HBdom} and \eqref{HBaction}. We start by proving the following.
\begin{lemma}\label{lemma:GLaexp}
For any $z \in \mathbb{C} \setminus \R^{+}$ and for all $\qv \in \mathbb{C}^4$, there holds
	\be \label{eq:GGzexp}
	\gmv(z)\,\qv = {\tx \sum_{s,\ell}}\, \giv_{-i \sqrt{z}\,,\,s}^{(\ell)}\,q_{s}^{(\ell)} .
	\ee
Moreover, for all $s,s' \!\in\! \{\uparrow,\downarrow\}$, $\ell,\ell' \!\in\! \{0,-1\}$ there holds
	\be \label{eq:Lzexp}
		\Lambda_{s\,s'}^{(\ell\,\ell')}(z) = \tfrac{\pi}{2 \sin(\pi \alpha)} \left[\big(\!-i \sqrt{z}\big)^{2 |\ell + \alpha|} -1\right]\delta_{s\,s'}\,\delta_{\ell\,\ell'}\,.
	\ee
	
\end{lemma}

\begin{proof}
To begin with, from \eqref{eq:Pauli2res}, \eqref{eq:Fresolv} and \eqref{eq:traceop} (see also \cite[\S 10.27 and \S 10.29(ii)]{OLBC10}), we deduce that
\bmln{
\tau_{s}^{(\ell)} R_{\mathrm{S}}^{(\mathrm{F})}(z^*) \,\psiv  \\
= \lim_{r \to 0^{+}}\! \tfrac{2^{|\ell+\alpha|-1}\, \Gamma\big(|\ell+\alpha|\big)}{r^{|\ell+\alpha|}}\; \big(|\ell+\alpha| + r\, \partial_r\big) \!\int^\infty_0\!\!\!\! \diff r'\, r'\, I_{\vert \ell+\alpha\vert}\big(\!-i \sqrt{z^*}\,(r\land r')\big)\,K_{\vert \ell+\alpha\vert}\big(\!-i \sqrt{z^*}\,(r\lor r')\big)\, \psi_{s}^{(\ell)}(r') \\
= 2^{|\ell+\alpha|-1}\,\Gamma\big(|\ell+\alpha|\big)\,i \sqrt{z^*} \lim_{r \to 0^{+}}\! \,r^{1-|\ell+\alpha|}\, \bigg[
	K_{\vert \ell + \alpha\vert - 1}\big(\!-i \sqrt{z^*}\, r\big)\! \int^{r}_0\! \diff r'\, r'\, I_{\vert \ell+\alpha\vert}\big(\!-i \sqrt{z^*}\, r'\big)\, \psi_{s}^{(\ell)}(r') \\
	 - I_{\vert \ell + \alpha\vert - 1}\big(\!-i \sqrt{z^*}\,r\big)\! \int^\infty_{r}\!\!\! \diff r'\, r'\,K_{\vert \ell+\alpha\vert}\big(\!-i \sqrt{z^*}\,r'\big)\, \psi_{s}^{(\ell)}(r') \bigg] \,.
}
Exploiting the asymptotic behavior and normalization properties of the Bessel functions (see, \emph{e.g.}, \cite[\S 10.30(i)]{OLBC10} and \cite[Eq. 6.521.3]{GR07}, together with the basic inequalities \eqref{eq: estI}\eqref{eq: estK}), for $\ell \in \{0,-1\}$ and $r \to 0^{+}$, we infer
\begin{gather*}
	r^{1-|\ell+\alpha|}\, K_{\vert \ell + \alpha\vert - 1}\big(\!-i \sqrt{z^*}\,r\big) = \tfrac{\Gamma(1-|\ell+\alpha|)}{2^{|\ell+\alpha|} (-i\sqrt{z^*})^{1-|\ell+\alpha|}} + \OO\big(r^{2-2|\ell+\alpha|}\big)\, , \\
	r^{1-|\ell+\alpha|}\, I_{\vert \ell + \alpha\vert - 1}\big(\!-i \sqrt{z^*}\,r\big) = \tfrac{2^{1-|\ell+\alpha|}}{\Gamma(|\ell+\alpha|)\,(-i\sqrt{z^*})^{1-|\ell+\alpha|}} + \OO(r^2)\,, \\
\end{gather*}
and
\bmln{
	\left\vert \int^r_0\! \diff r'\,r'\, I_{\vert \ell+\alpha \vert} \big(\!-i \sqrt{z^*}\,r'\big)\, \psi_{s}^{(\ell)}(r') \right\vert  
	\leqslant \big\| \psi_{s}^{(\ell)} \big\|_{L^2(\R^{+},\,r\,dr)}\left( \int^r_0\! \diff r'\,r'\, \big|I_{\vert \ell+\alpha \vert} \big(\!-i \sqrt{z^*}\,r'\big) \big|^2 \right)^{\!\!1/2} \\
	\leqslant \left(\!\tfrac{\sqrt{|z|}}{\Im \sqrt{z^*}}\right)^{\!\vert \ell+\alpha\vert} \big\| \psi_{s}^{(\ell)} \big\|_{L^2(\R^{+},\,r\,\diff r)}\left( \int^r_0\! \diff r'\,r'\, I^2_{\vert \ell+\alpha\vert}\big(\Im \sqrt{z^*}\,(r \land r')\big)\! \right)^{\!\!1/2}
	\leqslant C\,\|\psiv\|_{2}\; r^{1+|\ell+\alpha|} \xrightarrow[r \to 0^{+}]{} 0\,,
}
\bmln{
	\left\vert \int_{r}^{+\infty}\!\!\!\! \diff r'\,r'\, K_{\vert \ell +\alpha\vert}\big(\!-i \sqrt{z^*}\,r'\big)\, \psi_{s}^{(\ell)}(r') \right\vert
	\leqslant \big\| \psi_{s}^{(\ell)} \big\|_{L^2(\R^{+},\,r\,\diff r)} \left(\int_{r}^{+\infty}\!\!\!\! \diff r'\, r'\, \big|K_{\vert \ell + \alpha\vert}\big(\!-i \sqrt{z^*}\,r'\big) \big|^2 \right)^{\!\!1/2} \\
	\leqslant \big\| \psi_{s}^{(\ell)} \big\|_{L^2(\R^{+},\,r\,\diff r)} \left(\int_{r}^{+\infty}\!\!\!\! \diff r'\, r'\, K^2_{\vert \ell + \alpha\vert}\big(\Im \sqrt{z^*}\,r'\big)\! \right)^{\!\!1/2} 
	\leqslant C \left(\int_{0}^{+\infty}\!\!\!\! \diff r'\, r'\, \big|K_{\vert \ell + \alpha\vert}\big(\Im \sqrt{z^*}\,r'\big) \big|^2 \right)^{\!\!1/2}\!\|\psiv\|_{2} < +\infty\,.
}
In view of the above considerations, by dominated convergence, we obtain
\[
\tau_{s}^{(\ell)} R_{\mathrm{S}}^{(\mathrm{F})}(z^*) \,\psiv  = \int_{0}^\infty\!\!\!\! \diff r'\, r'\! \int_{0}^{2\pi}\!\!\!\! \diff\teta'\, \big(\!-i\sqrt{z^*}\big)^{|\ell+\alpha|} K_{\vert \ell+\alpha\vert}\big(\!-i\sqrt{z^*}\,r'\big)\, \psi_{s}(r',\teta') \, \tfrac{e^{-i \ell \teta'} }{ \sqrt{2\pi}}\,.
\]
The identity \eqref{eq:GGzexp} then follows by simple duality arguments, recalling the definition \eqref{eq:Gzdef} of $ \gmv(z)$ and the explicit expression for $\giv_{\lambda,s}^{(\ell)}$ given by \eqref{eq:g} and \eqref{G1234}. Notice also that $\big(K_{\nu}(w)\big)^{*} = K_{\nu}(w^*)$ for all $\nu > 0$, $w \in \mathbb{C}$ \cite[10.34.7]{OLBC10} and $\big(\!-i\sqrt{z^*}\big)^{*}\! = -i \sqrt{z}$. 

On the other side, exploiting again basic features of the Bessel functions, it can be checked that $\giv_{1,s}^{(\ell)} - \giv_{-i \sqrt{z},s}^{(\ell)} \in \dom\big(H_{\mathrm{P}}^{(\mathrm{F})}\big)$. Then, a direct computation yields
\bmln{
\big[\Lambda(z) \qv \big]_{s}^{(\ell)}
= {\tx \sum_{s',\ell'}}\, q_{s'}^{(\ell')}\, \tau_{s}^{(\ell)}\! \lf(\!\giv_{1,s'}^{(\ell')} - \giv_{-i \sqrt{z},s'}^{(\ell')}\ri)\\
= {\tx \sum_{s',\ell'}}\, q_{s'}^{(\ell')}\, 2^{|\ell+\alpha|-1}\, \Gamma\big(|\ell+\alpha|\big)\, \delta_{s\,s'} \delta_{\ell\,\ell'} \lim_{r \to 0^{+}} \tfrac{1}{r^{|\ell+\alpha|}} \big(|\ell+\alpha| + r\, \partial_r \big)\! \lf[\giv_{1,s'}^{(\ell')} - \giv_{-i \sqrt{z},s'}^{(\ell')}\ri]_{s}^{(\ell)} \\
= \tfrac{\pi}{2 \sin(\pi|\ell+\alpha|)} \left[\big(\!-i \sqrt{z}\big)^{2 |\ell + \alpha|} -1\right] q_{s}^{(\ell)} ,
}
which ultimately accounts for \eqref{eq:Lzexp}.
\end{proof}

\begin{proof}[Proof of \cref{prop:quadKrein}]
Fixing $z = -\lambda^2$, with $\lambda > 0$, a direct comparison makes evident that the self-adjoint extensions $H_{\mathrm{P}}^{(\bb)}$ and $H_{\mathrm{P}}^{(\Theta)}$ do indeed coincide if and only if $\phiv_{\lambda} = \varphiv_{-\lambda^2} \in \dom\big(H_{\mathrm{P}}^{(\mathrm{F})}\big)$ and, accordingly,
	\be\label{eq:LLambda}
	\big[L(\lambda) + \bb\big] \qv = \big[\Lambda(-\lambda^2) + \Theta\big] \qv\,, \qquad  \forall \qv \in \mathbb{C}^4 \,.
	\ee
On account of the explicit expressions for $L(\lambda)$ and $\Lambda(-\lambda^2)$ reported respectively in \eqref{eq:L} and \eqref{eq:Lzexp}, it appears that the above condition \eqref{eq:LLambda} is actually equivalent to \eqref{eq: ThetaB}.
\end{proof}
Finally, we can complete the proof of the classification of self-adjoint extensions.
\begin{proof}[Proof of Theorem \ref{thm:QBHB}]
\mbox{} \\
\textsl{iii)}
The exhaustiveness of Krein's classification, combined with the one-to-one correspondence provided by \cref{prop:quadKrein} yield the result.
\end{proof}
\subsection{Spectral and scattering properties}\label{sec:specscatt}

Let us mention that many of the results on scattering described in the sequel rely on the {\it Limiting Absorption Principle} (LAP) for resolvent operators. In this connection, we refer to the weighted spaces
	\be
		L^2_{u}(\mathbb{R}^2; \mathbb{C}^2) := L^2\big(\mathbb{R}^2,  (1+|\xv|^2)^{u/2} \diff\xv\big) \otimes \mathbb{C}^2,
	\ee
with $ u \in \R $, and to the associated Banach spaces of bounded operators
	\be
		\mathcal{B}(u,u') := \mathcal{B}\big(L^2_{u}(\mathbb{R}^2;\mathbb{C}^2);L^2_{u'}(\mathbb{R}^2;\mathbb{C}^2)\big)\,.
	\ee
Any resolvent $R_{\mathrm{P}}^{(\Theta)}(z)$	is said to enjoy LAP if the limits
	\bdm
		R_{\mathrm{P,\pm}}^{(\Theta)}(\lambda) := \lim_{\varepsilon \to 0^{+}} R^{(\Theta)}_{\mathrm{P}}(\lambda \pm i \varepsilon)
	\edm
exist in $\mathcal{B}(u,-u)$ for some $u > 0$ and for all $\lambda \in \sigma_{\mathrm{ac}}\big(H_{\mathrm{P}}^{(\Theta)}\big) \setminus \mathrm{e}_{+}\big(H_{\mathrm{P}}^{(\Theta)}\big)$, where $\mathrm{e}_{+}\big(H_{\mathrm{P}}^{(\Theta)}\big)$ is the (possibly empty) discrete set of eigenvalues embedded in the absolutely continuous spectrum.

 We start by considering the Friedrichs Hamiltonian.
\begin{proposition}[LAP for $ H_{\mathrm{P}}^{(\mathrm{F})} $]\label{lemma:LAPF}
	\mbox{}\\
	For any $\lambda \in \R^{+}$, the limits
		\bdm
			R^{(\mathrm{F})}_{\mathrm{P},\pm}(\lambda) := \lim_{\varepsilon \to 0^{+}} R^{(\mathrm{F})}_{\mathrm{P}}(\lambda \pm i \varepsilon)
		\edm
exist in $\mathcal{B}(u,-u)$ for any $u > 1$ and the convergence is locally uniform.
\end{proposition}
\begin{proof}
The thesis is a straightforward consequence of \cite[Prop. 7.3]{IT01} and of the diagonal structure of the Friedrichs Hamiltonian, see \eqref{HPFdiag}.
\end{proof}

The action of $R^{(\mathrm{F})}_{\mathrm{P},\pm}(\lambda)$ can be deduced from the explicit representation \eqref{eq:Fresolv} for the corresponding Schr\"odinger resolvent operator and reads, for $s \in \{\uparrow,\downarrow\}$,
	\bmln{
		\left( R_{\mathrm{P},\pm}^{(\mathrm{F})}(\lambda) \psiv \right)_{\!s}\!(r,\teta) 
		=\! \int^\infty_0\!\!\!\! \diff r'\, r'\! \int^{2\pi}_0\!\!\!\! \diff\teta'\, \sum_{\ell\in\Z} I_{\vert \ell+\alpha\vert}\big(\mp i\sqrt{\lambda}\,(r \land r')\big)\,K_{\vert \ell+\alpha\vert}\big(\mp i\sqrt{\lambda}\,(r \lor r')\big)\,\psi_{s}(r',\teta') \, \tfrac{e^{i \ell (\teta-\teta')}}{2\pi}\\
		= \tfrac{i}{4} \int^\infty_0\!\!\!\! \diff r'\, r'\! \int^{2\pi}_0\!\!\!\! \diff\teta'\, \sum_{\ell\in\Z}  J_{\vert \ell+\alpha\vert}\big(\!\pm \!\sqrt{\lambda}\,(r\land r')\big)\, H^{(1)}_{\vert \ell+\alpha\vert}\big(\!\pm\! \sqrt{\lambda}\,(r \lor r')\big)\,\psi_{s}(r',\teta')\, e^{i \ell (\teta-\teta')}\,.
	}
In the second line we have used the connection formulas reported in \cite[Eqs. 10.27.6 and 10.27.8]{OLBC10}.

\begin{proof}[Proof of \cref{prop: waveopF}]
Also in this case, the thesis follows from classical results on the Schr\"odinger Hamiltonian, see \cite{Ru83} and \cite[Prop. 7.4]{IT01}. The absence of singular continuous spectrum is ensured by the LAP established in \cref{lemma:LAPF} for $R^{(\mathrm{F})}_{\mathrm{P}}(z)$ (see, \emph{e.g.}, \cite[Thm. 6.1]{Ag75} and \cite[Cor. 4.7]{MPS18}). The stated properties of the scattering operator ultimately follow by standard arguments.
\end{proof}

\begin{proof}[Proof of \cref{prop: spectrumF}]
Since the Friedrichs quadratic form is non-negative (see \eqref{eq: Qpdef}), it appears that $\sigma\big(H_{\mathrm{P}}^{(\mathrm{F})}\big) \subseteq \R^{+}$. We already noticed that \cref{lemma:LAPF} ensures the absence of singular continuous spectrum. Moreover, the existence and completeness of the wave operators established in \cref{prop: waveopF} grants that $\sigma_{\mathrm{ac}}\big(H_{\mathrm{P}}^{(\mathrm{F})}\big) = \sigma_{\mathrm{ac}}(-\Delta_{\mathrm{P}}) = \R^{+}$. Finally, it can be checked by direct inspection that $H_{\mathrm{P}}^{(\mathrm{F})}$ has no eigenvalue (see also \cite[Thm. 3.3]{CFK20}).
\end{proof}
We now address the explicit form of the generalized eigenfunction of the Friedrichs Hamiltonian.	

\begin{proof}[Proof of \cref{thm: scattmatF}]
By decomposition in angular harmonics and some explicit computations, one obtains formula \eqref{eq: psiF} (see also \cite{AT11,Ru83,Ta11}).
Notice that, despite solving the radial eigenvalue problem, the Bessel functions $Y_{|\ell + \alpha|}$ do not appear in \eqref{eq: psiF} since they do not satisfy the proper local behavior close to $\xv = \mathbf{0}$. Let us further remark that the coefficients in the expansion \eqref{eq: psiF} have been fixed so as to fulfill the radiation conditions \eqref{eq: Sommerfeld}. Indeed, using the known asymptotic expansion of the Bessel functions \cite[Eq. 10.7.8]{OLBC10}, it can be checked that there exist $f^{(\mathrm{F},\pm)}_{\kv} \in L^2(\mathbb{S}^{1};\mathbb{C})$ such that
	\be\label{eq: psiFasy}
		\varphiv^{(\mathrm{F},\pm)}_{(s,\kv)}(\xv) =  \lf[ \tfrac{e^{i \kv \cdot \xv}}{2\pi} + f^{(\mathrm{F},\pm)}_{\kv}\,\tfrac{e^{\mp i |\kv|\,|\xv|}}{|\xv|^{1/2}} + \mathcal{O}\lf(\tfrac{1}{|\xv|^{3/2}}\ri) \ri] \begin{pmatrix} \delta_{s,\uparrow} \\ \delta_{s,\downarrow} \end{pmatrix}, \qquad \mbox{for\, $|\xv| \to +\infty$}\,.
	\ee
More precisely, one has
	\be\label{eq: fpmF}
		f^{(\mathrm{F},\pm)}_{\kv}(\teta) = \tfrac{e^{\pm i \frac{\pi}{4}}}{(2\pi)^{3/2} \sqrt{k}} \sum_{\ell \in \mathbb{Z}} \lf(e^{\pm i \pi |\ell + \alpha|} - e^{\pm i \pi |\ell|}\ri) e^{i \ell (\teta - \omega_{\pm})}\,.
	\ee
Incidentally, we notice that the above expansion shows that $\varphiv^{(\mathrm{F},\pm)}_{(s,\kv)} \in L^2_{-u}(\mathbb{R}^2;\mathbb{C}^2)$ for any $u > 2$.

Keeping in mind that $H_{\mathrm{P}}^{(\mathrm{F})}$ has purely absolutely continuous spectrum (see \cref{prop: spectrumF}), we define the modified Fourier transforms
	\be\label{eq: genFouFried}
		\mathfrak{F}^{(\mathrm{F})}_{\pm}\! : L^2(\mathbb{R}^2;\mathbb{C}^2) \to L^2(\mathbb{R}^2;\mathbb{C}^2)\,, \qquad
		\lf(\mathfrak{F}^{(\mathrm{F})}_{\pm} \psiv\ri)_{\!s}\!(\kv) := \!\sum_{s' \in \{\uparrow,\downarrow\}} \int_{\mathbb{R}^2} \!\!\diff\xv\; \overline{\big(\varphi^{(\mathrm{F},\pm)}_{(s,\kv)}(\xv)\big)_{s'}}\,\psi_{s'}(\xv)\,.
	\ee
Let us now return to the wave operators $\Omega^{(\mathrm{F})}_{\pm}$, whose existence and asymptotic completeness have been established in \cref{prop: waveopF}. These are known to fulfill the identity (see, \emph{e.g.}, \cite[Thm. 5.5]{MPS18})
	\bdm
		\Omega^{(\mathrm{F})}_{\pm} = \lf(\mathfrak{F}^{(\mathrm{F})}_{\pm}\ri)^{\!*}\! \mathfrak{F}\,.
	\edm
Accordingly, the unitary scattering operator is given by 
	\be\label{eq: SFfourier}
		\mathrm{S}^{(\mathrm{F})} = \mathfrak{F}^{*}\,\mathfrak{F}^{(\mathrm{F})}_{+}\! \lf(\mathfrak{F}^{(\mathrm{F})}_{-}\ri)^{\!*}\! \mathfrak{F}\,.
	\ee
Recalling the definition \eqref{eq: scattmat} of the scattering matrix $\mathrm{S}^{(\mathrm{F})}(\lambda)$ and making reference to \eqref{eq: psiF} \eqref{eq: genFouFried} and \eqref{eq: SFfourier}, we get \eqref{eq: SFexp} by distributional computations (see \cite[Eq. (4.8)]{Ru83} and \cite[p. 315]{IT01}). In particular, let us mention that an elementary calculation yields
	\bdm
		\sum_{s''} \int_{\mathbb{R}^2} \!\!\diff\xv\; \big(\varphi^{(\mathrm{F},+)}_{(s,\kv)}\big)^{*}_{\!s''}(\xv) \big(\varphi^{(\mathrm{F},-)}_{(s',\kv')}\big)_{\!s''}(\xv)
		=  \tfrac{1}{2\pi} \sum_{\ell \in \mathbb{Z}} e^{i \ell (\omega -\omega'+\pi) - i \pi\,|\ell + \alpha|} \int_{0}^{\infty}\!\!\!\diff r\;r \; J_{|\ell + \alpha|}(k\,r)\, J_{|\ell + \alpha|}(k' r)\,\delta_{s\,s'}\, .
	\edm
By means of \cite[Eq. 6.541]{GR07} and \cite[Eqs. 10.40.1-2]{OLBC10}, it can be checked that in the sense of distributions there holds
	\bmln{
		\int_{0}^{\infty}\!\!\!\diff r\;r \; J_{|\ell + \alpha|}(k\,r)\, J_{|\ell + \alpha|}(k' r)
		= \lim_{\zeta \to +\infty} \int_{0}^{\infty}\!\!\!\diff r\;\frac{\zeta^2\,r}{\zeta^2 + r^2} \; J_{|\ell + \alpha|}(k\,r)\, J_{|\ell + \alpha|}(k' r) \\
		= \lim_{\zeta \to +\infty} \zeta^2\,I_{|\ell + \alpha|}\big(\zeta\,(k \land k')\big)\,K_{|\ell + \alpha|}\big(\zeta\,(k \lor k')\big)
		= \tfrac{1}{ 2k}\,\delta(k'\!-k) = \delta\big((k')^2- k^2\big)\,.
	}
As a consequence, we obtain
	\bdm
		\sum_{s''} \int_{\mathbb{R}^2} \!\!\diff\xv\; \big(\varphi^{(\mathrm{F},+)}_{(s,\kv)}(\xv)\big)^{*}_{\!s''} \big(\varphi^{(\mathrm{F},-)}_{(s',\kv')}(\xv)\big)_{\!s''}
		=  \tfrac{1}{2\pi} \sum_{\ell \in \mathbb{Z}} e^{i \ell (\omega -\omega') + i \pi\,(\ell - |\ell + \alpha|)}\,\delta\big((k')^2- k^2\big)\,\delta_{s\,s'} \, ,
	\edm
which is a key ingredient for the derivation of the first equality in \eqref{eq: SFexp}. For the second equality we refer to \cite[Eq. (4.8)]{IT01}. Finally, on account of \eqref{eq: scattampl}, \eqref{eq: dcs} and \eqref{eq: SFexp} one readily infers \eqref{eq:1} and \eqref{eq:2}.
\end{proof}

We now deal with the other self-adjoint realizations, starting with the following analogue of \cref{lemma:LAPF}.
\begin{proposition}[LAP for $ H_{\mathrm{P}}^{(\Theta)} $]\label{lemma:LAPTheta}
\mbox{}\\
	Let $\Theta \in \mathrm{M}_{4,\mathrm{Herm}}(\mathbb{C})$. Then, there exists a $($possibly empty$)$ discrete set $\mathrm{e}_{+}^{(\Theta)}$, with $\# \big[ \mathrm{e}_{+}^{(\Theta)}\big] \leqslant 4$, such that for any $\lambda \in \R^{+} \setminus \mathrm{e}_{+}^{(\Theta)}$ the limits
		\bdm
			R^{(\Theta)}_{\mathrm{P},\pm}(\lambda) := \lim_{\varepsilon \to 0^{+}} R^{(\Theta)}_{\mathrm{P}}(\lambda \pm i \varepsilon)
		\edm
exist in $\mathcal{B}(u,-u)$ for any $u > 1$ and the convergence is locally uniform. Moreover, there holds
		\bdm
			R^{(\Theta)}_{\mathrm{P},\pm}(\lambda) = R^{(\mathrm{F})}_{\mathrm{P},\pm}(\lambda) + \gmv_{\pm}(\lambda) \big[\Lambda_{\pm}(\lambda) + \Theta\big]^{-1} \widecheck{\mathcal{G}}_{\pm}(\lambda)\,,
		\edm
where
		\begin{align}
			\Lambda_{\pm}(\lambda) & := \lim_{\varepsilon \to 0^{+}}\Lambda(\lambda \pm i \varepsilon) \;\in\; \mathrm{M}_{4,\mathrm{Herm}}(\mathbb{C})\,, \label{eq:Lambdapmdef} \\
			\gmv_{\pm}(\lambda) & := \lim_{\varepsilon \to 0^{+}} \gmv(\lambda \pm i \varepsilon) \;\in\; \mathcal{B}\big(\mathbb{C}^4; L^2_{-u}(\mathbb{R}^2;\mathbb{C}^2)\big)\,; \label{eq: Gpmlimdef} \\
			\widecheck{\mathcal{G}}_{\pm}(\lambda) & := \lim_{\varepsilon \to 0^{+}} \tau R_{\mathrm{P}}^{(\mathrm{F})}(\lambda \mp i \varepsilon) \;\in\; \mathcal{B}\big(L^2_{u}(\mathbb{R}^2;\mathbb{C}^2); \mathbb{C}^4\big)\,.\label{eq: Gcpmlimdef}
		\end{align}
\end{proposition}

\begin{proof} Let us refer to the Krein formula \eqref{eq: RThetaKrein} for the resolvent operator $R_{\mathrm{P}}^{(\Theta)}(z)$. We firstly recall that the Friedrichs resolvent $R_{\mathrm{P}}^{(\mathrm{F})}(z)$ enjoys LAP in $\mathcal{B}(u,-u)$ for any $u > 1$, see \cref{lemma:LAPF}. On the other hand, using the explicit expression \eqref{eq:Lzexp} for $\Lambda(z)$, we obtain
	\be\label{eq:Lambdapm}
		\big(\Lambda_{\pm}(\lambda)\big)_{s\,s'}^{(\ell\,\ell')} = \lim_{\varepsilon \to 0^{+}} \Lambda_{s\,s'}^{(\ell\,\ell')}(\lambda \pm i \varepsilon) = \tfrac{\pi}{2 \sin(\pi|\ell+\alpha|)} \left[e^{\mp i \pi |\ell + \alpha|}\, \lambda^{|\ell + \alpha|} - 1\right]\delta_{s\,s'}\,\delta_{\ell\,\ell'}\,.
	\ee
From here we deduce that, depending on the specific choice of $\Theta$, the matrices $\Lambda_{\pm}(\lambda) + \Theta \in \mathrm{M}_{4,\mathrm{Herm}}(\mathbb{C})$ can indeed become singular for suitable values of $\lambda \in \R^{+}$. We indicate with $\mathrm{e}_{+}\big(H_{\mathrm{P}}^{(\Theta)}\big)$ the collection of such exceptional points and notice that its cardinality is at most $4$. It is evident that the convergence in \eqref{eq:Lambdapm} is uniform on any compact subset of $\R^{+} \setminus \mathrm{e}_{+}\big(H_{\mathrm{P}}^{(\Theta)}\big)$, so the same can be said for the inverses $\big[\Lambda_{\pm}(\lambda) + \Theta\big]^{-1}$.

To say more, using \eqref{eq:g}, \eqref{G1234} and \eqref{eq:GGzexp}, together with the Bessel connection formula \cite[Eq. 10.27.8]{OLBC10}, for any $\qv \in \mathbb{C}^4$, we infer
	\be\label{eq: Gpmlim}
		\big(\gmv_{\pm}(\lambda)\,\qv\big)_{s}(r,\teta) = \lim_{\varepsilon \to 0^{+}} \big(\gmv(\lambda \pm i \varepsilon)\,\qv\big)_{s}(r,\teta)
		= \tfrac{i \pi}{2}\, {\tx \sum_{\ell}}\; q_{s}^{(\ell)}\;\big(\!\mp\! \sqrt{\lambda}\,\big)^{|\ell+\alpha|}\,  H^{(1)}_{|\ell+\alpha|}\big(\!\mp\! \sqrt{\lambda}\, r\big)\,\tfrac{e^{i \ell \teta}}{\sqrt{2\pi}}\,.
	\ee
Taking into account the regularity of the Hankel function $H^{(1)}_{\nu}$ and its asymptotic expansions for small and large arguments \cite[Eqs. 10.7.7 and 10.17.5]{OLBC10}, by an elementary change of variable it can be checked that
	\bdm
		\lf\| \gmv_{\pm}(\lambda)\,\qv \ri\|_{L^2_{-u}(\mathbb{R}^2;\mathbb{C}^2)}^2 
		\leqslant C\, \lambda^{|\ell + \alpha| - 1} \!\sum_{\ell \,\in\, \{0,-1\}}\! \lf( \int_{0}^{1}\!\diff \rho\,\frac{1}{\rho^{2|\ell+\alpha|-1}} + \lambda^{u/2}\! \int_{1}^{\infty}\!\!\! \diff \rho\,\frac{1}{\rho^u} \ri) \vert\qv\vert^2\,.
	\edm
 The above estimate shows that $ \gmv_{\pm}(\lambda)$, with $\lambda > 0$, are bounded operators from $\mathbb{C}^4$ into $L^2_{-u}(\mathbb{R}^2;\mathbb{C}^2)$ for any $u > 1$. On top of that, using a known integral representation for the Bessel functions \cite[Eq. 8.421.9]{GR07} it can be checked that the limit in \eqref{eq: Gpmlim} is attained uniformly on any compact subset of $\R^{+}$. 
Since $\widecheck{\mathcal{G}}(z)$ is the adjoint of $\gmv(z^*)$, see \eqref{eq:Guzdef} and \eqref{eq:Gzdef}, by elementary duality considerations the above arguments also prove that the limits $\widecheck{\mathcal{G}}_{\pm}(\lambda)$ defined in \eqref{eq: Gcpmlimdef} identify a pair of bounded operators from $ \big(L^2_{-u}(\mathbb{R}^2;\mathbb{C}^2) \big)' \simeq L^2_{u}(\mathbb{R}^2;\mathbb{C}^2) $ to $(\mathbb{C}^4)' \simeq \mathbb{C}^4$ for any $u > 1$.
\end{proof}

\begin{proof}[Proof of \cref{prop: waveopTheta}]
Let us first remark that the resolvent operator associated to $H_{\mathrm{P}}^{(\Theta)}$ is a finite rank perturbation of the resolvent related to the Friedrichs Hamiltonian $H_{\mathrm{P}}^{(\mathrm{F})}$. This suffices to infer that the wave operators $\Omega_{\pm}\big(H_{\mathrm{P}}^{(\Theta)},H_{\mathrm{P}}^{(\mathrm{F})}\big)$ exist and are complete \cite[Thm. XI.9]{RS81}. On the other hand, recall that existence and completeness of the wave operators $\Omega_{\pm}\big(H_{\mathrm{P}}^{(\mathrm{F})},-\Delta_{\mathrm{P}}\big)$ has already been established in \cref{prop: waveopF}. Then, existence and completeness of $\Omega^{(\Theta)}_{\pm}$ follows readily from the chain rule for wave operators \cite[Ch. XI, p. 18, Prop. 2]{RS81}. Finally, we deduce asymptotic completeness noting that the LAP established in \cref{lemma:LAPTheta} for $R_{\mathrm{P}}^{(\Theta)}(z)$ ensures the absence of singular continuous spectrum.
\end{proof}

\begin{proof}[Proof of \cref{prop: spectrumTheta}] On one side, the existence and asymptotic completeness of the wave operators $\Omega_{\pm}\big(H_{\mathrm{P}}^{(\Theta)},-\Delta_{\mathrm{P}}\big)$ ensure that $\sigma_{\mathrm{ac}}\big(H_{\mathrm{P}}^{(\Theta)}\big) = \sigma_{\mathrm{ac}}\big(-\Delta_{\mathrm{P}}\big) = \R^{+}$ and the absence of singular continuous spectrum. On the other side, from \cite[Thm. 3.4]{Po04} it follows that the map $\qv \mapsto \gmv(-\mu) \qv$ is a bijection from $\ker \big[\Lambda(-\mu) + \Theta \big]$ to $\ker\big(H_{\mathrm{P}}^{(\Theta)} + \mu\big)$, which proves the part of the thesis regarding $\sigma_{\mathrm{pp}}\big(H_{\mathrm{P}}^{(\Theta)}\big)$.
\end{proof}

\begin{proof}[Proof of \cref{thm: scattmatT}]
 We deduce the expression in \eqref{eq: eigenfunctTheta} by a straightforward adaptation of \cite[Thm. 5.1]{MPS18}. 
Using \eqref{eq: psiF} and the Bessel function asymptotics \cite[Eq. 10.7.3]{OLBC10} we get that
	\bml{
		\tau_{s'}^{(\ell)} \varphiv^{(\mathrm{F},\pm)}_{(s,\kv)}
		= 2^{|\ell+\alpha|-1}\, \Gamma\big(|\ell+\alpha|\big) \lim_{r \to 0^{+}} \frac{1}{r^{|\ell+\alpha|}}\big(|\ell+\alpha| + r\, \partial_r \big) \big(\varphi^{(\mathrm{F},\pm)}_{(s,\kv)}\big)_{s'}^{(\ell)} \\
		= \delta_{s\,s'}\,\tfrac{1}{\sqrt{2\pi}}\,2^{|\ell+\alpha|-1}\, \Gamma\big(|\ell+\alpha|\big)\, e^{\pm i \frac{\pi}{2}\,|\ell + \alpha| -i \ell \omega_{\pm}} \lim_{r \to 0^{+}} \frac{1}{r^{|\ell+\alpha|}}\big(|\ell+\alpha| + r\, \partial_r \big) \, J_{|\ell + \alpha|}(k\,r) \\
		=  ( \pm i\, k)^{|\ell + \alpha|}\,\tfrac{e^{-i \ell \omega_{\pm}}}{\sqrt{2\pi}}\,\delta_{s\,s'}\,. \label{eq:tauFried}
	}
Recalling the explicit expressions \eqref{eq:Lambdapm} and \eqref{eq: Gpmlim}, we obtain
	\bmln{
		\big(\varphiv^{(\Theta,\pm)}_{(s,\kv)}\big)_{\!s'} (r,\teta)
		= \tfrac{1 }{ 2\pi}\, \sum_{\ell \in \mathbb{Z}} e^{i \ell (\teta-\omega_{\pm}) \pm i \frac{\pi}{2}\,|\ell + \alpha|}\, J_{|\ell + \alpha|}(k\,r)\,  \delta_{s\,s'}\,\\
			+ \tfrac{i}{4} \sum_{\ell,\ell' \in \{0,-1\}}\! \Big\{\big[\Lambda_{\pm}(k^2) + \Theta\big]^{-1}\Big\}_{s'\!,s}^{\ell'\!,\ell}\, e^{i (\ell' \teta - \ell \omega_{\pm})}\;(\pm i k)^{|\ell + \alpha|}\,(\mp k)^{|\ell'+\alpha|}\,  H^{(1)}_{|\ell'+\alpha|}(\mp k\, r) \,.
	}
To say more, in view of \eqref{eq: psiFasy} and \eqref{eq: fpmF}, by means of \cite[Eq. 10.17.5]{OLBC10} we deduce the following, for\, $|\xv| \to +\infty$:

	\bml{
		f^{(\Theta,\pm)}_{(s,s'),\kv}
					= \tfrac{e^{\pm i \frac{\pi}{4}}}{(2\pi)^{3/2} \sqrt{k}} \sum_{\ell \in \mathbb{Z}} \lf(e^{\pm i \pi |\ell + \alpha|} - e^{\pm i \pi |\ell|}\ri) e^{i \ell (\teta - \omega_{\pm})} \delta_{s\,s'}\, \\
				+ \tfrac{i \pi \,e^{\pm i \frac{\pi}{4}}}{(2\pi)^{3/2} \sqrt{k}} \sum_{\ell,\ell' \in \{0,-1\}}\! \Big\{\big[\Lambda_{\pm}(k^2) + \Theta\big]^{-1}\Big\}_{s'\!,s}^{\ell'\!,\ell}\, e^{i (\ell' \teta - \ell \omega_{\pm})}\;(\pm i k)^{|\ell + \alpha| + |\ell'+\alpha|}\,.
	}
This confirms that $\varphiv^{(\Theta,+)}_{(s,\kv)}$ and $\varphiv^{(\Theta,-)}_{(s,\kv)}$ fulfill, respectively, the incoming and outgoing Sommerfeld radiation conditions.
\end{proof}

\begin{proof}[Proof of \cref{prop: zeroenresTheta}]
By decomposition in angular harmonics and an explicit calculation, it can be checked that the only distributional solutions of the zero-energy equation $H_{\mathrm{P}} \psiv = 0$ are of the form
	\be\label{eq: zeroenresexp}
		\psi_{s}(r,\teta) = \sum_{\ell \in \mathbb{Z}} \frac{e^{i \ell \teta}}{\sqrt{2\pi}} \lf[ c_{s}^{(\ell)}\,r^{-|\ell + \alpha|} + d_{s}^{(\ell)}\,r^{|\ell + \alpha|} \ri] ,
	\ee
with suitable coefficients $\mathbf{c}, \mathbf{d}  \in \mathbb{C}^4$ and for $s \in \{\uparrow,\downarrow\}$. The condition of uniform boundedness at infinity forces us to fix $\mathbf{d} = 0$. On one hand, the local Friedrichs conditions $\nabla \phiv_0, A_j \phiv_0 \in L^2_{\mathrm{loc}}(\mathbb{R}^2,\mathbb{C}^2)$ demand that $\mathbf c = 0$.
On the other hand, for a general extension, to exhibit the proper singular behavior at $\xv = \mathbf{0}$ encoded in $\dom\big(H_{\mathrm{P}}^{(\Theta)}\big)$, the above solutions must be locally of the form $\psiv = \phiv_{0} + \gmv(0) \qv$ with $\tau \phiv_0 = \big[\Lambda(0) + \Theta \big] \mathbf{q}$, see \eqref{eq:HThetadef}. Given that there is no regular part, namely $\phiv_{0} = \bm{0}$, this requirement can be fulfilled only if
	\begin{gather*}
		c_{s}^{(\ell)} = \lf\{\!\begin{array}{ll}
			\frac{\Gamma(|\ell + \alpha|)}{2^{1-|\ell+\alpha|}}\,q_{s}^{(\ell)}	& \mbox{for $\ell \in\{ 0,-1\}$}\,, \vspace{0.15cm}\\
			0 &	\mbox{for $\ell \in \Z \setminus \{ 0,-1\}$}\,, 
		\end{array}\ri. \qquad 
		\mbox{for some $\qv  \in \ker\big[\Lambda(0) + \Theta\big]$}\,,
	\end{gather*}
which yields the thesis.
\end{proof}

\subsection{Symmetries}\label{sec:sym}
\begin{proof}[Proof of \cref{pro: symmetries}]
Consider the transformations \eqref{eq:UST}, as the case \eqref{eq:VST} is completely analogous. Then,
\[
(\UU\psiv)(\xv) = e^{-i\eta_0-i\eta_3\sigma_3}\psiv\big(T^{-1}\xv\big)\,, \qquad T \in SO(2,\R)\,,\; \eta_0,\eta_3\in\R\,.
\]
Concerning the Friedrichs realization, the computations in \cref{app:2Dsym} show that
\begin{equation}\label{eq:Faction}
\UU H_{\mathrm{P}}^{(\mathrm{F})}(\mathbf A)\,\UU^{-1}=H_{\mathrm{P}}^{(\mathrm{F})}(\mathbf A)\,.
\end{equation}

We now consider a generic self-adjoint extension $H_{\mathrm{P}}^{(\bb)}$, $\bb \in \mathrm{M}_{4,\mathrm{Herm}}(\mathbb{C})$, belonging to the family characterized in \cref{thm:QBHB}. Notice that, in order to prove the invariance of the domain, it suffices to prove the inclusion 
\begin{equation}\label{eq:incl}
\UU\,\dom(H_{\mathrm{P}}^{(\bb)}(\mathbf A))\subseteq \dom (H_{\mathrm{P}}^{(\bb)}(\mathbf A))\,.
\end{equation}
Given $\psiv\in\dom (H_{\mathrm{P}}^{(\bb)}(\mathbf A))$, such an element decomposes as in \eqref{HBdom}
and then
\[
\UU\psiv=\UU\phiv_\lambda+\sum_{s,\ell}q^{(\ell)}_s \,\UU \giv^{(\ell)}_{\lambda,s}\, .
\]
Thus one needs to rewrite the last terms in the above formula as in \eqref{HBdom}, defining a new charge $\widetilde {\mathbf q}$, and then proceed to check the boundary conditions. To this aim, since the multiplicative factor $e^{-i\eta_0}$ drops in the latter, in order to simplify the computations we can take $\eta_0=0$, so that we actually consider
\[
(\UU\psiv)(\xv) = e^{-i\eta_3\sigma_3}\psiv\big(T^{-1} \xv\big)\,.
\]
Writing the rotation matrix $T$ as
\[
T=\begin{pmatrix}\cos\varsigma & -\sin\varsigma \\ \sin\varsigma & \cos\varsigma \end{pmatrix}\,,\qquad \varsigma\in[0,2\pi),
\] 
and using \eqref{eq:g} and \eqref{G1234} one finds
\[
\UU \giv^{(\ell)}_{\lambda,\uparrow} = e^{-i(\eta_3 + \varsigma \ell)}\, \giv^{(\ell)}_{\lambda,\uparrow}\,,\qquad  
\UU \giv^{(\ell)}_{\lambda,\downarrow} = e^{i (\eta_3 - \varsigma \ell)}\, \giv^{(\ell)}_{\lambda, \downarrow }\,.
\]
Then, $\widetilde\psiv:=\UU\psiv$ can be rewritten as
\[
\widetilde\psiv=\widetilde\phiv_\lambda+\sum_{s,\ell}\,\widetilde q^{(\ell)}_s\, \giv^{(\ell)}_{\lambda,s}\,,
\] 
with $\widetilde\phiv_\lambda:=\UU\phiv_\lambda$ and the new charges
\[
\widetilde q^{(\ell)}_{\uparrow} = e^{-i(\eta_3 + \varsigma \ell)}\, q^{(\ell)}_{\uparrow}\,, \qquad
\widetilde q^{(\ell)}_{\downarrow} = e^{i(\eta_3 - \varsigma \ell)}\, q^{(\ell)}_{\downarrow}\,.
\]
Notice that
\[
\widetilde\phiv_\lambda(x) = (\UU\phiv_\lambda)(x) 
= \begin{pmatrix}e^{-i\eta_3} \phi_{\lambda,\uparrow}(T^{-1}x) \\ e^{i\eta_3} \phi_{\lambda,\downarrow}(T^{-1}x)  \end{pmatrix} .
\]
Simple calculations give the following results for the Fourier coefficients involved 
\[
\widetilde\phi^{(\ell)}_{\lambda,\uparrow} = e^{-i(\eta_3 + \varsigma \ell)} \phi^{(\ell)}_{\lambda,\uparrow}\,, \quad 
\widetilde\phi^{(\ell)}_{\lambda,\downarrow} = e^{i(\eta_3 - \varsigma \ell)} \phi^{(\ell)}_{\lambda,\downarrow} \,.
\]
Let us now examine the boundary conditions in \eqref{HBdom}. After some long but straightforward algebraic computations one sees that the following conditions must be fulfilled:
\begin{align*}
	\big(L(\lambda)+\bb\big)_{\uparrow \uparrow}^{(\ell \ell')} \lf(e^{i \varsigma(\ell - \ell')} - 1\ri) = 0\,, & \qquad
	\big(L(\lambda)+\bb\big)_{\uparrow \downarrow}^{(\ell \ell')} \lf(e^{2i \eta_3 + i \varsigma(\ell - \ell')} - 1\ri) = 0\,, \\
	\big(L(\lambda)+\bb\big)_{\downarrow \uparrow}^{(\ell \ell')} \lf(e^{-2i\eta_3 + i \varsigma(\ell - \ell')} - 1\ri) = 0\,, & \qquad
	\big(L(\lambda)+\bb\big)_{\downarrow \downarrow}^{(\ell \ell')} \lf(e^{i \varsigma(\ell - \ell')} - 1\ri) = 0\,.
\end{align*}
Notice that the second and third equations are indeed equivalent, since $L(\lambda),\bb$ are Hermitian.
Therefore, for generic values of $\eta_3,\varsigma$, since $L(\lambda)$ is diagonal, $\bb$ must be a diagonal matrix too.

Concerning the action of the operator, thanks to \eqref{eq:Faction} there holds
\bmln{
\UU \big(H_{\mathrm{P}}^{(\bb)}\!+\lambda^2\big)\UU^{-1}(\UU\psiv) =\UU \big(H_{\mathrm{P}}^{(\bb)}+\lambda^2\big)\psiv=\UU \big(H_{\mathrm{P}}^{(\mathrm{F})}+\lambda^2\big)\psiv  \\
=\UU \big(H_{\mathrm{P}}^{(\mathrm{F})}+\lambda^2\big)\UU^{-1}(\UU\psiv)
= \big(H_{\mathrm{P}}^{(\mathrm{F})}+\lambda^2)(\UU\psiv\big)= \big(H_{\mathrm{P}}^{(\bb)}+\lambda^2\big)(\UU\psiv)\,.
}
\end{proof}

\subsection{Comparison with the Dirac operator}\label{sec:compdirac}

\begin{proof}[Proof of \cref{pro: dirac vs pauli}]
The proof of the thesis will be achieved in various steps. Without loss of generality, in what follows we fix $\lambda=1$, omitting the dependence on the spectral parameter for simplicity of notation. We will repeatedly use the expansion as $r\to0^+$ given in \eqref{eq:G01asy} (see \cite[Eq. 6.521.3]{GR07}), alongside with the basic identities 
\be\label{eq:gxi}
\left\{\!\begin{array}{l}
\giv^{(-1)}_\uparrow\! = \tfrac{1}{\sqrt{2\pi}}\,\tfrac{\xiv_{+}+\xiv_{-}}{2}\,, \\
\giv^{(0)}_\downarrow=\tfrac{1}{\sqrt{2\pi}}\,\tfrac{\xiv_{+}-\xiv_{-}}{2}\,,
\end{array}\right.
\qquad \Longleftrightarrow \qquad 
\left\{\!\begin{array}{l}
\xiv_{+} = \sqrt{2\pi}\, \big(\giv^{(-1)}_\uparrow + \giv^{(0)}_\downarrow\big)\,,\\
\xiv_{-}= \sqrt{2\pi} \big(\giv^{(-1)}_\uparrow - \giv^{(0)}_\downarrow\big)\,.
\end{array}\right.
\ee
which readily follow from \eqref{eq:g}, \eqref{G1234} and \eqref{eq:Ddefect}.

{\sl Step 1.}  We show that $\dom(H^{(\bb)}_{\mathrm{P}})\subseteq \dom(H_{\mathrm{D}}^{\gamma})$ if and only if one of the following alternatives holds:
\begin{enumerate}[(a{.1})]
\item $\gamma=0$, $\bb_{\uparrow \uparrow}^{(0 0)} = \bb_{\downarrow \downarrow}^{(0 0)} = \bb_{\downarrow \downarrow}^{(-1 -1)} = \infty$ and $\bb_{s s'}^{(\ell \ell')}$ arbitrary otherwise;

\item $\gamma=\pi$, $\bb_{\uparrow \uparrow}^{(0 0)} = \bb_{\uparrow \uparrow}^{(-1 -1)} = \bb_{\downarrow \downarrow}^{(-1 -1)} = \infty$ and $\bb_{s s'}^{(\ell \ell')}$ arbitrary otherwise;

\item $\gamma \in [0,2\pi)$ arbitrary, $\bb_{\uparrow \uparrow}^{(0 0)} = \bb_{\uparrow \uparrow}^{(-1 -1)} = \bb_{\downarrow \downarrow}^{(0 0)} = \bb_{\downarrow \downarrow}^{(-1 -1)} = \infty$ and $\bb_{s s'}^{(\ell \ell')}$ arbitrary otherwise.
\end{enumerate}
\noindent
Let us firstly notice that \eqref{asyF} entails $\phi_{\lambda}(r,\theta) = o(1)$ as $r \to 0^+$, for all $\phi_{\lambda} \in \dom\big(H^{(\mathrm{F})}_{\mathrm{P}}\big)$. From here and from \eqref{eq:g}, \eqref{G1234} and \eqref{eq:G01asy} it follows that, for any $\psiv \in \dom\big(H^{(\bb)}_{\mathrm{P}}\big)$, there holds
\[
\psiv 
= \varphiv_{\lambda} + \tx\sum_{s,\ell} q^{(\ell)}_s \giv^{(\ell)}_{\lambda,s}
= \!\begin{pmatrix} q^{(0)}_\uparrow\, \frac{\Gamma(\alpha)}{2^{1-\alpha}}\frac{1}{\sqrt{2\pi}}\frac{1}{r^\alpha} + q^{(-1)}_\uparrow\,\frac{\Gamma(1-\alpha)}{2^\alpha}\frac{e^{-i\theta}}{\sqrt{2\pi}}\frac{1}{r^{1-\alpha}} \vspace{0.1cm}\\
 q^{(0)}_\downarrow\, \frac{\Gamma(\alpha)}{2^{1-\alpha}}\frac{1}{\sqrt{2\pi}}\frac{1}{r^\alpha} + q^{(-1)}_\downarrow\,\frac{\Gamma(1-\alpha)}{2^\alpha}\frac{e^{-i\theta}}{\sqrt{2\pi}}\frac{1}{r^{1-\alpha}} \end{pmatrix} 
 + o(1)\,,\quad \mbox{as $r\to0^+$.}
\]
On the other hand, assuming that $\psi\in\dom\big(H_{\mathrm{D}}^{\gamma}\big)$, by similar arguments and \eqref{eq:gxi}, we infer
\[
\psiv=\phiv+\mu(\xiv_+ +e^{i\gamma}\xiv_-)=
\begin{pmatrix}
\mu(1+e^{i\gamma})\,\frac{\Gamma(1-\alpha)}{2^\alpha}\frac{e^{-i\theta}}{\sqrt{2\pi}}\frac{1}{r^{1-\alpha}} \vspace{0.1cm}\\
\mu(1-e^{i\gamma})\,\frac{\Gamma(\alpha)}{2^{1-\alpha}}\frac{1}{\sqrt{2\pi}}\frac{1}{r^\alpha}
\end{pmatrix}
+ o(1)\,,\quad \mbox{as $r \to 0^+$.}
\]
Matching the above asymptotics, we get $q^{(0)}_\uparrow=0$, $q^{(-1)}_\downarrow =0$, $q^{(-1)}_\uparrow=\mu(1+e^{i\gamma})$ and $q^{(0)}_\downarrow=\mu(1-e^{i\gamma})$. In particular, $\mu(1-e^{i\gamma}) q^{(-1)}_\uparrow = \mu(1+e^{i\gamma}) q^{(0)}_\downarrow$. Given that $q^{(-1)}_\uparrow, q^{(0)}_\downarrow \in\CC$ are independent parameters, it appears that the only admissible alternatives are the following:
\begin{enumerate}[(A{.1})]
\item $\gamma=0$,  $q^{(0)}_\uparrow=q^{(-1)}_\downarrow =q^{(0)}_\downarrow=0$ and $q^{(-1)}_\uparrow\in\CC$ arbitrary;
\item $\gamma=\pi$, $q^{(0)}_\uparrow=q^{(-1)}_\downarrow =q^{(-1)}_\uparrow=0$ and $q^{(0)}_\downarrow\in\CC$ arbitrary;
\item $\mu=0$ and $q^{(0)}_\uparrow=q^{(-1)}_\downarrow =q^{(-1)}_\uparrow=q^{(0)}_\downarrow=0$.
\end{enumerate}
It is easy to check that these are indeed equivalent to the conditions (a.1), (b.1) and (c.1), respectively.

{\sl Step 2.} Let us now prove that $H_{\mathrm{D}}^{\gamma} \big(\dom\big(H^{(\mathrm{F})}_{\mathrm P}\big)\big)
= H_{\mathrm{D}} \big(\dom\big(H^{(\mathrm{F})}_{\mathrm P}\big)\big)
\not\subseteq \dom\big(H^{\gamma'}_D\big)$ for all $\gamma,\gamma'\in[0,2\pi)$.
We first remark that $\dom\big(H^{(\mathrm{F})}_{\mathrm P}\big)$ consists of functions which are indeed regular at the origin. So, any self-ajoint realization $H^\gamma_{\mathrm{D}}$ of the Dirac operator acts on $\dom\big(H^{(\mathrm{F})}_{\mathrm P}\big)$ as the basic differential operator $H_{\mathrm{D}}$, which accounts for the first identity in the claim.
Taking this into account, we now proceed to show that $H_{\mathrm{D}}\phiv \notin \dom\big(H^{\gamma'}_D\big)$ for any given $\phiv \in\dom(H^{(F)}_{\mathrm P})$ and for any $\gamma' \in[0,2\pi)$, ultimately proving the thesis. 

Recalling that $H^{(F)}_{\mathrm P}=H^{(F)}_{\mathrm S}\oplus H^{(F)}_{\mathrm S}$, where $H^{(F)}_{\mathrm S}$ is the Friedrichs extension of the Schr\"odinger operator $H_{\mathrm S} = (-i\nabla+\mathbf{A})^2$ (see \cref{rem:HPF}), we deduce that any $\phiv \in \dom\big(H_{\mathrm{P}}^{(\mathrm{F})}\big)$ can be represented as
\begin{equation}\label{eq:Rf}
\phiv = R^{(\mathrm{F})}_{\mathrm S}(-1) \oplus R^{(\mathrm{F})}_{\mathrm S}(-1)\, \bm{f} \,,\qquad \mbox{for some $ \bm{f} =\begin{pmatrix}f_\uparrow \\ f_\downarrow \end{pmatrix}\in L^2(\mathbb{R}^2;\mathbb{C}^2)$}\,,
\end{equation}
where $R^{(\mathrm{F})}_{\mathrm{S}}(-1)$ is the resolvent of $H_{\mathrm{S}}^{(\mathrm{F})}$, acting by convolution with the integral kernel \eqref{eq:Fresolv}.
Keeping in mind that, in polar coordinates, we have
\[
H_{\mathrm{D}}= \begin{pmatrix} 0 & e^{-i\theta}(-i\partial_r - \frac{\partial_\theta+i\alpha}{r}) \\ e^{i\theta}(-i\partial_r + \frac{\partial_\theta+i\alpha}{r}) & 0 \end{pmatrix},
\]
in view of \eqref{eq:Rf} and \eqref{eq:Fresolv}, we get
\begin{multline}\label{eq:Dup}
(H_{\mathrm{D}}\phiv)_\uparrow=\sum_{k\in\Z}\left[ \left(-i\partial_r-i\tfrac{(k+\alpha)}{r}\right)K_{\vert k+\alpha\vert}( r)\right]\int^r_0 \,\diff r' r'I_{\vert k+\alpha\vert} ( r')f_\downarrow^k(r') \: \tfrac{e^{i(k-1)\theta}}{\sqrt{2\pi}}\\
 +\sum_{k\in\Z}\left[\left(-i\partial_r-i\tfrac{(k+\alpha)}{r}\right)I_{\vert k+\alpha\vert} ( r)\right]\int^{+\infty}_r  \,\diff r' r' K_{\vert k+\alpha\vert}( r')f_\downarrow^k(r') \: \tfrac{e^{i(k-1)\theta}}{\sqrt{2\pi}} \,,
 \end{multline}
 \begin{multline}\label{eq:Ddown}
 (H_{\mathrm{D}}\phiv)_\downarrow= \sum_{k\in\Z}\left[ \left(-i\partial_r+i\tfrac{(k+\alpha)}{r}\right)K_{\vert k+\alpha\vert}( r)\right]\int^r_0 \,\diff r' r'I_{\vert k+\alpha\vert} ( r')f_\uparrow^k(r') \: \tfrac{e^{i(k+1)\theta}}{\sqrt{2\pi}}\\
 +\sum_{k\in\Z}\left[\left(-i\partial_r+i\tfrac{(k+\alpha)}{r}\right)I_{\vert k+\alpha\vert} ( r)\right]\int^{+\infty}_r \,\diff r' r'K_{\vert k+\alpha\vert}( r')f_\uparrow^k(r') \: \tfrac{e^{i(k+1)\theta}}{\sqrt{2\pi}} \,,
\end{multline}
where the $f^k_s$ are the Fourier coefficients in the angular wave expansion of $f_s$.

Now, recall that the boundary conditions encoded in $\dom\big(H_{\mathrm{D}}^\gamma\big)$ entail the evaluation of the trace maps $c^s_{-\alpha},c^s_{\alpha - 1} : \dom\big(H_{\mathrm{D}}^\gamma\big) \to \CC$, see \eqref{eq:Diractraces} and \eqref{bcDir1}, \eqref{bcDir2}.
Using the known identities $(- \partial_r \pm \frac{\nu}{r})K_\nu( r) =  K_{\nu \pm 1}(r)$, $(- \partial_r \pm \frac{\nu}{r})I_\nu( r) = -  I_{\nu \pm 1}(r)$ \cite[Eq. 10.29.2]{OLBC10} and $K_{-\nu}(z) = K_{\nu}(z)$ \cite[Eq. 10.27.3]{OLBC10}, we infer
\begin{gather}
 \lf\langle (H_{\mathrm{D}}\phiv)_\uparrow \ri\rangle(r) = \tfrac{i}{\sqrt{2\pi}}\left( K_\alpha(r) \int^r_0 \!\diff r' r'I_{1+\alpha} ( r')f_\downarrow^1(r') - I_\alpha(r) \int^{\infty}_r\!\!\diff r' r'K_{1+\alpha}( r')f_\downarrow^1(r') \right), \label{eq:avgint1}
\\
 \lf\langle (H_{\mathrm{D}}\phiv)_\uparrow\, e^{i\theta} \ri\rangle(r) = \tfrac{i}{\sqrt{2\pi}}\left( K_{1-\alpha}( r) \int^r_0 \!\diff r' r' I_{\alpha} ( r')f_\downarrow^0(r') - I_{-(1-\alpha)}( r) \int^{\infty}_r\!\!\diff r' r'K_{\alpha}( r')f_\downarrow^0(r') \right),
\\
 \lf\langle (H_{\mathrm{D}}\phiv)_\downarrow \ri\rangle(r) = \tfrac{i}{\sqrt{2\pi}}\left( K_\alpha( r) \int^r_0 \!\diff r' r' I_{ 1-\alpha} ( r')f_\uparrow^{-1}(r') - I_{-\alpha}( r) \int^{\infty}_r \!\!\diff r' r'K_{1-\alpha}( r')f_\uparrow^{-1}(r') \right),
\\
 \lf\langle (H_{\mathrm{D}}\phiv)_\downarrow\, e^{i\theta} \ri\rangle(r) = \tfrac{i}{\sqrt{2\pi}}\left( K_{1-\alpha}( r) \int^r_0 \!\diff r' r' I_{ 2-\alpha} ( r')f_\uparrow^{-2}(r') - I_{1-\alpha}( r) \int^{\infty}_r \!\!\diff r' r'K_{2-\alpha}( r')f_\uparrow^{-2}(r') \right). \label{eq:avgint4}
\end{gather}
Exploiting the asymptotics of the Bessel functions $I_\nu,K_\nu$ \cite[Eq. 10.30.1-2]{OLBC10}, it can be checked that
\begin{gather*}
\left\vert \int^r_0 \!\diff r' r'I_{ 1+\alpha} ( r')f_\downarrow^1(r') \right\vert\leq \lf\Vert f_\downarrow \ri\Vert_2\left( \int^r_0 \!\diff r' r'\,\big|I_{ 1+\alpha} ( r')\big|^2\right)^{\!1/2}\leq C\, r^{2+\alpha}\,,
\\
\left\vert\int^{\infty}_r \!\!\diff r' r' K_{1+\alpha}( r')f_\downarrow^1(r') \right\vert\leq  \lf\Vert f_\downarrow\ri\Vert_2 \left(\int^{\infty}_r \!\!\diff r' r'\,\big|K_{1+\alpha}( r')\big|^2 \right)^{\!1/2}\leq C\,r^{-\alpha}\,;
\\
\left\vert \int^r_0 \!\diff r' r' I_{\alpha} ( r')f_\downarrow^0(r') \right\vert\leq \lf\Vert f_\downarrow \ri\Vert_2\left( \int^r_0 \!\diff r' r'\,\big|I_{ \alpha} ( r')\big|^2\right)^{\!1/2}\leq C\, r^{1+\alpha}\,,
\\
\left\vert \int^r_0 \!\diff r' r'I_{ 1-\alpha} ( r')f_\uparrow^{-1}(r') \right\vert\leq \lf\Vert f_\uparrow \ri\Vert_2\left( \int^r_0 \!\diff r' r'\,\big|I_{1-\alpha} ( r')\big|^2\right)^{\!1/2}\leq C\, r^{2-\alpha}\,,
\\
\left\vert \int^r_0 \!\diff r' r'I_{ 2-\alpha} ( r')f_\uparrow^{-2}(r') \right\vert\leq \lf\Vert f_\uparrow \ri\Vert_2\left( \int^r_0 \!\diff r' r'\,\big|I_{ 2-\alpha} ( r')\big|^2\right)^{\!1/2}\leq C\, r^{3-\alpha}\,,
\\
\left\vert\int^{\infty}_r \!\!\diff r' r' K_{2-\alpha}( r')f_\uparrow^{-2}(r') \right\vert\leq  \lf\Vert f_\uparrow\ri\Vert_2 \left(\int^{\infty}_r \!\!\diff r' r'\,\big|K_{2-\alpha}( r')\big|^2 \right)^{\!1/2}\leq C\,r^{-(1-\alpha)}\,.
\end{gather*}
On the other hand, keeping in mind that $\bm{f} = (H_{P} + 1) \phiv$, see \eqref{eq:Rf}, recalling the definition of the trace maps $\tau_s^{(\ell)} : \dom\big(H^{(\mathrm{F})}_{\mathrm P}\big) \to \CC$, see \eqref{eq:traceop}, integrating by parts and using the asymtptotics \eqref{eq:G01asy}, we obtain
\begin{multline*}
\int^{\infty}_r \!\!\diff r' r' K_{\alpha}( r')f_\downarrow^0(r') 
= \, \braketl{ G_{\downarrow}^{(0)}}{\bm{f}}_{L^2(\R^2 \setminus B_r(\bm{0}))}
= \sum_{s' \in \{\uparrow,\downarrow\}} \int_{\partial B_r(\bm{0})}\! \diff\Sigma_r\; \big[ \big(G_{\downarrow}^{(0)}\big)_{s'} \partial_r \phi_{s'} - \partial_r\big(G_{\downarrow}^{(0)}\big)_{s'}  \phi_{s'} \big]
\\
=  \int_{0}^{2\pi}\!\! r\diff\theta\; \big[ g^{(0)} \partial_r \phi_{\downarrow} - \partial_r g^{(0)}  \phi_{\downarrow} \big] 
= \tfrac{\Gamma( \alpha ) }{ 2^{1-\alpha}}\, \lim_{r \to 0^+} \tfrac{1}{r^{\alpha}} \int_{0}^{2\pi}\!\! \diff\theta\; (\alpha + r \partial_r) \phi_{\downarrow}\,\tfrac{1}{ \sqrt{2\pi}} + o(1)
=  \tau_{\downarrow}^{(0)} \phiv + o(1)\,,
\end{multline*}
as $r \to 0^{+}$. A similar computation yields
\[
\int^{\infty}_r \!\!\diff r' r' K_{1-\alpha}( r')f_\uparrow^{-1}(r') = \, \tau_{\uparrow}^{(-1)} \phiv + o(1) \,.
\]
Combining the above results, we find
\begin{gather*}
c^{\uparrow}_{\alpha-1}(H_{\mathrm{D}}\phiv)
	=\lim_{r\to0^+}r^{1-\alpha}\,\langle (H_{\mathrm{D}}\phiv)_{\uparrow}\, e^{i\theta}\rangle 
	= -\, \tfrac{i\,2^{1-\alpha}}{\sqrt{2\pi}\,\Gamma(\alpha)} \,  \tau_{\downarrow}^{(0)} \phiv\,, \\
c^{\downarrow}_{-\alpha}(H_{\mathrm{D}}\phiv)
	= \lim_{r\to0^+}r^\alpha\,\langle (H_{\mathrm{D}}\phiv)_{\downarrow}\rangle
	= - \,\tfrac{i\,2^{\alpha}}{\sqrt{2\pi}\,\Gamma(1-\alpha)}\, \tau_{\uparrow}^{(-1)} \phiv \,,
\end{gather*}
which imply, in turn,
\begin{equation}\label{eq:Hdphiv}
H_{\mathrm{D}}\phiv =
\begin{pmatrix}
-i\,\frac{2^{1-\alpha}}{\Gamma(\alpha)}\,[\tau^{(0)}_\downarrow\phiv]\frac{e^{-i\theta}}{\sqrt{2\pi}}\frac{1}{r^{1-\alpha}} \\
-i\, \frac{2^{\alpha}}{\Gamma(1-\alpha)}\,[\tau^{(-1)}_\uparrow\phiv]\frac{1}{\sqrt{2\pi}}\frac{1}{r^{\alpha}}
\end{pmatrix}+o(1)\,.
\end{equation}
On the other hand, if $H_{\mathrm{D}}\phiv$ were to belong to $\dom(H^{\gamma'}_D)$, there should exist some $\mu' \in \CC$ such that
\[
H_{\mathrm{D}}\phiv=
\begin{pmatrix}
\mu'(1+e^{i\gamma'})\frac{\Gamma(1-\alpha)}{2^\alpha}\frac{e^{-i\theta}}{\sqrt{2\pi}}\frac{1}{r^{1-\alpha}} \\
\mu'(1-e^{i\gamma'})\frac{\Gamma(\alpha)}{2^{1-\alpha}}\frac{1}{\sqrt{2\pi}}\frac{1}{r^{\alpha}} 
\end{pmatrix}+o(1)\,.
\]
Matching the coefficients in the above expansions, we would obtain $\sin(\gamma'/2)\,\tau^{(0)}_\downarrow\phiv = i\,\cos(\gamma'/2)\,\tau^{(-1)}_\uparrow \phiv$, which is absurd, given that the traces $\tau^{(0)}_\downarrow \phiv,\tau^{(-1)}_\uparrow \phiv\in\CC$ are independent for a generic $\phiv\in\dom(H^{(F)}_{\mathrm P})$. 

{\sl Step 3}. We now prove that $\dom(H^{(\bb)}_{\mathrm P})\subseteq \dom((H_{\mathrm{D}}^{\gamma})^2)$ if and only if one of the following alternatives holds:
\begin{enumerate}[(a{.3})]
\item condition (a.1) in Step 1 holds and $\bb_{\uparrow \uparrow}^{(-1 -1)} = 0$;

\item condition (b.1) in Step 1 holds and $\bb_{\downarrow \downarrow}^{(0 0)} = 0$.
\end{enumerate}

Let us firstly remark that, considering the basic definition $\dom\big((H_{\mathrm{D}}^{\gamma})^2\big) \!=\! \{\psiv \in\dom(H_{\mathrm{D}}^{\gamma})\,|\, H_{\mathrm{D}}^{\gamma}\psiv \in\dom(H_{\mathrm{D}}^{\gamma})\}$, the claim proved in \emph{Step 2} ensures that $\dom \big(H^{(\mathrm{F})}_{\mathrm{P}}\big)$ is not a subset of $\dom\big((H_{\mathrm{D}}^{\gamma})^2\big)$ for any $\gamma \in [0,2\pi)$. This rules out the alternative (c.1) in \emph{Step 1}. So, we only have to examine alternatives (a.1),  (b.1) therein. As an example, hereafter we discuss case (a.1).

Assume $\dom(H^{(\bb)}_{\mathrm P})\subseteq \dom\big((H^0_{\mathrm{D}})^2\big)$, for some suitable $\bb \in \mathrm{M}_{4,\mathrm{Herm}}(\mathbb{C})$ fulfilling condition (a.1) in \emph{Step 1}. On one hand, for any $\psiv \in\dom(H^{(\bb)}_{\mathrm P})$, we have
\[
\psiv =\phiv +q^{(-1)}_\uparrow \giv^{(-1)}_\uparrow
= \phiv +\,q^{(-1)}_\uparrow\,\tfrac{\xiv_+ + \xiv_- }{2\sqrt{2\pi}}\;,
\]
where $\phiv \in\dom\big(H^{(\mathrm{F})}_{\mathrm P}\big) \subset \dom\big[Q^{(\mathrm{F})}_{\mathrm P}\big]$. Recalling the definition of $H^0_{\mathrm{D}}$, see \cref{thm:Ddom}, and using the asymptotic expansion \eqref{eq:Hdphiv}, we obtain
\[
H^0_{\mathrm{D}}\psiv = H_{\mathrm{D}}\phiv + \tfrac{i}{2}\,q^{(-1)}_\uparrow\,(\xiv_+ - \xiv_-)
= \begin{pmatrix}
-i\,\frac{2^{1-\alpha}}{\Gamma(\alpha)}\big(\tau^{(0)}_\downarrow\phiv\big)\frac{e^{-i\theta}}{\sqrt{2\pi}}\,\frac{1}{r^{1-\alpha}} \\
\left[-i\,\frac{2^\alpha}{\Gamma(1-\alpha)}\big(\tau^{(-1)}_\uparrow\phiv\big) + i\frac{\Gamma(\alpha)}{2^{1-\alpha}}\,q^{(-1)}_\uparrow \right]\frac{1}{\sqrt{2\pi}}\,\frac{1}{r^\alpha} \end{pmatrix}+o(1)\,, \quad \mbox{as $r \to 0^+$}.
\]
On the other hand, since $H^0_{\mathrm{D}}\psiv\in\dom(H^0_{\mathrm{D}})$, there must exist some $\tilde{\phiv} \in \dom\big[Q^{(\mathrm{F})}_{\mathrm P}\big]$ and $\tilde{\mu} \in \CC$ such that
\[
H^0_{\mathrm{D}}\psiv = \tilde{\phiv} + \tilde{\mu}(\xiv_++\xiv_-)=
\begin{pmatrix} \tilde{\mu}\, 2^{1-\alpha}\,\Gamma(1-\alpha)\,\frac{e^{-i\theta}}{\sqrt{2\pi}} \,\frac{1}{r^{1-\alpha}} \\
0
\end{pmatrix}+o(1)\,, \qquad \mbox{as $r \to 0^+$}.
\]
Matching the $\downarrow$\! -\! component in the above asymptotics, we get
\[
\tau^{(-1)}_\uparrow\phiv = \frac{\pi}{2\sin(\pi\alpha)}\, q^{(-1)}_\uparrow\,.
\]
At the same time, the boundary condition in \eqref{HBdom} (see also \eqref{eq:traceop}) entails
\[
\tau^{(-1)}_\uparrow\phiv=\left[(L(1)+\bb)\qv \right]^{(-1)}_\uparrow=\left[\frac{\pi}{2\sin(\pi\alpha)}+\bb_{\uparrow \uparrow}^{(-1,-1)} \right]q^{(-1)}_\uparrow\,.
\]
The above identities yield $\bb_{\uparrow \uparrow}^{(-1,-1)}=0$, which ultimately proves claim (b.3).

{\sl Step 4.}
To infer the thesis, we need to show that the inclusion $\dom\big(H^{(\bb)}_{\mathrm P}\big)\subseteq \dom\big((H_{\mathrm{D}}^{\gamma})^2\big)$, proved in the preceding \emph{Step 3} for suitable choices of $\gamma$ and $\beta$, is indeed an equality.
As a matter of fact, the operators $H^{(\bb)}_{\mathrm P}$ and $(H_{\mathrm{D}}^{\gamma})^2$, with $\gamma$ and $\beta$ fixed as in \emph{Step 3}, have to coincide because they are self-adjoint extensions of the same closable symmetric operator, namely,
$\big( \bm{\sigma}\cdot(-i\nabla+\mathbf{A})\big)^2$ on $C^\infty_c(\R^2\setminus\{\bm{0}\})$.
\end{proof}

	\begin{proof}[Proof of \cref{prop: zeroenresDirac}]
	Proceeding as in the proof of \cref{prop: zeroenresTheta}, it can be checked that the only distributional solutions of the zero-energy equation $H_{\mathrm{D}} \psiv = 0$ are of the form
	\be
		\psiv(r,\teta) = \begin{pmatrix}
			\sum_{\ell \in \mathbb{Z}} d_{\uparrow}^{(\ell)}\,\frac{e^{i \ell \teta}}{\sqrt{2\pi}}\, r^{\ell + \alpha} \vspace{0.1cm} \\
			\sum_{\ell \in \mathbb{Z}} d_{\downarrow}^{(\ell)}\,\frac{e^{i \ell \teta}}{\sqrt{2\pi}}\, r^{-(\ell + \alpha)}
		\end{pmatrix}\,,
	\ee
with suitable coefficients $\mathbf{d} \in \mathbb{C}^4$. The condition of uniform boundedness at infinity forces us to fix $d_{\uparrow}^{(\ell)} = 0$ for $\ell \geqslant 0$ and $d_{\downarrow}^{(\ell)} = 0$ for $\ell \leqslant -1$. On the other hand, the condition $\psiv \in L^2_{\mathrm{loc}}(\R^2)$ entails $d_{\uparrow}^{(\ell)} = 0$ for $\ell \leqslant - 2$ and $d_{\downarrow}^{(\ell)} = 0$ for $\ell \geqslant 1$. Summing up, the only admissible zero-energy resonances have the form
	\be
		\psiv(r,\teta) = \begin{pmatrix}
			d_{\uparrow}^{(-1)}\,\frac{e^{-i \teta}}{\sqrt{2\pi}}\, \frac{1}{r^{1 - \alpha}} \vspace{0.1cm} \\
			d_{\downarrow}^{(0)}\,\frac{1}{\sqrt{2\pi}}\, \frac{1}{r^{\alpha}}
		\end{pmatrix}\,.
	\ee
Then, it can be easily checked that the boundary conditions associated to $\dom(H_{\mathrm{D}}^{(\gamma)})$, see \eqref{eq:Ddom2}, are verified if and only if
$$
d_\uparrow^{(-1)} = i\cot(\gamma/2)\, \tfrac{2^{1-2\alpha}\Gamma(1-\alpha)}{\Gamma(\alpha)}\, d_\downarrow^{(0)}\,,
$$
which concludes the proof.
	\end{proof}

\appendix

\section{Symmetries of 2D Pauli and Dirac operators}\label{app:2Dsym}

We consider here generic Pauli and Dirac operators of the form
\begin{gather}
	H_{\mathrm{P}}(\mathbf{A}) := \big(\bm{\sigma} \cdot (-i \nabla + \mathbf{A})\big)^2 = \sigma_{j}\, \sigma_{\ell}\,(-i \partial_j + A_j)(-i \partial_\ell + A_\ell) \,, \label{eq:Pauli2exp} \\
	H_{\mathrm{D}}(\mathbf{A}) := \bm{\sigma} \cdot (-i \nabla + \mathbf{A}) = \sigma_{j} \,(-i \partial_j + A_j)\,, \label{eq:Dirac2exp}
\end{gather}
where $\mathbf{A} = (A_1,A_2)$ is any vector-valued distribution in $\mathbb{R}^2$ and we are using Einstein's convention to sum over the repeated indices $j,\ell \in \{1,2\}$. Let us stress that here we are uniquely concerned with the algebraic features of $H_{\mathrm{P}}(\mathbf{A})$ and $H_{\mathrm{D}}(\mathbf{A})$. Accordingly, we neglect all domain and self-adjointness issues.

We henceforth proceed to classify the transformations which leave the structure of $H_{\mathrm{P}}(\mathbf{A})$ and $H_{\mathrm{D}}(\mathbf{A})$ invariant. More precisely, we consider (anti-)linear transformations involving both spin and coordinate degrees of freedom of the form
\begin{gather}
	 \UU(S,T,\xv_0) : L^2(\mathbb{R}^2;\mathbb{C}^2) \to L^2(\mathbb{R}^2;\mathbb{C}^2)\,, \qquad (\UU \psi)(\xv) = S(\xv)\, \psi(T^{-1} \xv - \xv_0)\,; \label{eq:USTx0} \\
	 \VV(S,T,\xv_0) : L^2(\mathbb{R}^2;\mathbb{C}^2) \to L^2(\mathbb{R}^2;\mathbb{C}^2)\,, \qquad (\VV \psi)(\xv) = S(\xv)\, \psi^*(T^{-1} \xv - \xv_0)\,. \label{eq:VSTx0}
\end{gather}
Here $ S(\xv) : \mathbb{R}^2 \to GL(2,\mathbb{C})$ is any smooth section of the trivial fiber bundle $\mathbb{R}^2 \times GL(2,\mathbb{C})$, $T   \in GL(2,\mathbb{R})$ is any constant matrix, and $\xv_0 \in \mathbb{R}^2$ is any fixed vector.
In the forthcoming subsections \cref{subsec:transfP} and \cref{subsec:transfD} we identify all (anti-)unitary operators $\WW$ of the form \eqref{eq:USTx0} or \eqref{eq:VSTx0} fulfilling the following identities, for some suitable vector-valued distribution $\widetilde{\mathbf{A}} = \big(\widetilde{A}_1,\widetilde{A}_2\big)$:
\begin{gather}
	\WW H_{\mathrm{P}}(\mathbf{A})\,\WW^{-1} = H_{\mathrm{P}}\big(\widetilde{\mathbf{A}}\big) \,, \label{eq:VHPAV} \\
	\WW H_{\mathrm{D}}(\mathbf{A})\,\WW^{-1} = H_{\mathrm{D}}\big(\widetilde{\mathbf{A}}\big)\,. \label{eq:VHDAV}
\end{gather}
Note that local coordinates transformations always modify the structure of the differential operators $H_{\mathrm{P}}(\mathbf{A})$ and $H_{\mathrm{D}}(\mathbf{A})$. In fact, non-affine transformations always introduce non-trivial curvature contributions. For this reason we restrict the attention to transformations which are affine in the space coordinates and local in the spin degree of freedom.

\begin{remark}[Symmetries]
\mbox{}	\\
We say that a (anti-)unitary operator $\WW$ is a (physical) \textsl{symmetry} of the differential operators $H_{\mathrm{P}}(\mathbf{A})$ and $H_{\mathrm{D}}(\mathbf{A})$ if the associated magnetic field $b := \mathrm{curl} \, \mathbf{A} = \partial_1 A_2 - \partial_2 A_1$ remains invariant, namely,
\bdm
	\widetilde{b} = \mathrm{curl} \, \widetilde{\mathbf{A}} =  \mathrm{curl} \,\mathbf{A} = b\,.
\edm
\end{remark}

\subsection{Symmetries of the Pauli Hamiltonian}\label{subsec:transfP}
Let us first focus on the Pauli operator \eqref{eq:Pauli2exp} and classify all transformations of the form \eqref{eq:USTx0} and \eqref{eq:VSTx0} fulfilling \eqref{eq:VHPAV}.

\subsubsection{Linear transformations}
For any map $\UU$ of the form \eqref{eq:USTx0}, by a direct calculation we obtain
\bmln{
	\UU\, H_{\mathrm{P}}(\mathbf{A})\, \UU^{-1} 
	= S(\xv)\, \sigma_j\, \sigma_\ell\, {S}^{-1}(\xv)\, T_{hj} T_{m \ell}\; \times \\
	\times \big(-i \partial_{h} - i S(\xv) \partial_h {S}^{-1}(\xv) + T^{-1}_{k h}A_{k}(T^{-1}\xv - \xv_0)\big)\big(-i \partial_{m} - i S(\xv) \partial_m {S}^{-1}(\xv) + T^{-1}_{m n} A_{m}(T^{-1}\xv - \xv_0) \big)\,.
}
This shows that $\UU H_{\mathrm{P}}(\mathbf{A})\, \UU^{-1}$ is itself a Pauli operator of the form \eqref{eq:Pauli2exp}, if and only if the following two conditions are simultaneously fulfilled, for some suitable vector potential $\widetilde{\mathbf{A}}$:
\begin{gather}
S(\xv)\, \sigma_j\, \sigma_\ell\, {S}^{-1}(\xv)\, T_{hj} T_{m \ell} = \sigma_{h}\, \sigma_{m}\,; \label{eq:changeS} \\
T^{-1}_{k h}A_{k}(T^{-1}\xv - \xv_0)\,\one - i\, S(\xv) \,\partial_{h} {S}^{-1}(\xv) = \widetilde{A}_{h}(\xv)\,\one\,. \label{eq:changeA}
\end{gather}

Using the basic algebraic identities $\{\sigma_{h},\sigma_{m}\} = 2 \delta_{hm} \one$ and $[\sigma_{h}, \sigma_{m}] = 2i \varepsilon_{h m r} \sigma_r$, from \eqref{eq:changeS} we deduce $T_{hj} T_{mj} = \delta_{hm}$ and $S(\xv)\, \sigma_{\mathrm{P}}\, {S}^{-1}(\xv)\, \varepsilon_{j \ell p}\, T_{hj} T_{m \ell} = \varepsilon_{h m r}\, \sigma_r$. These are equivalent to
\begin{gather}
T\, T^{t}\! = \one\,, \label{eq:TTone}\\
S(\xv) \,\sigma_3\, {S}^{-1}(\xv) = \big(\!\det T^{-1}\big)\,\sigma_3\,. \label{eq:Ssigma3S}
\end{gather}
Recalling that we are assuming $T \in GL(2,\mathbb{R})$, \eqref{eq:TTone} clearly entails $T \in O(2,\mathbb{R})$. As a consequence, we see that $\UU$ is unitary in $L^2(\mathbb{R}^2;\mathbb{C}^2)$ only if $S = S(\xv)$ is a smooth section of $\mathbb{R}^2 \times \mathrm{U}(2,\mathbb{C})$. Notice that by a trivial application of Stone's theorem we have $\mathrm{U}(2,\mathbb{C}) = \{ e^{-i (\eta_0 \one + \eta_1 \sigma_1 + \eta_2 \sigma_2 + \eta_3 \sigma_3)}\,|\, \eta_0,\eta_1,\eta_2,\eta_3 \in \mathbb{R} \}$. Keeping in mind that $\det T^{-1} = \det T = \pm 1$ and using elementary algebraic properties of the Pauli matrices, we then infer that \eqref{eq:Ssigma3S} can be fulfilled only if either of the following two alternatives occurs:
\beq
T \in SO(2,\mathbb{R}), \qquad  S = e^{-i \eta_0 \one - i \eta_3 \sigma_3}, \label{eq:TSO2loc}
\eeq
for some $ \eta_0,\eta_3 \in C^{\infty}(\mathbb{R}^2) $, or  
\beq
T \in O(2,\mathbb{R}) \setminus SO(2,\mathbb{R}),	\qquad S = e^{-i \eta_0 \one - i \eta_1 \sigma_1 - i \eta_2 \sigma_2},\label{eq:TO2loc}
\eeq
for some $ \eta_0,\eta_1,\eta_2 \in C^{\infty}(\mathbb{R}^2) $, such that 
\beq
	\label{condition}
	\sqrt{\eta_1^2(\xv) + \eta_2^2(\xv)} = \tfrac{\pi}{2}, \tfrac{3\pi}{2}, \tfrac{5\pi}{2}, \ldots \;.
\eeq

On the other hand, matching the condition \eqref{eq:changeA} requires that
\be\label{SSh}
S(\xv)\, \partial_h {S}^{-1}(\xv) = s_h(\xv)\,\one \,,  
\ee
for some scalar function $s_h \in C^{\infty}(\mathbb{R}^2) $.
Taking \eqref{eq:TSO2loc} and \eqref{eq:TO2loc} into account, in can be checked that \eqref{SSh} can be fulfilled only if either $\eta_{3} $ or $\eta_{1},\eta_{2} $ are constant, respectively. Summing up, the only transformations $U$ of the form \eqref{eq:USTx0} satisfying \eqref{eq:VHPAV}, for some suitable $\widetilde{\mathbf{A}}$, correspond to
\beq
T \in SO(2,\mathbb{R}), \qquad  S = e^{-i \eta_0 \one - i \eta_3 \sigma_3},  \label{eq:TSO2}
\eeq
for some $ \eta_0 \in C^{\infty}(\mathbb{R}^2), \eta_3 \in \mathbb{R} $, or
\beq
T \in O(2,\mathbb{R}) \setminus SO(2,\mathbb{R}), \qquad  S = e^{-i \eta_0 \one - i \eta_1 \sigma_1 - i \eta_2 \sigma_2},  \label{eq:TO2}
\eeq
for some $ \eta_0 \in C^{\infty}(\mathbb{R}^2) $ and  $ \eta_1,\eta_2 \in \R $ satisfying \eqref{condition}.

\begin{remark}[Special symmetries]
\mbox{}	\\
In both cases \eqref{eq:TSO2} and \eqref{eq:TO2}, the identity \eqref{eq:changeA} reduces to
\bdm
\widetilde{\mathbf{A}}(\xv) = T\, \mathbf{A}(T^{-1}\xv - \xv_0) + \nabla \eta_0(\xv)\,,
\edm
which implies, in turn,
\begin{align*}
	\widetilde{b}(\xv) 
	= (\det T)\, b(T^{-1}\xv - \xv_0)\,.
\end{align*}
This shows that, in general, $\UU $ is a symmetry of $H_{\mathrm{P}}(\mathbf{A})$ only if $T = \one$, $\xv_0 = \mathbf{0}$ and, in compliance with \eqref{eq:TSO2}, $S = e^{-i \eta_0 \one - i \eta_3 \sigma_3}$ for some $\eta_0 \in C^{\infty}(\mathbb{R}^2)$, $\eta_3 \in \mathbb{R}$. Of course, whenever the magnetic field exhibits specific features, the class of symmetries of the model could comprise additional transformations. For instance, it appears that a uniform magnetic field $b = \mbox{const}. $ is left invariant by any transformation $\UU $ of the form \eqref{eq:USTx0} with $T$, $S$ as in \eqref{eq:TSO2} and $\xv_0 \in \mathbb{R}^2$.
\end{remark}

\begin{example}
	For $S = e^{-i \eta_0}$ with $\eta_0 \in C^{\infty}(\mathbb{R}^2)$, $T = \one$ and $\xv_0 = \mathbf{0}$, $\UU$ is the standard (local) $\mathrm{U}(1)$ electromagnetic gauge transformation. In this case we have $\widetilde{\mathbf{A}} = \mathbf{A}+\nabla\eta_0$.
\end{example}

\begin{example}
	For $S = e^{-i \eta_3 \sigma_3}$ with $\eta_3 \in \mathbb{R}$, $T = \one$ and $\xv_0 = \mathbf{0}$, $\UU$ is the $\mathrm{U}(1)$ (global) axial gauge transformation. In this case we have $\widetilde{\mathbf{A}} = \mathbf{A}$.
\end{example}

\begin{example}
	For $S = \one$, $T \in SO(2,\mathbb{R})$ and $\xv_0 = \mathbf{0}$, $\UU$ is a simple rotation of the coordinate system. In this case, we have $\widetilde{\mathbf{A}}(\xv) = T\, \mathbf{A}(T^{-1}\xv - \xv_0)$ and $\widetilde{b}(\xv) = b(T^{-1}\xv)$. Of course, any such transformation is a symmetry of the Hamiltonian $H_{\mathrm{P}}(\mathbf{A})$ whenever the magnetic field is radial, \emph{i.e.}, $b(\xv) = b\big(|\xv|\big)$.
\end{example}

\subsubsection{Anti-linear transformations}\label{subsec:alinP}
For any map $\VV$ as in \eqref{eq:VSTx0}, taking into account that both the vector potential $\mathbf{A} = (A_1,A_2)$ and the matrix elements $T_{j h}$ are real, we get
\bmln{
	\VV\, H_{\mathrm{P}}(\mathbf{A})\, \VV^{-1} 
	= S(\xv) \,{\sigma_j}^*\, {\sigma_\ell}^*\, {S}^{-1}(\xv)\, T_{hj} T_{m \ell}\; \times \\
	\times \big(-i \partial_{h} - i S(\xv) \partial_h {S}^{-1}(\xv) - T^{-1}_{k h}A_{k}(T^{-1}\xv - \xv_0)\big)\big(-i \partial_{m} - i S(\xv) \partial_m {S}^{-1}(\xv) - T^{-1}_{m n} A_{m}(T^{-1}\xv - \xv_0) \big)\,.
}
Therefore, for $\VV\, H_{\mathrm{P}}(\mathbf{A})\, \VV^{-1}$ to fulfill \eqref{eq:VHPAV}, it is necessary and sufficient that the following conditions are both verified, for some vector-valued distribution $\widetilde{\mathbf{A}} = (\widetilde{A}_1,\widetilde{A}_2)$:
\begin{gather}
S(\xv)\, \sigma_j^*\, \sigma_\ell^*\, \mathrm{S}^{-1}(\xv) T_{h j} T_{m \ell} = \sigma_{h} \sigma_{m}\,, \label{eq:changeSanti} \\
T^{-1}_{k h}A_{k}(T^{-1}\xv - \xv_0)\,\one + i\, S(\xv) \,\partial_{h} {S}^{-1}(\xv) = -\,\widetilde{A}_{h}(\xv)\,\one\,. \label{eq:changeAanti}
\end{gather}
Using again basic commutation relations for the Pauli matrices and noting that $\sigma_3 = \sigma_3^*$, from \eqref{eq:changeSanti} we deduce
\begin{gather}
T\, T^{t}\! = \one\,, \nonumber \\
S(\xv) \,\sigma_3\, {S}^{-1}(\xv) = -\, \big(\!\det T^{-1}\big)\,\sigma_3\,. \label{eq:Ssigma3Santi}
\end{gather}
In particular we have $T \in O(2,\mathbb{R})$. So, $|\det T| = 1$ and $\VV$ is anti-unitary in $L^2(\mathbb{R}^2,\mathbb{C}^2)$ only if $S = S(\xv)$ is a smooth section of $\mathbb{R}^2 \times \mathrm{U}(2,\mathbb{C})$. Taking this into account, by arguments similar to those outlined in the previous subsection, we deduce that \eqref{eq:Ssigma3Santi} can be fulfilled if and only if either one of the following two alternatives happens:
\beq
T \in SO(2,\mathbb{R}),	\qquad S = e^{-i \eta_0 \one - i \eta_1 \sigma_1 - i \eta_2 \sigma_2}, 
 \label{eq:TSO2anti}
\eeq
for some $ \eta_0 \in C^{\infty}(\mathbb{R}^2), \eta_1,\eta_2 \in \mathbb{R} $, such that \eqref{condition} holds, or
\beq
T \in O(2,\mathbb{R}) \setminus SO(2,\mathbb{R}), \qquad	S = e^{-i \eta_0 \one - i \eta_3 \sigma_3}, \label{eq:TO2anti}
\eeq
for some $ \eta_0 \in C^{\infty}(\mathbb{R}^2), \eta_3 \in \mathbb{R} $.

\begin{remark}[Special symmetries]
\mbox{}	\\
In both cases \eqref{eq:TSO2anti} and \eqref{eq:TO2anti}, the condition \eqref{eq:changeAanti} reduces to
\be\label{eq:Aalin}
\widetilde{\mathbf{A}}(\xv) = - \,T\, \mathbf{A}(T^{-1}\xv-\xv_0) + \nabla \eta_0(\xv)\,,
\ee
which yields
\begin{align*}
	\widetilde{b}(\xv) 
	= -\, (\det T)\, b(T^{-1}\xv - \xv_0)\,.
\end{align*}
This makes evident that, in general, none of the admissible anti-unitary transformations \eqref{eq:VSTx0} leaves the magnetic field invariant. Nevertheless, $\VV$ can be a symmetry if $b$ possesses specific features. For example, if the magnetic field is radial, namely $b(\xv) = b\big(|\xv|\big)$, then any transformation $\VV $ of the form \eqref{eq:VSTx0} with $T,S$ as in \eqref{eq:TO2anti} and $\xv_0 = \mathbf{0}$ is indeed a symmetry.
\end{remark}

\begin{example}
	For $S = \sigma_1$, $T = \one$ and $\xv_0 = \mathbf{0}$, $\VV$ is the so-called spin-flip transformation. In this case $\widetilde{\mathbf{A}}(\xv) = -\mathbf{A}(\xv)$ and $\widetilde{b}(\xv) = -\, b(\xv)$, showing that $\VV$ coincides with the charge conjugation as well.
\end{example}

\begin{example}
	For $S = \sigma_2$, $T = - \one \in SO(2,\mathbb{R})$ (notice that this choice of $T$ describes a rotation of an angle $\pi$ in the plane) and $\xv_0 = \mathbf{0}$, $\VV$ is the so-called Kramers map (see, \emph{e.g.}, \cite{LMS09}). In this case $\widetilde{\mathbf{A}}(\xv) = -\mathbf{A}(-\xv)$ and $\widetilde{b}(\xv) = -\, b(-\xv)$.
\end{example}

\begin{example}
	For $S = \one$, $T = P \in O(2,\mathbb{R}) \setminus SO(2,\mathbb{R})$ with $P(x,y) = (-x,y)$, and $\xv_{0} = \mathbf{0}$, $\VV$ is the CP transformation. In this case $\widetilde{\mathbf{A}}(\xv) = \big(A_1(-x,y),-A_2(-x,y)\big)$ and $\widetilde{b}(x,y) = b(-x,y)\,.$
\end{example}

\subsection{Symmetries of Dirac Hamiltonians}\label{subsec:transfD}
We now consider the Dirac operator \eqref{eq:Dirac2exp} and characterize all transformations of the form \eqref{eq:USTx0} and \eqref{eq:VSTx0} fulfilling \eqref{eq:VHDAV}. Since $\big[H_{\mathrm{D}}(\mathbf{A})\big]^2 = H_{\mathrm{P}}(\mathbf{A})$ at the pure algebraic level, it appears that any admissible symmetry of $H_{\mathrm{D}}(\mathbf{A})$ must also be a symmetry of $H_{\mathrm{P}}(\mathbf{A})$. Accordingly, in the sequel we shall restrict the attention to the family of transformations classified before in \eqref{eq:TSO2} \eqref{eq:TO2} and \eqref{eq:TSO2anti} \eqref{eq:TO2anti}. 
For later reference let us recall the well-known Rodrigues rotation formula, holding for arbitrary $\bm{\eta} = (\eta_1,\eta_2,\eta_3)\in \mathbb{R}^3 \setminus \{\mathbf{0}\}$ with $\eta := |\bm{\eta}|$:
\be\label{eq:rodriguez}
	e^{-i \bm{\eta} \cdot \bm{\sigma}} \bm{\sigma} \, e^{i \bm{\eta} \cdot \bm{\sigma}} = \cos(2\eta)\, \bm{\sigma} - \tfrac{1}{\eta} \sin(2\eta)\,\bm{\eta} \land \bm{\sigma}  +\tfrac{1}{\eta^2} \big(1-\cos(2\eta)\big)\, \lf( \bm{\eta} \cdot \bm{\sigma} \ri) \bm{\eta}.
\ee

\subsubsection{Linear transformations}\label{subsubUDir}
We first notice that, for any map $\UU$ of the form \eqref{eq:USTx0}, there holds
\bdm
	\UU H_{\mathrm{D}}(\mathbf{A})\, \UU^{-1} 
	= S(\xv)\, \sigma_j\, {S}^{-1}(\xv)\, T_{hj} \big(-i \partial_{h} - i S(\xv) \partial_h {S}^{-1}(\xv) + T^{-1}_{k h}A_{k}(T^{-1}\xv - \xv_0)\big)\,.
\edm
Therefore $\UU H_{\mathrm{D}}(\mathbf{A})\, \UU^{-1}$ is itself a Dirac Hamiltonian of the form \eqref{eq:Dirac2exp} if and only if, for some suitable $\widetilde{\mathbf{A}} = (\widetilde{A}_1,\widetilde{A}_2)$, \eqref{eq:changeA} is verified and
\be
S(\xv)\, \sigma_j\, {S}^{-1}(\xv)\, T_{hj} = \sigma_{h} \,. \label{eq:changeSdirac}
\ee
Let us stress that \eqref{eq:changeSdirac} implies \eqref{eq:changeS}. We also recall that the only admissible maps $\UU$ fulfilling \eqref{eq:VHPAV} certainly belong to either of the two families characterized in \eqref{eq:TSO2} and \eqref{eq:TO2}.

On one side, consider any choice of $T$ and $S$ as in \eqref{eq:TSO2}. Making reference to the parametrization 
\bdm
	T = \begin{pmatrix} \cos \varsigma & - \sin \varsigma \\ \sin\varsigma & \cos\varsigma \end{pmatrix}, \qquad \mbox{with\; $\varsigma \in [0,2\pi)$}\,,
\edm
and using \eqref{eq:rodriguez}, by a trivial relabeling of the indexes from \eqref{eq:changeSdirac} we infer, for $j \in \{1,2\}$,
\begin{align*}
\cos(2\eta_3)\,\sigma_{j} + \sin(2\eta_3)\,\varepsilon_{3 j \ell}\, \sigma_{\ell} 
= e^{ - i \eta_3 \sigma_3} \, \sigma_j\, e^{i \eta_3 \sigma_3} 
= T^{t}_{j \ell}\, \sigma_{\ell}
= (\cos\varsigma)\,\sigma_{j} + (\sin\varsigma)\;\varepsilon_{3 j \ell}\, \sigma_{\ell} \,.
\end{align*}
Clearly, this chain of identities can be fulfilled only if $\varsigma = 2\eta_3\; \mbox{mod}\, 2\pi$. 

On the other side, let us fix $T$ and $S$ as in \eqref{eq:TO2}. Using the parametrization
\bdm
	T = \begin{pmatrix} \cos \varsigma & \sin \varsigma \\ \sin\varsigma & -\cos\varsigma \end{pmatrix}, \qquad \mbox{with\; $\varsigma \in [0,2\pi)$}\,,
\edm
by arguments similar to those outlined above, from \eqref{eq:changeSdirac} we deduce
\begin{align*}
-\,\sigma_{j} + \tfrac{2\eta_j \eta_\ell}{\eta^2}\,\sigma_{\ell}
= e^{ - i (\eta_1 \sigma_1 + \eta_2 \sigma_2)} \, \sigma_j\, e^{ i (\eta_1 \sigma_1 + \eta_2 \sigma_2)}
= T^{t}_{j \ell}\, \sigma_{\ell}
= (\sin\varsigma)\,\sigma_{j} + (\cos\varsigma)\; \varepsilon_{3 j \ell}\, \sigma_{\ell} \,.
\end{align*}
Keeping in mind that $\sigma_1$ is real and $\sigma_2$ is imaginary, it can be checked by direct inspection that there is no admissible choice of $\varsigma$ and $\eta_1,\eta_2$ compatible with \eqref{eq:TO2} fulfilling the latter chain of identities.

Summing up, the only linear transformations $\UU$ of the form \eqref{eq:USTx0} fulfilling \eqref{eq:VHDAV}, for some suitable $\widetilde{\mathbf{A}}$, are given by
\be
T = \begin{pmatrix} \cos (2\eta_3) & - \sin (2\eta_3) \\ \sin(2\eta_3) & \cos(2\eta_3) \end{pmatrix}, \qquad	 S = e^{-i \eta_0 \one - i \eta_3 \sigma_3}, \label{eq:TSO2Dirac}
\ee
for some $ \eta_0 \in C^{\infty}(\mathbb{R}^2), \eta_3 \in \mathbb{R} $.
Notice that these maps describe simultaneous rotations of the space coordinates in the plane $\mathbb{R}^2$ and of the spin degree of freedom, together with the usual $\mathrm{U}(1)$ electromagnetic local gauge transformation.

\subsubsection{Anti-linear transformation}
We first notice that, for any map $\VV$ of the form \eqref{eq:VSTx0}, there holds
\bdm
	\VV H_{\mathrm{D}}(\mathbf{A})\, \VV^{-1} 
	= -\,S(\xv)\, \sigma_j^*\, {S}^{-1}(\xv)\, T_{hj} \big(- i \partial_{h} - i S(\xv) \partial_h {S}^{-1}(\xv) - T^{-1}_{k h}A_{k}(T^{-1}\xv - \xv_0)\big)\,.
\edm
This makes evident that $\VV H_{\mathrm{D}}(\mathbf{A})\, \VV^{-1}$ is itself a Dirac Hamiltonian of the form \eqref{eq:Dirac2exp} ({\sl cf.} \eqref{eq:VHDAV}) if and only if, for some suitable $\widetilde{\mathbf{A}} = (\widetilde{A}_1,\widetilde{A}_2)$, \eqref{eq:changeAanti} is verified and
\be
S(\xv)\, \sigma_j^*\, {S}^{-1}(\xv) \,T_{h j} = -\, \sigma_{h} \,. \label{eq:changeSantidirac}
\ee
Notice that \eqref{eq:changeSanti} is indeed a consequence of \eqref{eq:changeSantidirac}. To proceed, let us point out that \eqref{eq:changeSantidirac} can be equivalently rephrased as
\be
{S}^{-1}(\xv)\,\sigma_{h}\, S(\xv) = -\,T_{h j}\,\sigma_j^* \,, \label{eq:changeSantidirac2}
\ee
and notice that ${S}^{-1}(\xv) = e^{i ( \eta_0 \one + \bm{\eta} \cdot \bm{\sigma} )}$ for all $S(\xv) = e^{-i ( \eta_0 \one + \bm{\eta} \cdot \bm{\sigma} )}$. Moreover, recall once more that $\sigma_1$ is real and $\sigma_2$ is imaginary.

On one hand, by arguments similar to those described in subsection \cref{subsubUDir} we deduce that the only transformations $\VV$ of the form \eqref{eq:VSTx0} with $T$ and $S$ as in \eqref{eq:TSO2anti}, fulfilling \eqref{eq:VHDAV} for some $\widetilde{\mathbf{A}}$, are
\be
T = \!\begin{pmatrix} \tfrac{\eta_2^2 - \eta_1^2}{\eta_1^2 + \eta_2^2} & \tfrac{2\eta_1 \eta_2}{\eta_1^2 + \eta_2^2} \vspace{0.08cm}\\ -\tfrac{2\eta_1 \eta_2}{\eta_1^2 + \eta_2^2} & \tfrac{\eta_2^2 - \eta_1^2}{\eta_1^2 + \eta_2^2} \end{pmatrix}, \qquad  S = e^{-i \eta_0 \one - i \eta_1 \sigma_1 - i \eta_2 \sigma_2}, 
\ee
for some $ \eta_0 \in C^{\infty}(\mathbb{R}^2)  $ and $ \eta_1,\eta_2 \in \mathbb{R} $, such that \eqref{condition} holds.
On the other hand, the only maps $\VV$ of the form \eqref{eq:VSTx0} with $T$ and $S$ as in \eqref{eq:TO2anti}, fulfilling \eqref{eq:VHDAV} for some $\widetilde{\mathbf{A}}$, are given by
\beq
T = \!\begin{pmatrix} \cos({\pi \over 4} \!+\! k \pi)\!\! & \sin({\pi \over 4} \!+\! k \pi) \vspace{0.08cm}\\ \sin({\pi \over 4} \!+\! k \pi)\!\! & -\cos({\pi \over 4} \!+\! k \pi) \end{pmatrix}, \qquad S = e^{-i \eta_0 \one - i\,({5\pi \over 8} + {\pi k \over 2})\, \sigma_3}, 
\eeq
for some $ \eta_0 \in C^{\infty}(\mathbb{R}^2), k \!\in\! \mathbb{Z} $.

Summarizing and changing slightly the parametrization, we infer that the only anti-linear transformations $\VV$ of the form \eqref{eq:VSTx0} fulfilling \eqref{eq:VHDAV}, for some suitable $\widetilde{\mathbf{A}}$, are given by
\beq
T = \!\begin{pmatrix} \cos(2\eta)  & -\sin(2\eta) \\ \sin(2\eta) & \cos(2\eta) \end{pmatrix}, \qquad S = e^{-i \eta_0 \one} (\sigma_1\sin\eta - \sigma_2\cos\eta)\,,
\eeq
or
\beq
T = \!\begin{pmatrix} 1/\sqrt{2} & 1/\sqrt{2} \\ 1/\sqrt{2} & -1/\sqrt{2} \end{pmatrix}, \qquad S = \pm\, e^{-i \eta_0 \one - i\,{5\pi \over 8}\, \sigma_3}, 
\eeq
or
\beq
T = \!\begin{pmatrix} -1/\sqrt{2} & -1/\sqrt{2} \\ -1/\sqrt{2} & 1/\sqrt{2} \end{pmatrix}, \qquad	 S  = \pm\, e^{-i \eta_0 \one - i\,{\pi \over 8}\, \sigma_3}, 
\eeq
with $\eta_0 \in C^{\infty}(\mathbb{R}^2) $ and $  \eta \in [0,2\pi) $.



\begin{thebibliography}{mich9999}

\bibitem[AT98]{AT98} R. Adami, A. Teta, On the Aharonov-Bohm Hamiltonian, {\it Lett. Math. Phys.} {\bf 43}, 43--54 (1998).

\bibitem[Ag75]{Ag75} S. Agmon, Spectral properties of Schr\"odinger operators and scattering theory, {\it Ann. Sc. Norm. Super. Pisa -- Cl. Sci.} {\bf 11}, 151--218 (1975).

\bibitem[AB59]{AB59} Y. Aharonov, D. Bohm, Significance of electromagnetic potentials in the quantum theory, {\it Phys. Rev.} \textbf{115}, 485--491 (1959).


\bibitem[AT11]{AT11} I. Alexandrova, H. Tamura, Resonance free regions in magnetic scattering by two solenoidal fields at large separation, {\it J. Funct. Anal.} {\bf 260}, 1836--1885 (2011).

\bibitem[AHK23]{AHK23} S. Avramska-Lukarska, D. Hundertmark, H. Kovarik, Absence of positive eigenvalues of magnetic Schrödinger operators, {\it Calc. Var. Partial Differential Equations} {\bf 62}, 63 (2023).

\bibitem[BTRS21]{BTRS21} J.-M. Barbaroux, L. Le Treust, N. Raymond, E. Stockmeyer, On the semiclassical spectrum of the Dirichlet-Pauli operator, {\it J. Eur. Math. Soc.} {\bf 23}, 3279--3321 (2021).

\bibitem[BTRS23]{BTRS23} J.-M. Barbaroux, L. Le Treust, N. Raymond, E. Stockmeyer, The Dirac bag model in strong magnetic fields, {\it Pure App. Anal.} {\bf 5}, 643--727 (2023).

\bibitem[BW83]{BW83} H. Baumg\"artel, M. Wollenberg, {\it Mathematical Scattering Theory}, Akademie-Verlag, Berlin (1983).

\bibitem[BK23]{BK23} J. Breuer, H. Kovarik, Resonances at the Threshold for Pauli Operators in Dimension Two, preprint {\it arXiv:2304.06289 [math-ph]},  to appear on {\it Ann. H. Poincar\'e} (2023).

\bibitem[CDYZ]{CDYZ}F. Cacciafesta, P. D'Ancona, Z. Yin, J. Zhang, Dispersive estimates for Dirac equations in Aharonov-Bohm magnetic fields: massless case, {\it eprint arXiv:2407.12369} (2024).

\bibitem[CFP18]{CFP18} C. Cacciapuoti, D. Fermi, A. Posilicano, On inverses of Krein's Q-functions, {\it Rend. Mat. Appl.} {\bf 39}, 229--240 (2018).

\bibitem[CF20]{CF20} M. Correggi, D. Fermi, Magnetic perturbations of anyonic and Aharonov–Bohm Schrödinger operators, {\it J. Math. Phys.} {\bf 62}, 032101 (2021).

\bibitem[CF23]{CF23} M. Correggi, D. Fermi, Schr\"odinger operators with multiple Aharonov-Bohm fluxes, {\it Ann. Henri Poincaré} (2024).

\bibitem[CF24]{CF24} M. Correggi, D. Fermi, Deficiency indices for singular magnetic Schr\"odinger operators, {\it Milan J. Math.} (2024). 

\bibitem[CO18]{CO18} M. Correggi, L. Oddis, Hamiltonians for two-anyons systems, {\it Rend. Mat. Appl.} {\bf 39}, 277--292 (2018).

\bibitem[CFK20]{CFK20} L. Cossetti, L. Fanelli, D. Krejcirik, Absence of eigenvalues of Dirac and Pauli Hamiltonians via the method of multipliers, {\it Commun. Math. Phys.} {\bf 379}, 633--691 (2020).

\bibitem[DS98]{DS98} L. Dabrowski, P. \u St'\!ov\'i\v cek, Aharonov-Bohm effect with $\delta$-type interaction, {\it J. Math. Phys.} {\bf 39}, 47--62 (1998).

\bibitem[DF23]{DF23} J. Derezi\'nski, J. Faupin, Perturbed Bessel operators. Boundary conditions and closed realizations, {\it J. Funct. Anal.} {\bf 284}, 109728 (2023).

\bibitem[DG21]{DG21} J. Derezi\'nski, V. Georgescu, On the Domains of Bessel Operators, {\it Ann. H. Poincar\'e} {\bf 22}, 3291--3309 (2021).

\bibitem[DR17]{DR17} J. Derezi\'{n}ski, S. Richard, On Schr\"odinger operators with inverse square potentials on the half-line, {\it Ann. Henri Poincar\'{e}} {\bf 18}, 869--928 (2017).

\bibitem[EV02]{EV02} L. Ers\"{o}s, V. Vougalter, Pauli operator and Aharonov-Casher Theorem for measure valued magnetic fields, {\it Commun. Math. Phys.} {\bf 225}, 399--421 (2002).

\bibitem[F23]{F23} D. Fermi, Quadratic forms for Aharonov-Bohm Hamiltonians, in \textsl{Quantum Mathematics I}, M. Correggi, M. Falconi (eds.), Springer INdAM Series {\bf 57}, Springer (2023).

\bibitem[Fi23]{Fi23} M. Fialov\' a, Aharonov-Casher theorems for manifolds with boundary and APS boundary condition, {\it Ann. Henri Poincaré} (2024).

\bibitem[FS92]{FS92} J, Fr\"{o}lich, U.M. Studer, $U(1)\times SU(2)$-Gauge invariance of non-relativistic quantum mechanics and generalized hall effect, {\it Commun. Math. Phys.} {\bf 148}, 553--600 (1992).

\bibitem[FS93]{FS93} J. Fr\"{o}lich, U.M. Studer, Gauge invariance and current algebra in nonrelativistic many-body theory, {\it Rev. Mod. Phys.} {\bf 65}, 733 (1993).

\bibitem[GR07]{GR07} I.S. Gradshteyn, I.M. Ryzhik, {\sl Table of Integrals, Series, and Products}, Academic Press, Elsevier, Amsterdam (2007).

\bibitem[GHT23]{GHT23} F. Gesztesy, M. Hunziker, G. Teschl, Essential self-adjointness of even-order, strongly singular, homogeneous half-line differential operators, {\it Ann. Henri Poincaré} (2024).

\bibitem[GPS22]{GPS22} F. Gesztesy, M.M.H. Pang, J. Stanfill, On domain properties of Bessel-type operators, {\it Discrete Contin. Dyn. Syst. - Ser. S} (2022).

\bibitem[GS04]{GS04} V.A. Geyler, P. Stovicek, On the Pauli operator for the Aharonov–Bohm effect with two solenoids, {\it J. Math. Phys.} {\bf 45}, 51--75 (2004).

\bibitem[HH23]{HH23} D. Hundertmark, H. Kovarik, Absence of embedded eigenvalues of Pauli and Dirac operators, {\it J. Funct. Anal.} {\bf 286}(4), 110288 (2024).

\bibitem[IT01]{IT01} H.T. Ito, H. Tamura, Aharonov–Bohm effect in scattering by point-like magnetic fields at large separation, {\it Ann. H. Poincar\'e} {\bf 2}, 309--359 (2001).

\bibitem[Ko22]{Ko22} H. Kovarik, Spectral properties and time decay of the wave functions of Pauli and Dirac operators in dimension two, {\it Adv. Math.} {\bf 398} (2022).

\bibitem[LB22]{LB22} E. Lavigne Bon, Semiclassical spectrum of the Dirichlet–Pauli operator on an annulus, {\it J. Math. Phys.} {\bf 63}, 052102 (2022).

\bibitem[MPS18]{MPS18} A. Mantile, A. Posilicano, M. Sini, Limiting absorption principle, generalized eigenfunctions and scattering matrix for Laplace operators with boundary conditions on hypersurfaces, {\it J. Spectr. Theory} {\bf 8}, 1443--1486 (2018).

\bibitem[LMS09]{LMS09} M. Loss, T. Miyao, H. Spohn, Kramers Degeneracy Theorem in Nonrelativistic QED, {\it Lett. Math. Phys.} {\bf 89}, 21--31 (2009).

\bibitem[OLBC10]{OLBC10} F.W.J. Olver, D.W. Lozier, R.F. Boisvert, C.W. Clark, {\it NIST Handbook of Mathematical Functions}, Cambridge University Press, Cambridge (2010).

\bibitem[PR11]{PR11} K. Pankrashkin, S. Richard, Spectral and scattering theory for the Aharonov-Bohm operators, {\it Rev. Math. Phys.} \textbf{23}, 53--81 (2011).

\bibitem[Pe05]{Pe05} M. Persson, On the Aharonov-Casher formula for different self-adjoint extensions of the Pauli operator with singular magnetic field, {\it Electron. J. Differential Equations} {\bf 55}, 1--16 (2005).

\bibitem[Pe06]{Pe06} M. Persson, On the Dirac and Pauli Operators with Several Aharonov–Bohm Solenoids. {\it Lett. Math. Phys.} {\bf 78}, 139--156 (2006).

\bibitem[Po01]{Po01} A. Posilicano, A Krein-like Formula for Singular Perturbations of Self-Adjoint Operators and Applications, {\it J. Funct. Anal.} {\bf 183}, 109--147 (2001).

\bibitem[Po03]{Po03} A. Posilicano, Self-adjoint extensions by additive perturbations, {\it Ann. Sc. Norm. Super. Pisa -- Cl. Sci.} {\bf 2}, 1--20 (2003).

\bibitem[Po04]{Po04} A. Posilicano, Boundary triples and Weyl functions for singular perturbations of self-adjoint operators, {\it Methods Funct. Anal. Topology} {\bf 10}, 57--63 (2004).

\bibitem[Po08]{Po08} A. Posilicano, Self-adjoint extensions of restrictions, {\it Oper. Matrices} {\bf 2}, 483--506 (2008).

\bibitem[RS81]{RS81} M. Reed, B. Simon, {\it Methods of modern mathematical physics. Vol. I-IV}, Academic Press (1981).

\bibitem[Ru83]{Ru83} S.N.M. Ruijsenaars, The Aharonov-Bohm Effect and scattering theory, {\it Ann. Phys.} {\bf 146}, 1--34 (1983).

\bibitem[Ta03]{Ta03} H. Tamura, Resolvent convergence in norm for Dirac operator with Aharonov–Bohm field, {\it J. Math. Phys.} {\bf 53}, 2967--2993 (2003).

\bibitem[Ta11]{Ta11} H. Tamura, Asymptotic analysis for Green functions of Aharonov-Bohm Hamiltonian with application to resonance widths in magnetic scattering, {\it Math. J. Okayama Univ.} {\bf 53}, 1--37 (2011).

\bibitem[Te90]{Te90} A. Teta, Quadratic forms for singular perturbations of the Laplacian, {\it Publ. Res. Inst. Math. Sci.} \textbf{26}, 803--817 (1990).

\end{thebibliography}
\end{document}